\documentclass[twoside,letterpaper,10pt]{article}

\usepackage{fancyhdr}

\usepackage{amsmath, amstext, amsthm, amssymb}

\usepackage{newtxtext}
\usepackage{newtxmath}
\usepackage[scaled=1]{newtxtt}

\usepackage{graphicx}
\usepackage{float}
\usepackage{cite}
\usepackage{verbatim}

\usepackage[shortlabels,inline]{enumitem}
\usepackage[bookmarks=true]{hyperref}
\usepackage{epstopdf}
\usepackage[all]{hypcap}
\usepackage{booktabs}
\usepackage[algoruled, linesnumbered, vlined]{algorithm2e}

\usepackage{subdepth}

\usepackage{tikz}
\usepackage{tikz-qtree}
\usepackage[all]{xy}

\theoremstyle{definition}
\newtheorem{defi}{Definition}[section]

\newtheorem{rem}[defi]{Remark}
\newtheorem{exa}[defi]{Example}

\theoremstyle{plain}
\newtheorem{thm}[defi]{Theorem}
\newtheorem{lem}[defi]{Lemma}
\newtheorem{prop}[defi]{Proposition}
\newtheorem{cor}[defi]{Corollary}

\hypersetup{
	breaklinks=true,
	bookmarksdepth=3,
	hidelinks
}

\SetStartEndCondition{ (}{) }{}
\DontPrintSemicolon
\SetKwBlock{Def}{Definition}{end Definition}
\SetAlFnt{\small}
\SetCommentSty{small}

\newcommand{\lra}{\longrightarrow}
\newcommand{\Ra}{\Rightarrow}
\newcommand{\ra}{\rightarrow}

\newcommand{\mc}{\mathcal}

\newcommand{\wt}{\widetilde}

\newcommand{\rat}{\rightarrowtail}
\newcommand{\vtr}{\vartriangleright}

\newcommand{\vtrd}{\mathrel{\mbox{$\vtr\!\!\!\!\!\cdot\,\,$}}}

\newcommand{\modelsm}{\models^{\mc{F}}_\rho}

\newcommand{\intrp}[1]{\mbox{$[\![#1]\!]^{\mc{F}}_\rho$}}
\newcommand{\pos}[1]{\mbox{$\langle #1\rangle$}}
\newcommand{\onto}[1]{\stackrel{#1}{\rightsquigarrow}}
\newcommand{\lonto}[1]{\stackrel{#1}{-\!\!\!\rightsquigarrow}}
\newcommand{\leads}[1]{\stackrel{#1}{\,\,\shortparallel\!\!\!\!\leadsto}}

\newcommand{\siz}[1]{\mbox{$|\!|#1|\!|$}}

\newcommand{\tref}[1]{{\hypersetup{linkcolor=black}\textttz{\ref{#1}}}}

\newcommand{\cp}{\mbox{$\mc{C}$} }
\newcommand{\cpi}{\cp_\imath}
\newcommand{\vmodels}{\mbox{$\;|\!\!\!\approx\;$}}

\newcommand{\pdl}{\texttt{PDL}\xspace}
\newcommand{\cpdl}{\texttt{CPDL}\xspace}
\newcommand{\tpdl}{$\tau$\pdl}

\newcommand{\atf}{\texttt{AtF}\xspace}
\newcommand{\atp}{\texttt{AtP}\xspace}
\newcommand{\ev}{\texttt{Ev}\xspace}
\newcommand{\lab}{\texttt{l}\xspace}
\newcommand{\actf}{\texttt{actF}}
\newcommand{\rdf}{\texttt{rdF}}
\newcommand{\sts}{\texttt{sts}}
\newcommand{\rd}{\texttt{rd}}
\newcommand{\true}{\texttt{tt}}
\newcommand{\false}{\texttt{ff}}

\newcommand{\reach}{\texttt{reach}}

\newcommand{\unsat}{\texttt{UnSat}\xspace}
\newcommand{\sat}{\texttt{Sat}\xspace}
\newcommand{\tempsat}{\texttt{TempSat}\xspace}

\newcommand{\fp}{\mc{F\!P}}

\newcommand{\fd}{\mc{F\!D}}
\newcommand{\stp}{\mc{SP}}

\newcommand{\tp}{\mc{T\!P}}
\newcommand{\fpath}{f\!path}
\newcommand{\cpr}{C\!R}

\title{\textsc{ExpTime} Tableaux for Type PDL}
\author{
	Agathoklis Kritsimallis\\ 
	\texttt{\small agiskr@gmail.com, agiskr@uom.edu.gr}
	\and
	Ioannis Refanidis\\
	\texttt{\small yrefanid@uom.edu.gr}
}
\date{\small Department of Applied Informatics, University of Macedonia, Thessaloniki, Greece.}

\begin{document}

\maketitle
\thispagestyle{empty}

\begin{abstract}
	The system of Type \texttt{PDL} ($\tau$\texttt{PDL}) is an extension of Propositional Dynamic Logic (\texttt{PDL}) and its main goal is to provide a formal basis for reasoning about types of actions (modeled by their preconditions and effects) and agent capabilities.
	The system has two equivalent interpretations, namely the standard relational semantics and the type semantics, where process terms are interpreted as types, i.e. sets of binary relations.
	Its satisfiability problem is decidable, as a \textsc{NExpTime} decision procedure was provided based on a filtration argument and it was suggested that the satisfiability problem for $\tau$\texttt{PDL} should be solvable in deterministic, single exponential time.
	In this paper, we address the problem of the complexity of the satisfiability problem of $\tau$\texttt{PDL}.
	We present a deterministic tableau-based satisfiability algorithm and prove that it is sound and complete and that it runs in \textsc{ExpTime}. 
	Additionally, the algorithm detects satisfiability as earlier as possible, by restricting or-branching whenever possible.
\end{abstract}

\section{Introduction}

The satisfiability problem of Propositional Dynamic Logic (\pdl) \cite{FischerLadner,Harel}) is \textsc{ExpTime}-{complete}. Various decision procedures have been given such as best-case exponential \cite{Pratt1}, working in multiple stages \cite{Pratt2,Nguyen-CPDL}, on-the-fly \cite{Giacomo,Abate} and with analytic cut-rule for the converse \cite{Giacomo,Nguyen-CPDL}, most of them based on and-or tableaux \cite{Gore-talk}. The algorithm of \cite{Gore-PDL} for \pdl and its extension for the converse (\cpdl) \cite{Gore-CPDL} reveal the crucial role that global caching plays to the achievement of an optimal algorithm.

The system of \tpdl (Type \pdl) \cite{Hartonas1,Hartonas2} is an extension of \pdl. The main goal of such a system is to provide a formal basis for reasoning about types of actions instead of merely actions. 
As described in \cite{Hartonas1}, the system is inspired by the \texttt{OWL-S} process model of services \cite{owl-s}, where concrete actions are hidden from public view and what is publicly known is only the type of actions a service can perform, modeled by their preconditions and effects.
Syntactically, the language of \tpdl extends that of \pdl with abstract process types designated by their preconditions and effects, written in the form $\varphi\Ra\psi$, with agent capabilities statements $\cpi A$ which declare that agent $\imath$ can do $A$ and with a backwards possibility operator, a weak form of converse.

The interpretation of the introduced system of \tpdl in \cite{Hartonas1} is different from the standard \pdl semantics.
As the main motivation of \tpdl is reasoning about types of processes, type semantics has been proposed where process terms are interpreted as types, i.e. sets of binary relations. On the other hand, in \cite{Hartonas2}, it was shown that the standard relational semantics is equivalent to the type semantics. 

As far as the capabilities operator is concerned, the semantic options made in \cite{Hartonas1,Hartonas2} are slightly different from those made in the KARO Framework \cite{Hoek1, Hoek2, Hoek3} and it is that difference in semantics, as detailed in \cite{Hartonas1,Hartonas2} that allows for a finitary axiomatization and a decidability proof, provided in \cite{Hartonas1,Hartonas2} and missing in \cite{Hoek1, Hoek2, Hoek3}. 

Also, it was shown that the satisfiability problem is decidable. A \textsc{NExpTime} decision procedure was provided based on a filtration argument and it was suggested that the satisfiability problem for \tpdl should be solvable in deterministic, single exponential time.

Here, we address the problem of the complexity of the satisfiability problem of \tpdl, by devising a suitable tableau-based algorithm.
We prove that it is sound and complete and that it runs in \textsc{ExpTime}.

The first approach in this direction was presented in \cite{bci_tpdl} (joint work of the first author and professor Chrysafis Hartonas).
In that work, the capabilities statements are not handled properly and the presented algorithm is a direct extension of the one given in \cite{Gore-CPDL} for \cpdl.
Here, based on previously mentioned approaches for \pdl and \cpdl, as those in \cite{Nguyen-CPDL} and \cite{Gore-PDL}, we propose a new approach.

We choose here not to deal with the backward possibility operator. 
First, as argued in \cite{Hartonas1}, the specific construct is a weak form of the well-known converse operator which has been successfully handled for \cpdl in \cite{Gore-CPDL,Widmann_Thesis}. 
Second, the converse, as well as the backwards possibility operator, has a significant impact on the design of a decision procedure, as it is revealed by \cite{Gore-PDL} and \cite{Gore-CPDL} for \pdl and \cpdl, respectively.
The converse requires to maintain the necessary data structures in order to examine `backwards' whether ancestor nodes have the necessary formulas and accordingly, to `restart' special nodes. Moreover, its addition has as a result to reuse only the states (global state caching) and not all the nodes of a tableau.
Thus, at this point, we choose to abstract from the details of such an operator and study the impact of the extension of \pdl with the precondition-effect construct and with the capabilities statements. 

Our approach handles the precondition-effect processes and the capabilities statements through the necessary tableau rules. 
The precondition-effect construct allows us to express the universal definable relation through the process term $\Omega = \true\Ra\true$ (where $\true$ is any tautology), as in all the states of any model the preconditions $\true$ hold and so do the effects.
This proves to be quite convenient in order to handle the formulas in which processes are abstractly determined by their preconditions and effects, as they can be expressed through $\Omega$.
On the other hand, the interpretation of the capabilities statements of atomic processes or of processes designated by their preconditions-effects is based on the interpretation of their own processes and thus, special attention is required for them. For example, consider processes $A_1$ and $A_2$ of the previous form such that the interpretation of $A_1$ is a subrelation of the interpretation of $A_2$. Now, if an agent has the capability to execute $A_2$ at a specific state, then this immediately implies that the same agent has also the capability to execute $A_1$ at the same state.

Besides the new operators, since \tpdl is an extension of \pdl, we also have to deal with the well-known problems of \pdl, such as efficiency of the algorithm (global caching) and fulfillment of eventualities. 
The proposed algorithm has available all the nodes of a tableau for possible reuse and thus, we obtain an exponential algorithm.
Additionally, it maintains and examines the appropriate data structure in order to detect unfulfilled eventualities. 

Also, the algorithm has been designed to avoid or-branching whenever possible. 
Previous approaches for \pdl \cite{Nguyen-CPDL,Gore-PDL} fully explore or-nodes, even those which do not work in multiple stages. Intuitively speaking, since loops (i.e. graph cycles) occur due to the iteration operator and due to caching, the algorithm cannot determine whether an or-branch is definitely satisfiable or not, until the satisfiability status of an ancestor node is calculated. One should also keep in mind that loops are formed even when eventualities are not involved and thus, the data strucures which are used for the detection of unfulfilled eventualities cannot help. 
Since the satisfiability status of a node might be dependent on ancestor nodes, what we have done here is to record these dependencies and as a result, the algorithm has the ability to distinguish the independent nodes (whose satisfiability status is not going to change). 
Of course, this also applies for the children of or-nodes and thus, whenever loops have not been formed and the appropriate status indicates satisfiability, the algorithm avoids exploring the remaining or-branches.

The impact of the previous characteristic of the proposed algorithm becomes more apparent in the case of satisfiable formulas. In the case of unsatisfiable formulas, the algorithm is forced to explore all the possible cases exhaustively.
In \cite{Hustadt}, as well as in \cite{Gore-PDL}, experimental evaluations are provided which reveal the importance of caching, but also the importance of the earliest possible detection of satisfiability, so that the search space can be restricted as much as possible. 
In \cite{Widmann_Thesis}, for the case of \pdl and \cpdl, a proposed optimization is the restriction of or-branching only in the absence of diamond formulas with atomic processes, as loops are avoided.

The rest of the paper is organized as follows.
In Section \ref{sec: TPDL}, we present the syntax of \tpdl and its standard relational semantics. 
In Subsection \ref{subsec: preliminaries}, we discuss properties of the new constructs, introduce the Smullyan's unifying notation for \tpdl, namely the $\alpha$/$\beta$ formulas, we define the eventualities and the reduction sets of an $\alpha$/$\beta$ formula and show syntactic and semantic properties of interest related to them.
In Section \ref{sec: Hintikka structures}, we introduce the Hintikka structures, namely transition systems with specific properties and with states which are labelled with formulas, and show that they imply the satisfiability of their formulas through the appropriate semantic model. 
In Section \ref{sec: algorithm}, we present our tableau-based satisfiability algorithm. First, in Subsection \ref{subsec: tableau calculus}, we give the necessary tableau calculus, as its rules are used to expand the nodes of a tableau, and show its soundness. In Subsection \ref{subsec: algorithm}, we first present the necessary definitions and then the procedures of the algorithm and in Subsection \ref{subsec: examples}, we give examples of tableaux.
In Section \ref{sec: correctness and complexity}, we prove the correctness of the algorithm and argue about its complexity. First, in Subsection \ref{subsec: tableau properties}, we present various properties that concern a tableau which are used in the subsequent subsections. Then, in Subsection \ref{subsec: soundness}, we show that the algorithm is sound (by defining the appropriate Hintikka structure), in Subsection \ref{subsec: completeness}, we show that it is complete and in Subsection \ref{subsec: complexity}, we show that it is in \textsc{ExpTime}.
Finally, in Section \ref{sec: conclusions}, we conclude our work and give future directions.

\section{Type \pdl and Preliminaries}	\label{sec: TPDL}

In this section, we review the syntax and the semantics of \tpdl and give various definitions and properties that will be useful in the remainder of the paper.

\subsection{Syntax and Intuitive Semantics}

The \tpdl language that we consider in this paper extends the standard \pdl language by capabilities statements and terms designating processes by their preconditions and effects, written $\varphi\Ra\psi$.
Holding a capability statement $\cpi A$ at some state indicates that the agent $\imath$ has the capability to execute the action $A$ at that state.

Assume countable, non-empty and disjoint sets of atomic formulas $\atf$ and atomic program terms $\atp$, and let $I$ be a set of agent (or service) names, disjoint from $\atf$ and $\atp$.
The formal syntax of \tpdl follows:
\begin{equation*}
\begin{split}
\mc{L}_s\ni \varphi    \;&:=\;    p   \mid    \neg\varphi   	\mid   \forall A.\varphi    \mid		\cpi A										\\
\mc{L}_a\ni 	A	   \;&:=\;    a   \mid     \varphi          \mid 	\varphi\Ra\varphi 	\mid   	      AA        \mid       A+A      \mid   A^*
\end{split}
\end{equation*}
where $p\in\atf$, $a\in\atp$ and $\imath\in I$. In what follows, $p,q$ range over atomic formulas and $a,b,c$ over atomic programs.
We adopt a slightly different notation from the one usually met in the literature and we write $AA$ for $A;A$, $A+A$ for $A\cup A$, and $\varphi$ for $\varphi?$, when $\varphi$ acts as a program. 
Moreover, we write $\forall A.\varphi$ for the familiar boxed formula $[A]\varphi$ in order to stay consistent with the notation in \cite{Hartonas1}, where it received a non-standard, type interpretation.
The possibility operator and the propositional connectives are not taken as primitive, but they can be defined as usual (e.g. $\exists A.\varphi \equiv \neg\forall A.\neg\varphi$ and $\varphi\ra\psi \equiv \forall \varphi.\psi$). 

We let $\Sigma$ be generated by the grammar $a\; (a\in\atp) \mid \varphi \mid \varphi\Ra \varphi\;(\varphi\in\mc{L}_s)$ and $\wt{\Sigma}$ be its restriction generated by the grammar $a\; (a\in\atp) \mid \varphi\Ra \varphi\;(\varphi\in\mc{L}_s)$. The language of the action (program) terms $\mc{L}_a$ is then the language of regular expressions on the alphabet $\Sigma$.

We trust that the reader is familiar with the standard relational interpretation of the language of \pdl.
In \cite{Hartonas1} and \cite{Hartonas2}, \tpdl has received two kinds of interpretations: a type interpretation and a standard relational one. In  \cite{Hartonas2}, the two interpretations are shown to be equivalent, meaning that exactly the same formulas are validated in the two classes of models.
In \cite{Hartonas1}, a set $\atp$ of basic types, test types and precondition-effect types are considered  and then complex process types are built by using the regular operators of composition, choice and Kleene star.
The precondition-effect program terms $\varphi\Ra\psi$ can be viewed as anonymous processes/actions or as a type of processes/actions. They are specified by the preconditions which must hold in order for the action to be executed and by the effects that the execution of the action has as a result. 
The sentential fragment of the \tpdl language is an extension of the \pdl language by the introduction of capabilities statements.

\subsection{Standard Relational Semantics}

In this paper, we follow the standard relational semantics of \tpdl, as presented in \cite{Hartonas2}.

\begin{defi}
A \textit{frame} $\mc{F}$ is a structure $\langle W, P, \onto{}, I \rangle$, also denoted as $\langle W, (\onto{\pi})_{\pi\in P}, I \rangle$, where

\begin{itemize}
\item $(W,P,\onto{})$ is  a labeled transition system, with $W$ a nonempty set of states of the system.
$P$ is a set of labels and
$\onto{}$ is a transition function assigning to each label $\pi\in P$ a binary relation  $\onto{\pi}\; \subseteq W\times W$.
The map $\,\onto{}\,$ extends to all  of $P^+$ (finite, non-empty sequences of items in $P$), by composition (thus we write, for example, $\lonto{\sigma,\tau}\;$  for the composition $\onto{\sigma}\onto{\tau}$).

\item $I$ is a set of agent names.
\end{itemize}
The frame is an $\mc{L}$-\textit{frame} if the set $P$ of process labels is the set $\atp$ of atomic process terms.
\end{defi}

\begin{defi}	\label{def: l-model} 
An $\mc{L}$-\textit{model} is a triple $\mc{M}=\langle \mc{F},  \rho, (\imath^\mc{M})_{\imath\in I}\rangle$,
where
\begin{enumerate}
\item $\mc{F}$ is an $\mc{L}$-frame $\langle W, (\onto{a})_{a\in\atp}, I \rangle$,
\item $\rho: \atf \lra 2^{W}$ is an interpretation function of atomic properties of states,
\item for each $\imath\in I$, $\imath^\mc{M}$ is a map assigning to the agent $\imath$ capabilities that it has at each state $w\in W$, more
specifically $\imath^\mc{M}(w)  \subseteq  \bigcup \{ \onto{A} \mid A\in\Sigma^+ \}$, where $\onto{A}$ is a binary relation on $W$, as we define below.
\end{enumerate}

The interpretation function $\rho$ and the transition function $\onto{}$ are extended to all formulas and actions, respectively, as shown in Table \ref{tab: int}, where for $\varphi\in\mc{L}_s$, its interpretation $\intrp{\varphi} \subseteq W$ is the set of states  where $\varphi$ holds and for $A \in \mc{L}_a$, the set $\onto{A} \;\subseteq\; W\times W$  is a process, i.e. a binary relation on $W$.

\begin{table}[ht] \label{tab: int}
\caption{Interpretation in Relational Semantics}
\small
\vskip1.2mm
\centering 
\begin{tabular}{lll l lll l lll}
	\toprule
	\hspace{-0.15 cm}$\intrp{p}$ 				&=&\hspace{-0.3 cm} 	$\rho(p)$    
	\\[1mm]
	\hspace{-0.15 cm}$\intrp{\neg\varphi}$ 		&=&\hspace{-0.3 cm} 	$W\setminus\intrp{\varphi}$
	&\hspace{-0.8 cm}$\intrp{\forall A.\varphi}$ &\hspace{-0.35 cm}=&\hspace{-0.3 cm}  	$\{  w\!\in\!W  \mid  \forall w'\!\in W (w\onto{A}w' \Ra w'\!\in\intrp{\varphi})  \}$
	\\[1mm]
	\hspace{-0.15 cm}$\intrp{\cpi a}$ 					&=&\hspace{-0.3 cm}  $\{w\in W \mid \,\onto{a}\,\subseteq\imath^\mc{M}(w)\}$
	&\hspace{-0.8 cm}$\intrp{\cpi(AB)}$   	&\hspace{-0.35 cm}=&\hspace{-0.3 cm} $\{ w\!\in\!W \mid w\!\in\!\intrp{\cpi A}  \land \forall w' (w\!\onto{A}\!w' \!\Ra w'\!\!\in\! \intrp{\cpi B}) \}$
	\\[1mm]
	\hspace{-0.15 cm}$\intrp{\cpi\varphi}$				&=&\hspace{-0.3 cm}  $W$
	&\hspace{-0.8 cm}$\intrp{\cpi(A\!+\!B)}$ &\hspace{-0.35 cm}=&\hspace{-0.3 cm} $\intrp{\cpi A} \cap \intrp{\cpi B}$
	\\[1mm]
	\hspace{-0.15 cm}$\intrp{\cpi(\varphi\!\Ra\!\psi)}$\hspace{-0.35 cm}  	&=&\hspace{-0.3 cm}  $\{w\!\in\!W\mid \,\lonto{\varphi\Ra\psi}\,\subseteq\imath^\mc{M}(w)\}$
	\hspace{0.65 cm}
	&\hspace{-0.8 cm}$\intrp{\cpi A^*}$ 		&\hspace{-0.35 cm}=&\hspace{-0.3 cm} $\bigcup \left\{ \intrp{\varphi} \mid \intrp{\varphi} \subseteq \intrp{\cpi A} \cap \intrp{\forall A.\varphi} \right\}$
	\\\bottomrule
\end{tabular}

\vspace{5 mm}
\begin{tabular}{lll l lll l lll}
	\toprule
	$\onto{a}$ 				 &$\subseteq$&  $W\times W$							
		&$\lonto{AB}$ 				&=& 	$\onto{A}\,\onto{B}$
	\\[1mm]
	$\onto{\varphi}$ 		 &=& 			$\{(w,w)\mid w\!\in\!\intrp{\varphi}\!\}$	
		&$\lonto{A +  B}$ 			&=& 	$\onto{A}\;\cup\;\onto{B}$
	\\[1mm]
	$\lonto{\varphi\Ra\psi}$ &=&  			$\bigcup\!\big\{ \!\onto{A}\,  \mid  A\in \Sigma^+ \text{ and } \forall w\in\intrp{\varphi}\,\forall w'\,( w\onto{A}w' \Ra w'\in\intrp{\psi})\big\}$ 
	\hspace{1.1 cm}
		&$\lonto{A^*}$ 				&=& 	$\bigcup_{n\geq 0}\left(\onto{A}\right)^n$
	\\\bottomrule
\end{tabular}
\end{table}

Furthermore, the interpretation and the capabilities assignment function are required to satisfy the normality condition
\eqref{eq: normality}, for each $\imath\in I$ and $w\in W$.

\begin{equation}
\imath^\mc{M}(w)=
\bigcup\{\lonto{\varphi\Ra\psi} \mid w\in\intrp{\cpi(\varphi\Ra\psi)}\}    \cup   
\bigcup\{\onto{a}  \,\mid  w\in\intrp{\cpi a}\}														\label{eq: normality}
\end{equation}
\end{defi}

Assuming an $\mc{L}$-model $\mc{M} = \langle \mc{F},  \rho, (\imath^\mc{M})_{\imath\in I}\rangle$, we define the satisfaction relation $\modelsm$ (or equivalently denoted as $\models^\mc{M}$) as usual: $w \modelsm \varphi$ iff $w\in\intrp{\varphi}$.
We also write $w\models^\mc{M}\Gamma$, where $\Gamma$ is a set of formulas iff $w\models^\mc{M}\varphi$, for each $\varphi$ in $\Gamma$.
Additionally, we say that a formula is (logically) \textit{valid} iff it is true at any state of any $\mc{L}$-model.

For a discussion on semantic intuitions, the reader is referred to \cite{Hartonas1,Hartonas2}. 
Here, we draw the reader's attention only to the relation $\lonto{\varphi\Ra\psi}$,  defined in Table \ref{tab: int} as the largest definable relation connecting states $w,w'$ such that if $\varphi$ holds at $w$, then $\psi$ holds at $w'$.

In the sequel, we will be interested in and make use of two types of semantic facts, written as $\gamma\models^\mc{M}\Psi$ and $\gamma\vmodels^{\!\mc{M}} A$, where $\gamma= A_1,\dotsc,A_n$ is a finite sequence of action terms and $\Psi$ a finite set of formulas. 
The first is defined to denote that for any sequence of transitions $w_1\onto{A_1}w_2\onto{A_2}\cdots\onto{A_n}w_n$ in $\mc{M}$, the end state satisfies some $\psi\in\Psi$. 
Only in this case, $\gamma$ could be the empty sequence which merely represents the trivial test $\true$ (where $\true$ is any tautology). 
The second relation $\gamma\vmodels^{\!\mc{M}}A$ is defined to mean that the composite relation interpreting $\gamma$ is a subrelation of that interpreting $A$, i.e. $\onto{\gamma} \,\subseteq\, \onto{A}$. As usual, the meaning of $\gamma\models\Psi$ and $\gamma\vmodels A$ is obtained by quantifying universally over models.

\subsection{Preliminaries}	\label{subsec: preliminaries} 

It follows from the semantics that
\begin{equation}	\label{prop1}
\varphi\models\forall A.\psi 		\quad\text{ iff }\quad			A\vmodels\varphi\Ra\psi			\quad\text{ iff }\quad
\varphi,A\models\psi 				\quad\text{ iff }\quad			\varphi A\models\psi
\end{equation}
and, as detailed in \cite{Hartonas1}, these equivalences reveal the naturalness of precondition-effect process type operator. 

For the purposes of this paper, we point out first that the relation designated by the construct $\true\Ra\true$, which we shall hereafter designate by $\Omega$, is the universal definable relation, i.e. a universal process such that for any term $A$, we have $A\vmodels\Omega$ (meaning, as we previously explained, that $\onto{A} \,\subseteq\, \onto{\Omega} \,=\, \lonto{\true\Ra\true}$). Consequently, $\onto{\Omega}$ is transitive (i.e. $\Omega\Omega\vmodels\Omega$) and it is also clearly a reflexive relation (since $\true\vmodels\Omega$, i.e. $\onto{\true} \,\subseteq\, \onto{\Omega}$). For later use, we list this as a Lemma.

\begin{lem}[Universal Definable Relation]		\label{lem: Omega}
	Letting $\Omega=\true\Ra\true$, where $\true$ is any tautology, the following hold for $\Omega$.
	\begin{enumerate}
	\item For any process term $A$, $A\vmodels\Omega$ ($\onto{\Omega}$ is the universal definable relation)
	\item $\true\vmodels\Omega$ ($\onto{\Omega}$ is reflexive)
	\item $\Omega\Omega\vmodels\Omega$ ($\onto{\Omega}$ is transitive).
	\end{enumerate}
	Therefore, $\forall\Omega.\underline{\hspace{0.25 cm}}$ is an $S4$ modality and of course $\onto{\Omega}$ is the smallest (in the sense of set inclusion) reflexive and transitive relation containing $\onto{\Omega}$ and thus, $\onto{\Omega^*} \,=\, \onto{\Omega}$.
	\qed
\end{lem}

In what follows, we often make use of the set $\atp\cup\{\Omega\}$. Thus, for simplicity of notation, we denote it as $\atp_\Omega$.
\begin{lem}	\label{=>1}
	The following properties hold:
	\begin{enumerate}[1)]
	\item For any program term $A$, $(\neg\varphi)A\vmodels\varphi\Ra\psi$ and $A\psi\vmodels\varphi\Ra\psi$. Taking $A=\true$, we have $\neg\varphi\vmodels\varphi\Ra\psi$ and $\psi\vmodels\varphi\Ra\psi$. Also, the sentential test designated by $\varphi\ra\psi$ is a subrelation of the relation designated by $\varphi\Ra\psi$, in symbols $\varphi\ra\psi\vmodels\varphi\Ra\psi$.
	\item For any $\vartheta,\varphi,\psi$, $\;\vartheta\vmodels\varphi\Ra\psi$ iff $\vartheta\models\varphi\ra\psi$. Consequently, in any $\mc{L}$-model $\mc{M}$,
	\begin{equation}
	\lonto{\varphi\Ra\psi} \;=\;
	\lonto{\varphi\ra\psi}\;\cup\;\bigcup\{ \onto{A} \mid  A\in\widetilde{\Sigma}^+ \text{ and } \forall w,w'  ( w\onto{A}w' \text{ and } w\models^\mc{M}\varphi \,\text{ implies }\, w'\models^\mc{M}\psi )  \}
	\end{equation}
	\end{enumerate}
\end{lem}
\begin{proof}
	All claims of 1) follow from the semantic equivalences in equation \eqref{prop1}.
	For 2), the claims follow immediately from definitions, recalling that $\wt{\Sigma}$ is obtained from $\Sigma$ by omitting sentential tests ($\Sigma$ is generated by the grammar $a\; (a\in\atp) \mid \varphi \mid \varphi\Ra \varphi\;(\varphi\in\mc{L}_s)$).
\end{proof}

\begin{cor}	\label{cor: =>2}
	In any $\mc{L}$-model $\mc{M}$, the following hold for a program term $\varphi\Ra\psi$
	\begin{align}
	\lonto{\varphi\Ra\psi}\;
	&=\; \{  (w,w')  \mid  \exists A\in\Sigma^+ (w\onto{A}w' \text{ and } ( w\models^\mc{M}\neg\varphi, \,\text{ or }\, w'\models^\mc{M}\psi ) )  \}	\label{=>2 eq1} \\
	&=\;\; \lonto{\varphi\ra\psi}   \cup   \;\{ (w,w')  \mid  \exists A \in \widetilde{\Sigma}^+  (w\onto{A}w' \text{ and } ( w\models^\mc{M}\neg\varphi, \,\text{ or }\, w'\models^\mc{M}\psi ) ) \}	\label{=>2 eq2}
	\end{align}
	The relation $\lonto{\varphi\Ra\psi}$ can also be defined in terms of $\Omega$ (recall that $\Omega$ was defined as $\true\Ra\true$):
	\begin{equation}
	\lonto{\varphi\Ra\psi}  \;=\;  \lonto{\neg\varphi\Omega}  \cup  \lonto{\Omega\psi}		\label{=>2 eq3}
	\end{equation}
\end{cor}
\begin{proof}
For equations \eqref{=>2 eq1} and \eqref{=>2 eq2}, the inclusion left to right is immediate. For the converse inclusion, let $(w,w')$ be a pair of states and $A\in\Sigma^+$ be such that $w\onto{A}w'$ and either $w\models^\mc{M}\neg\varphi$, or $w'\models^\mc{M}\psi$. By Lemma \ref{=>1}, in the first case, $(\neg\varphi)A\vmodels\varphi\Ra\psi$ and in the second $A\psi\vmodels\varphi\Ra\psi$, and so $w\lonto{\varphi\Ra\psi}w'$. Similarly for the second identity, now using also the second part of Lemma \ref{=>1}.
For the last case, by using equation \eqref{=>2 eq1}, we have
\begin{align*}
\lonto{\varphi\Ra\psi}		\,&=\;		\{(w,w') \mid \exists A\in\Sigma^+ (w\onto{A}w' \text{ and } w\models^\mc{M}\neg\varphi)  \}		\,\cup\,
										\{(w,w') \mid \exists A\in\Sigma^+ (w\onto{A}w' \text{ and } w'\models^\mc{M}\psi )  \}\\
							\,&=\;		\{(w,w') \mid \exists A\in\Sigma^+ (w\onto{\neg\varphi}w\onto{A}w')  \}								\,\cup\,
										\{(w,w') \mid \exists A\in\Sigma^+ (w\onto{A}w'\onto{\psi}w')  \}\\
							\,&=\;\,	\onto{\neg\varphi}\! \{(w,w') \mid \exists A\!\in\!\Sigma^+ (w\!\onto{A}\!w')  \}					\,\cup\,
										\{(w,w') \mid \exists A\!\in\Sigma^+ (w\!\onto{A}\!w')  \}	\!\onto{\psi}\\
							\;&=\;\,\onto{\neg\varphi} \onto{\Omega}	\,\cup\,	\onto{\Omega} \onto{\psi}	
							\;\;=\;\;	\lonto{\neg\varphi\Omega}	\,\cup\,	\lonto{\Omega\psi} 								\qedhere
\end{align*}
\end{proof}

\begin{lem}		\label{lem: => equiv}
The following semantical equivalences hold
\begin{align}
\forall(\varphi\!\Ra\!\psi).\vartheta \;&\equiv\; (\varphi \land \forall\Omega.\forall\psi.\vartheta) \lor \forall\Omega.\vartheta	   \label{sequiv1}\\
\neg\forall(\varphi\!\Ra\!\psi).\vartheta	\;&\equiv\;	(\neg\varphi \land \neg\forall\Omega.\vartheta)  \lor \neg\forall\Omega.\forall\psi.\vartheta) 																																		\label{sequiv2}
\end{align}
\end{lem}
\begin{proof}
For the first equivalence, by semantics and by using equation \eqref{=>2 eq3} of Corollary \ref{cor: =>2}, we have
\begin{align*}
\forall(\varphi\!\Ra\!\psi).\vartheta	
					&\equiv \forall(\neg\varphi\Omega\!+\!\Omega\psi).\vartheta
					\equiv \forall(\neg\varphi\Omega).\vartheta  \land \forall(\Omega\psi).\vartheta
					\equiv \forall(\neg\varphi).\forall\Omega.\vartheta  \land \forall\Omega.\forall\psi.\vartheta
					\equiv (\varphi \!\lor\! \forall\Omega.\vartheta)  \land \forall\Omega.\forall\psi.\vartheta\\
					&\equiv (\varphi \land \forall\Omega.\forall\psi.\vartheta)\lor(\forall\Omega.\vartheta \land \forall\Omega. \forall\psi. \vartheta)
\end{align*}
Observing that $\forall\Omega.\vartheta \ra \forall\Omega.\forall\psi.\vartheta$ is logically valid, the second disjunct can be equivalently replaced merely by $\forall\Omega.\vartheta$ and this completes the proof of \eqref{sequiv1}.
For the second equivalence, we can either use Corollary \ref{cor: =>2} in a similar way or apply negation to both sides of equation \eqref{sequiv1}. 
\end{proof}

\begin{lem} \rm	\label{lem: cap}
	In any $\mc{L}$-model $\mc{M}$, the following semantical equivalences hold
	\begin{equation}									\label{cap equiv} 
	\cpi(AB)\equiv\cpi A\land\forall A.\cpi B	\hspace{5mm}	\cpi(A\!+\!B)\equiv \cpi A\land\cpi B	\hspace{5mm}\cpi A^*\equiv\forall A^*.\cpi A
	\end{equation}
	and if the interpretation of a process $A\in\wt{\Sigma}$ is the empty set (i.e. $\onto{A}\,=\,\emptyset$), then the capability statement $\cpi A$ is satisfied by all the states of $\mc{M}$. By contraposition, if there is a state which does not satisfy the statement $\cpi A$, then $\onto{A} \, \neq \emptyset$.
\end{lem}
\begin{proof}
	For the equivalences, we refer the reader to \cite{Hartonas1,Hartonas2}. 
	The second property follows from the semantics of the capabilities statements of the form $\cpi a$ and $\cpi(\varphi\Ra\psi)$. In any $\mc{L}$-model $\mc{M}$, if the interpretation of some process $A\in\wt{\Sigma}$ is the empty set, then for any state $w$ and any agent $\imath\in I$, we have that $\onto{A} \,\subseteq \imath^\mc{M}(w)$ which signifies that $\imath$ has the ability to execute $A$.
\end{proof}

In the following of this paper, it is convenient to make use of Smullyan's distinction to $\alpha$ and $\beta$ formulas, appropriately extending this notion to the syntax of \tpdl, as presented in Table \ref{tab: alpha/beta}. As usual, the conjunctive cases are classified as $\alpha$-formulas and the disjunctive as $\beta$-formulas.

\begin{table}[ht]\small
\caption{The $\alpha$ and $\beta$-formulas}
\label{tab: alpha/beta}
\vspace{2 mm} 
\centering
\begin{tabular}{	@{\hspace{0.1 cm}}c@{\hspace{0.2 cm}}		@{\hspace{0.15 cm}}c@{\hspace{0.15 cm}}		@{\hspace{0.15 cm}}c@{\hspace{0.15 cm}}		@{\hspace{0.15 cm}}c@{\hspace{0.15 cm}}		@{\hspace{0.15 cm}}c@{\hspace{0.15 cm}}		@{\hspace{0.15 cm}}c@{\hspace{0.15 cm}}		@{\hspace{0.15 cm}}c@{\hspace{0.15 cm}}		@{\hspace{0.15 cm}}c@{\hspace{0.15 cm}}		@{\hspace{0.15 cm}}c@{\hspace{0.15 cm}}		@{\hspace{0.15 cm}}c@{\hspace{0.15 cm}}		@{\hspace{0.15 cm}}c@{\hspace{0.1 cm}}}
\toprule 
$\alpha$	
&	$\neg\neg\varphi$					&	$\neg\forall\psi.\varphi$		&	$\forall AB.\varphi$			&	$\neg \forall AB.\varphi$			
&	$\forall(A\!+\!B).\varphi$			&	$\forall A^*.\varphi$			
&	$\cpi(AB)$							&	$\cpi(A\!+\!B)\!\!$				&	$\cpi (A^*)\!\!$				&	$\neg\cpi (A^*)\!\!$\\
\midrule 
$\alpha_1$
&	$\varphi$							&   $\neg\varphi$					&	$\forall A.\forall B.\varphi$	&	$\neg\forall A.\forall B.\varphi$
& 	$\forall A.\varphi$					&	$\varphi$						
&	$\cpi A$							&	$\cpi A$						&	$\forall A^*.\cpi A$			&	$\neg\forall A^*.\cpi A$\\
$\alpha_2$
&										&	$\psi$								&									&
&	$\forall B.\varphi$					&	$\forall A.\forall A^*\!.\varphi$
&	$\forall A.\cpi B$					&	$\cpi B$							&									&\\
\bottomrule
\end{tabular}

\vspace{3.5 mm}
\begin{tabular}{ @{\hspace{0.1 cm}}c@{\hspace{0.3 cm}}		@{\hspace{0.2 cm}}c@{\hspace{0.2 cm}}		@{\hspace{0.2 cm}}c@{\hspace{0.2 cm}}		@{\hspace{0.2 cm}}c@{\hspace{0.2 cm}}		@{\hspace{0.2 cm}}c@{\hspace{0.2 cm}}		@{\hspace{0.2 cm}}c@{\hspace{0.2 cm}}		@{\hspace{0.2 cm}}c@{\hspace{0.2 cm}}		@{\hspace{0.2 cm}}c@{\hspace{0.1 cm}} }
\toprule 
$\beta$
&	$\forall\psi.\varphi$	  							&	$\forall(\varphi\!\Ra\!\psi).\vartheta$
&	$\neg\forall(\varphi\!\Ra\!\psi).\vartheta$			&	$\neg\forall(A\!+\!B).\varphi$
&	$\neg\forall A^*.\varphi$							&	$\neg\,\cpi(AB)$
&	$\neg\,\cpi(A\!+\!B)\!$\\
\midrule 
$\beta_1$
&	$\neg\psi$													&	$\varphi \!\land\! \forall\Omega^*.\forall\psi.\vartheta$
&	$\neg\forall(\neg\varphi).\forall\Omega.\vartheta$	&	$\neg\forall A.\varphi$
&	$\neg\varphi$												&	$\neg\,\cpi A$
&	$\neg\,\cpi A$\\
$\beta_2$
&	$\varphi$											&	$\forall\Omega^*.\vartheta$
&	$\neg\forall\Omega.\forall\psi.\vartheta$			&	$\neg\forall B.\varphi$
&	$\neg\forall A.\forall A^*.\varphi$					&	$\neg\forall A.\cpi B$
&	$\neg\,\cpi B$\\
\bottomrule
\end{tabular}
\end{table}

\begin{rem}[Notational Conventions]			\label{rem: alpha-beta formulas}
	In the sequel we will be referring to rules that expand an $\alpha$-formula (resp. a $\beta$-formula) as $\alpha$-rules (resp. $\beta$-rules). We write $\alpha/\beta$ when we wish to state something about $\alpha$ or $\beta$ formulas and we often write $\alpha$ (resp. $\beta$) as a formula of the language, i.e. one of the $\alpha$-formulas (resp. $\beta$-formulas).
	Now, observe that there are cases in which $\alpha_2$ formula does not exist (see Table \ref{tab: alpha/beta}). For simplicity of notation, we treat uniformly all the $\alpha$ formulas as if $\alpha_2$ existed for all of them. So, in what follows, with some abuse of notation, we list $\alpha_2$ as an element of a set, while there are cases that we shouldn't, and we write $\alpha_1\!\land\!\alpha_2$, while there are cases that we should just write $\alpha_1$.
\end{rem}

\begin{prop}		\label{prop: equiv alpha beta} 
	The formulas $\alpha \leftrightarrow \alpha_1\land\alpha_2$ and $\beta \leftrightarrow \beta_1\lor\beta_2$ are logically valid.
\end{prop}
\begin{proof}
	The equivalences concerning the standard \pdl operators are well-known \cite{FischerLadner,Harel}. 
	The formulas that correspond to the capabilities statements are valid by Lemma \ref{lem: cap}.
	The equivalences for the $\beta$ formulas $\forall(\varphi\!\Ra\!\psi).\vartheta$ and $\neg\forall(\varphi\!\Ra\!\psi).\vartheta$ follow by Lemma \ref{lem: => equiv}. 
	Also, have in mind that in the case of $\forall(\varphi\!\Ra\!\psi).\vartheta$ we use the process term $\Omega^*$ instead of $\Omega$. By Lemma \ref{lem: Omega}, we know that this makes no difference.
\end{proof}

As in \pdl, formulas of interest are the eventualities due to the iteration operator. 
Special attention is required for this kind of formulas, as global properties of a tableau (as opposed to local properties of a node) need to be examined in order to ensure that some appropriate formula is eventually satisfied.

\begin{defi}
	An \textit{eventuality} is a formula of the form $\neg\forall A_1.\dotsb\forall A_k.\forall A^*\!.\varphi$, for some $k\geq0$. The set of all eventualities will be designated by \ev.
\end{defi}

\begin{defi}\label{def: vtr}
	The binary relation $\vtr$ relates formulas $\varphi$ and $\psi$, where $\varphi$ is an $\alpha/\beta$ formula of the form $\neg\forall A.\chi$, under the following conditions:
	\begin{enumerate*}[i)]
		\item If $\varphi$ is an $\alpha$ formula, then $\psi=\alpha_1$ and
		\item if $\varphi$ is a $\beta$ formula, then $\psi\in\{\beta_1, \beta_2\}$.
	\end{enumerate*}
\end{defi}

A tableau-based satisfiability algorithm usually uses the $\alpha$ and $\beta$ rules for the expansion of a tableau. A common issue of these rules related to eventualities is that, after a series of applications of them, an eventuality may reoccur without ever being fulfilled (see also in \cite{Abate, Gore-PDL, Gore-CPDL} the problem with ``at a world'' cycles). Consider the following example.
\begin{exa}
	Eventualities for which all the necessary $\alpha$ and $\beta$ rules were applied, but without ever being fulfilled:
	\begin{itemize}
	\item 
	$\neg\forall a^{**}.\psi		\;\vtr\;		\neg\forall a^{*}.\forall a^{**}.\psi		\;\vtr\;		\neg\forall a^{**}.\psi $
	
	\item 
	$\neg\forall(\vartheta\xi)^*.\psi  \;\vtr\;  \neg\forall\vartheta\xi.\forall(\vartheta\xi)^*.\psi   \;\vtr\;   \neg\forall\vartheta.\forall\xi.\forall(\vartheta\xi)^*.\psi   \;\vtr\;   \neg\forall\xi.\forall(\vartheta\xi)^*.\psi  \;\vtr\; \neg\forall(\vartheta\xi)^*.\psi$
	
	\item 
	$\neg\forall a^*b^*.\forall(a^*b^*)^*.\psi		\vtr		\neg\forall a^*.\forall b^*.\forall(a^*b^*)^*.\psi		\vtr		\neg\forall b^*.\forall(a^*b^*)^*.\psi    \vtr\!		\neg\forall(a^*b^*)^*.\psi		\vtr\!\allowbreak		\neg\forall a^*b^*.\forall(a^*b^*)^*\!.\psi$
	
	\item 
	$\neg\forall(a+\vartheta).\forall(a+\vartheta)^*.\psi    \;\vtr\;    \neg\forall\vartheta.\forall(a+\vartheta)^*.\psi     \;\vtr\;     \neg\forall(a+\vartheta)^*.\psi     \;\vtr\;      \neg\forall(a+\vartheta).\forall(a+\vartheta)^*.\psi$
	\end{itemize}
\end{exa}

In order to detect such cases of unfulfilled eventualities, we monitor the application of the $\alpha$ and $\beta$ rules for some eventuality in the way that the following two definitions determine. 
The same definitions have also been used for the case of hybrid \pdl (i.e. \pdl with nominals) in \cite{hpdl_ictac}.

\begin{defi}	 		\label{def: decomposition} 
The \textit{decomposition set} $\mc{D}(\varphi)$ of an $\alpha$/$\beta$ eventuality $\varphi$ is a set of triples $(\mc{P}, \mc{T}, \vartheta)$ where $\mc{P}, \mc{T}\subseteq \mc{L}_s$ are sets of formulas and $\vartheta$ is a formula.
It contains the triple $\big(\emptyset, \emptyset, \varphi\big)$ and it is closed under the following \textit{decomposition rules}, where a rule $\dfrac{(\mc{P},\mc{T},\vartheta)}{(\mc{P}',\mc{T}',\vartheta')}$ is applied iff $\vartheta$ is an eventuality and $\vartheta\notin\mc{P}$:
	\begin{gather*}
	\dfrac
	{\big(\mc{P}, \mc{T}, \neg\forall B^*.\chi\big)}
	{\big(\mc{P} \cup \{\neg\forall B^*.\chi\}, \mc{T}, \neg\chi\big)}
	\quad\quad
	\dfrac
	{\big(\mc{P}, \mc{T}, \neg\forall B^*.\chi\big)}
	{\big(\mc{P} \cup \{\neg\forall B^*.\chi\}, \mc{T}, \neg\forall B.\forall B^*.\chi\big)}
	\quad\quad
	\dfrac
	{\big(\mc{P}, \mc{T}, \neg\forall B_1B_2.\chi\big)}
	{\big(\mc{P} \!\cup\! \{\neg\forall B_1B_2.\chi\}, \mc{T}, \neg\forall B_1.\forall B_2.\chi\big)}
	\end{gather*}\vspace{-0.2 cm} 
	\begin{gather*}
	\dfrac
	{\big(\mc{P}, \mc{T}, \neg\forall(B_1\!+\!B_2).\chi\big)}
	{\big(\mc{P} \!\cup\! \{\neg\forall(B_1\!+\!B_2).\chi\}, \mc{T}, \neg\forall B_1.\chi\big)}
	\quad 
	\dfrac
	{\big(\mc{P}, \mc{T}, \neg\forall(B_1\!+\!B_2).\chi\big)}
	{\big(\mc{P} \!\cup\! \{\neg\forall(B_1\!+\!B_2).\chi\}, \mc{T}, \neg\forall B_2.\chi\big)}
	\quad 
	\dfrac
	{\big(\mc{P}, \mc{T}, \neg\forall\vartheta.\chi\big)}
	{\big(\mc{P} \!\cup\! \{\neg\forall\vartheta.\chi\}, \mc{T}\cup\{\vartheta\}, \neg\chi\big)}
	\end{gather*}\vspace{-0.2 cm} 
	\begin{gather*}
	\dfrac
	{\big(\mc{P}, \mc{T}, \neg\forall(\varphi\Ra\psi).\chi\big)}
	{\big(\mc{P} \cup \{\neg\forall(\varphi\Ra\psi).\chi\}, \mc{T}, \neg\forall(\neg\varphi).\forall\Omega.\chi\big)}
	\quad\quad
	\dfrac
	{\big(\mc{P}, \mc{T}, \neg\forall(\varphi\Ra\psi).\chi\big)}
	{\big(\mc{P} \cup \{\neg\forall(\varphi\Ra\psi).\chi\}, \mc{T}, \neg\forall\Omega.\forall\psi.\chi\big)}
	\end{gather*}
The \textit{finalized decomposition set} $\fd(\varphi)$ is the set of all the pairs $(\mc{T},\vartheta)$ such that there is a triple $(\mc{P},\mc{T},\vartheta)$ in $\mc{D}(\varphi)$ such that no decomposition rule can be applied to it and at the same time $\vartheta$ is not in $\mc{P}$.
\end{defi}

Intuitively speaking, the decomposition set of an $\alpha$/$\beta$ eventuality represents series of applications of the $\alpha$/$\beta$ rules. In a triple $(\mc{P},\mc{T},\vartheta)$, $\vartheta$ is used as the principal formula of the rule, the set $\mc{P}$ gathers all the principal formulas from the application of the decomposition rules and $\mc{T}$ gathers all the sentential tests (i.e. formulas that act as programs) that occur. Observe that, in a specific sequence of applications of the decomposition rules, a rule cannot be applied to the same eventuality more than once. So, if there is no re-occurrence of a principal formula, then the application of the rules stops when either a non-eventuality appears or some eventuality of the form $\neg\forall A.\chi$ with $A\in\atp_\Omega$ appears.

\begin{defi}	 		\label{def: reduction sets}
	The (family of) \textit{reduction sets} $\mc{R}_1^\varphi, \dotsc, \mc{R}_n^\varphi \subseteq\mc{L}_s$  of an $\alpha$/$\beta$ formula $\varphi$, for some $n\geq1$ depending on $\varphi$, are defined as follows:
	\begin{itemize}
		\item If $\varphi$ is an $\alpha$ non-eventuality, then its unique reduction set is the set $\mc{R}_1^\varphi = \{\alpha_1, \alpha_2\}$ and 
			if $\varphi$ is a $\beta$ non-eventuality, then the family of its reduction sets consists of the singletons $\mc{R}_1^\varphi = \{\beta_1\}$ and $\mc{R}_2^\varphi = \{\beta_2\}$.
		
		\item If $\varphi$ is an $\alpha$/$\beta$ eventuality, then for each pair $(\mc{T},\vartheta)\in \fd(\varphi)$ and for each $i=1,\dotsc,n$, where $n$ is equal to the cardinality of $\fd(\varphi)$, we define $\mc{R}_i^\varphi = \mc{T}\cup\{\vartheta\}$.
	\end{itemize}
	The size $n\geq1$ of the family of reduction sets will be referred to as the \textit{reduction degree} of the formula $\varphi$.
\end{defi}

\begin{exa}
We calculate the reduction sets for various forms of an $\alpha$/$\beta$ eventuality:
\begin{enumerate}
\item Let $\varphi$ be the formula $\neg\forall a^{***}.\psi$ and $\xi$ the formula $\forall a^{***}.\psi$ such that $\neg\psi$ is not in $\ev$:
\begin{align*}
	\mc{D}(\varphi) = \Big\{
		&\Big(\emptyset,\emptyset, \varphi\Big),\,
		\Big(\{\varphi\},\emptyset, \neg\psi\Big),\,   
		\Big(\{\varphi\},\emptyset, \neg\forall a^{**}\!.\xi\Big),\,
		\Big(\{\varphi,\neg\forall a^{**}\!.\xi\},\emptyset, \varphi\Big),\,
		\Big(\{\varphi,\neg\forall a^{**}\!.\xi\},\emptyset, \neg\forall a^*.\forall a^{**}\!.\xi\Big), \\
		&\Big(\{\varphi,\neg\forall a^{**}\!.\xi,\neg\forall a^*.\forall a^{**}.\xi\},\emptyset, \neg\forall a^{**}\!.\xi\Big), \;
		\Big(\{\varphi,\neg\forall a^{**}\!.\xi,\neg\forall a^*.\forall a^{**}\!.\xi\},\emptyset, \neg\forall a.\forall a^*.\forall a^{**}\!.\xi\Big)
	\Big\}
\end{align*}
\begin{equation*}
	\mc{R}^\varphi_1 = \{\neg\psi\}  \qquad  
	\mc{R}^\varphi_2 = \{\neg\forall a.\forall a^*.\forall a^{**}.\forall a^{***}.\psi \}
\end{equation*}

\item Let $\varphi$ be the formula $\neg\forall\xi^*.\psi$ such that $\neg\psi$ is not in $\ev$:
\begin{align*}
	\mc{D}(\varphi) = \Big\{ 
		\Big(\emptyset,\emptyset, \varphi\Big),\,
		\Big(\{\varphi\},\emptyset, \neg\psi\Big),   \,
		\Big(\{\varphi\}, \emptyset, \neg\forall\xi.\forall\xi^*.\psi\Big),    \,
		\Big(\{\varphi, \neg\forall\xi.\forall\xi^*.\psi\}, \{\xi\}, \neg\forall\xi^*.\psi\Big)
	\Big\}
	\quad \mc{R}^\varphi_1 = \{ \neg\psi \}
\end{align*}

\item Let $\varphi$ be the formula $\neg\forall(\vartheta a)^*.\psi$ and $\xi$ the formula $\forall(\vartheta a)^*.\psi$ such that $\neg\psi$ is not in $\ev$:
\begin{align*}
	\mc{D}(\varphi) \!=\! \Big\{
		&\big(\emptyset,\emptyset, \varphi\big),
		\big(\{\varphi\},\emptyset, \!\neg\psi\big),
		\big(\{\varphi\},\emptyset, \!\neg\forall\vartheta\! a.\xi\big),
		\Big(\!\{\varphi,\!\neg\forall\vartheta\! a.\xi\}, \emptyset, \!\neg\forall\vartheta\!.\forall a.\xi\Big), \\
		&\Big(\!\{\varphi,\!\neg\forall\vartheta\! a.\xi, \neg\forall\vartheta\!.\forall a.\xi\}, \{\vartheta\}, \neg\forall a.\xi\Big)
	\Big\}
\end{align*}
\begin{equation*}
	\mc{R}^\varphi_1 = \{ \neg\psi \}   \qquad   
	\mc{R}^\varphi_2 = \{ \vartheta, \neg\forall a.\forall(\vartheta a)^*.\psi \}
\end{equation*}

\item Let $\varphi$ be the formula $\neg\forall(A\!+\!\vartheta).\forall(A\!+\!\vartheta)^*.\psi$ such that $A$ is the process $\chi_1\!\Ra\!\chi_2$ and $\neg\psi$ is not in $\ev$. We let the interested reader calculate the decomposition set of $\varphi$ and we just give its reduction sets:
\begin{equation*}
	\mc{R}^\varphi_1 = \{ \neg\forall\Omega.\forall\chi_2.\forall((\chi_1\!\Ra\!\chi_2)\!+\!\vartheta)^*.\psi\}   					\qquad
	\mc{R}^\varphi_2 = \{ \neg\chi_1, \neg\forall\Omega.\forall((\chi_1\!\Ra\!\chi_2)\!+\!\vartheta)^*.\psi\}   					\qquad
	\mc{R}^\varphi_3 = \{ \vartheta,\!\neg\psi \}
\end{equation*}
\end{enumerate}
\end{exa}

\begin{defi} 	  	\label{def: vtrd}
	The binary relation $\vtrd$ relates formulas $\varphi$ and $\psi$ and we write $\varphi\vtrd\psi$ iff:
	\begin{itemize}
	\item $\varphi$ is an $\alpha$/$\beta$ formula of the form $\neg\forall A.\chi$ and $\psi$ belongs to one of the reduction sets of $\varphi$ and
	\item if $\varphi$ is a non-eventuality, then $\varphi\vtr\psi$ and 
	if $\varphi$ is an eventuality, then there is a pair $(\mc{T},\vartheta)\in\fd(\varphi)$ such that $\vartheta=\psi$ (notice that the set $\mc{T}\cup\{\vartheta\}$ is a reduction set of $\varphi$).
	\end{itemize}
\end{defi}

\begin{lem}	 		\label{lem: properties of vtrd}
	If $\varphi\vtrd\psi$ and $\varphi$ is an eventuality $\neg\forall A_1.\dotsb\forall A_n.\forall A^*.\chi$ with $n\geq0$ and $\neg\chi\notin\ev$,
	then $\psi$ is either the non-eventuality $\neg\chi$, or an eventuality of the form $\neg\forall B_1.\dotsb\forall B_r.\forall A^*.\chi$ with $r\geq1$ and $B_1\in\atp_\Omega$ (i.e. $B_1\in\atp \cup\{\Omega\}$).
\end{lem}
\begin{proof}
	According to Definition \ref{def: vtrd}, there is a reduction set $\mc{R}$ of $\varphi$ such that $\psi\in\mc{R}$ and
	by Definition \ref{def: reduction sets}, we know that $\mc{R}$ is defined based on a pair from the finalized decomposition set of $\varphi$. 
	By Definition \ref{def: decomposition}, we know that there is a sequence $(\mc{P}_0,\mc{T}_0,\xi_0), \dotsc, (\mc{P}_n,\mc{T}_n,\xi_n)$ of triples of $\mc{D}(\varphi)$, with $n\geq1$, such that:
	\begin{enumerate*}[i)]
		\item $(\mc{P}_0,\mc{T}_0,\xi_0)=(\emptyset,\emptyset,\varphi)$,
		\item for $j=0,\dotsc,n-1$, $\dfrac  {(\mc{P}_i,\mc{T}_i,\xi_i)}  {(\mc{P}_{i+1},\mc{T}_{i+1},\xi_{i+1})}$ is an instance of one of the decomposition rules of Definition \ref{def: decomposition} such that $\xi_i$ is an $\alpha$/$\beta$ eventuality,
		\item $(\mc{T}_n,\xi_n)$ is in $\fd(\varphi)$ and $\mc{R} = \mc{T}_n\cup \{\xi_n\}$,
		\item $\xi_n=\psi$ is either a non-eventuality, or an eventuality of the form $\neg\forall B.\vartheta$ with $B\in\atp_\Omega$.
	\end{enumerate*}
	Due to the syntactic form of $\varphi$ and of the decomposition rules, we can show that $\xi_0,\dotsc,\xi_{n-1}$ are eventualities of the form $\neg\forall C_1.\dotsb\forall C_m.\forall A^*.\chi$ with $m\geq0$.
	If $\xi_n\in\ev$, then the previous property also holds for $\xi_n=\psi$.
	In the case that $\xi_n\notin\ev$, since $\xi_{n-1}$ is an eventuality $\neg\forall C_1.\dotsb\forall C_m.\forall A^*.\chi$ with $m\geq0$, by careful inspection of the decomposition rules, we have to admit that $\xi_{n-1} = \neg\forall A^*.\chi$ and $\xi_n=\psi$ is the non-eventuality $\neg\chi$.
\end{proof}

\begin{defi}	 		\label{def: fpath} 
	Let $\mc{M}$ be an $\mc{L}$-model, $w$ a state of $\mc{M}$, $\varphi\!=\!\neg\forall A_1.\dotsb\forall A_r.\forall A^*.\vartheta$ an eventuality, with $r\geq0$ and $\neg\vartheta\notin\ev$, and $R$ one of the relations $\vtr$ and $\vtrd$.
	A \textit{fulfilling path} $\fpath(\mc{M},w,\varphi,R)$  is a finite sequence $(w_1,\psi_1),\dotsc,(w_k,\psi_k)$ of pairs of states and formulas, with $k\geq2$, such that:
	\begin{itemize}
	\item $w_1=w$, $\psi_1=\varphi$ and $\psi_k$ is the formula $\neg\vartheta$ which appears only in the last pair,
	\item for $i=1,\dotsc,k$, $w_i\models^\mc{M}\psi_i$,
	\item for $i=1,\dotsc,k\!-\!1$, 
	if $\psi_i$ is some formula $\neg\forall A.\chi$ such that $A\in\atp_\Omega$, then $\psi_{i+1}=\neg\chi$ and $w_i\!\onto{A}\!w_{i+1}$, or else, $\psi_iR\psi_{i+1}$ and $w_i\!=\!w_{i+1}$.
	\end{itemize}
	If $R$ is the relation $\vtrd$, then for $i=1,\dotsc,k-1$, if $\psi_i\vtrd\psi_{i+1}$ and $w_i=w_{i+1}$, then there is a reduction set $\mc{R}$ of $\psi_i$ such that $\psi_{i+1}\in\mc{R}$ and $w_i\models^\mc{M}\mc{R}$.
\end{defi}

\begin{lem}	 		\label{lem: fpaths}
	If an eventuality $\varphi=\neg\forall A_1.\dotsb\forall A_r.\forall A^*.\vartheta$, with $r\geq0$ and $\neg\vartheta\notin\ev$, is satisfied by a state $w$ of an $\mc{L}$-model $\mc{M}$, then there are fulfilling paths $\fpath(\mc{M},w,\varphi,\vtr)$ and $\fpath(\mc{M},w,\varphi,\vtrd)$.
\end{lem}
\begin{proof}
The proof of the existence of $\fpath(\mc{M},w,\varphi,\vtr)$ is similar to the corresponding proof for the case of \pdl, see \cite[Chapter 4]{Widmann_Thesis}, besides the case that concerns the precondition-effect construct. 

First, we consider a formula $\neg\forall A.\vartheta$ and we assume that there are states $w$ and $w'$ of an $\mc{L}$-model $\mc{M}$ such that
 $w\onto{A}w'$ and $w\models^\mc{M} \neg\forall A.\vartheta$ and $w'\models^\mc{M} \neg\vartheta$. 
We show that there is a sequence $(w,\neg\forall A.\vartheta),\dotsc,(w',\neg\vartheta)$ of pairs of states and formulas with the same properties as those of a fulfilling path. The proof is by structural induction on $A$.
We argue only for the case $A=\varphi\!\Ra\!\psi$, as the rest cases are similar to those of \pdl \cite{Widmann_Thesis}.
Since $\onto{\Omega}$ is the universal definable relation, $\varphi\Ra\psi \vmodels \Omega$ and thus, $w\onto{\Omega}w'$. Moreover, since $w\onto{\Omega}w'$ and $w'\models^\mc{M}\neg\vartheta$, we have that $w\models^\mc{M}\neg\forall\Omega.\vartheta$.
We know that $w$ satisfies either $\neg\varphi$ or $\varphi$.
In the former case, having in mind that $w\models^\mc{M}\neg\forall\Omega.\vartheta$, we can conclude that $w\models^\mc{M}\neg\forall(\neg\varphi).\forall\Omega.\vartheta$. Observe that $\neg\forall (\varphi\Ra\psi).\chi \vtr \neg\forall(\neg\varphi).\forall\Omega.\vartheta$. So, the required sequence is the sequence $(w,\neg\forall A.\vartheta), (w,\neg\forall(\neg\varphi).\forall\Omega.\vartheta), (w,\neg\forall\Omega.\vartheta), (w',\neg\vartheta)$.
Now, in the case that $w$ satisfies $\varphi$, as $w$ satisfies $\neg\forall(\varphi\Ra\psi).\vartheta$, by Proposition \ref{prop: equiv alpha beta}, we know that $w$ should also satisfy $\neg\forall\Omega.\forall\psi.\vartheta$ (it cannot satisfy $\neg\forall(\neg\varphi).\forall\Omega.\vartheta$). By Corollary \ref{cor: =>2}, since $w\lonto{\varphi\Ra\psi}w'$, we know that either $w\models^\mc{M}\neg\varphi$, or $w'\models^\mc{M}\psi$. Since $w\models^\mc{M}\varphi$, we have to admit that $w'\models^\mc{M}\psi$. Since $w'\models^\mc{M}\psi$ and $w'\models^\mc{M}\neg\vartheta$, we can also conclude that $w'\models^\mc{M} \neg\forall\psi.\vartheta$.
Thus, the necessary sequence this time is the sequence $(w,\neg\forall A.\vartheta), (w,\neg\forall\Omega.\forall\psi.\vartheta), (w',\neg\forall\psi.\vartheta), (w',\neg\vartheta)$.

Now, we consider an eventuality $\varphi=\neg\forall A_1.\dotsb\forall A_r.\forall A^*.\vartheta$ which is satisfied by a state $w$ of an $\mc{L}$-model $\mc{M}$ and we show that there is a fulfilling path $\fpath(\mc{M},w,\varphi,\vtr)$. To simplify notation, we define $n=r+1$ and $A_n=A^*$ and thus, $\varphi = \neg\forall A_1.\dotsb\forall A_n.\vartheta$.
Since $w\models^\mc{M}\varphi$, there is a sequence $w_1,\dotsc,w_{n+1}$ of states such that $w_1=w$ and $w_1\onto{A_1}w_2\dotsb w_n\onto{A_n}w_{n+1}$ and for $i=1,\dotsc,n$, $w_i\models^\mc{M} \neg\forall A_i.\dotsb\forall A_n.\vartheta$ and $w_{n+1}\models^\mc{M}\neg\vartheta$.
By the properties shown earlier, we have that for $i=1,\dotsc,n-1$, there is the appropriate sequence $\pi_i = (w_i,\neg\forall A_i.\dotsb\forall A_n.\vartheta),\dotsc, (w_{i+1},\neg\forall A_{i+1}.\dotsb\forall A_n.\vartheta)$ and there is also the sequence $\pi_n = (w_n,\neg\forall A_n.\vartheta),\dotsc, (w_{n+1},\neg\vartheta)$.
Thus, the desired $\fpath(\mc{M},w,\varphi,\vtr)$ is the concatenation of $\pi_1,\dotsc,\pi_n$, in which we contract the consecutive identical pairs.

In the sequel, we take advantage of $\fpath(\mc{M},w,\varphi,\vtr)$ in order to show the existence of $\fpath(\mc{M},w,\varphi,\allowbreak\vtrd\nolinebreak)$.
First, we replace each subsequence $(w_1,\chi_1),\dotsc,(w_n,\chi_n)$ of $\fpath(\mc{M},\allowbreak w,\varphi,\vtr)$ such that $w_1=\dotsb=w_n$ and $\chi_1=\chi_n$ with the single pair $(w_1,\chi_1)$. Suppose that the modified fulfilling path is the sequence $(w_1,\psi_1),\dotsc,(w_k,\psi_k)$.

Let $\psi_i$, where $i=1,\dotsc,k-1$, be an $\alpha$/$\beta$ non-eventuality. By the properties of $\fpath(\mc{M},\allowbreak w,\varphi,\vtr)$, we know that $\psi_i\vtr\psi_{i+1}$ and $w_i=w_{i+1}$ and $w_i\models^\mc{M}\psi_i$ and $w_{i+1}\models^\mc{M}\psi_{i+1}$.
Since $\psi_i$ is not an eventuality, by Definition \ref{def: vtrd}, we know that $\psi_i\vtrd\psi_{i+1}$. Moreover, by Definitions \ref{def: vtr} and \ref{def: reduction sets} and Proposition \ref{prop: equiv alpha beta}, we can conclude that there is a reduction set of $\psi_i$ which is satisfied by $w_i$ and at the same time, $\psi_{i+1}$ is in it. 

For any two subsequent pairs $(w_r,\psi_r),(w_{r+1},\psi_{r+1})$ of $\fpath(\mc{M},w,\varphi,\vtr)$, with $1\leq r<k$, such that $\psi_r$ is not a formula of the form $\neg\forall B.\chi$ with $B\in\atp_\Omega$, the state $w_r$ is the same with $w_{r+1}$ and there is a decomposition rule $\dfrac  {(\mc{P}_1,\mc{T}_1,\xi_1)}  {(\mc{P}_2,\mc{T}_2,\xi_2)}$ of Definition \ref{def: decomposition} such that $\psi_r$ is of the form $\xi_1$ and $\psi_{r+1}$ is of the form $\xi_2$ and $\xi_1\vtr\xi_2$. Moreover, the decomposition rules are essentially a variation of the usual $\alpha$/$\beta$ rules (recall Remark \ref{rem: alpha-beta formulas}): $\dfrac  {(\mc{P},\mc{T},\alpha)}  {(\mc{P}\cup\{\alpha\},\mc{T}\cup\{\alpha_2\},\alpha_1)}$ and $\dfrac  {(\mc{P},\mc{T},\beta)}  {(\mc{P}\cup\{\beta\},\mc{T},\beta_i)}$, where $i\in\{1,2\}$. Thus, by Proposition \ref{prop: equiv alpha beta}, whenever some premise triple $(\mc{P}_1,\mc{T}_1,\xi_1)$ is satisfied by a state $w$ of an $\mc{L}$-model $\mc{M}$ (i.e. the set $\mc{P}_1 \cup \mc{T}_1 \cup \{\xi_1\}$ is satisfied by $w$), we can conclude that there is a decomposition rule such that its conclusion triple $(\mc{P}_2,\mc{T}_2,\xi_2)$ is also satisfied by the same state.
So, for $i=1,\dotsc,k$, if $\psi_i$ is an $\alpha$/$\beta$ eventuality, then there is a sequence $(\mc{P}_0,\mc{T}_0,\xi_0), \dotsc, (\mc{P}_n,\mc{T}_n,\xi_n)$ of triples of $\mc{D}(\psi_i)$, with $(\mc{P}_0,\mc{T}_0,\xi_0)=(\emptyset,\emptyset,\psi_i)$ (recall that $w_i\models^\mc{M}\psi_i$) and $n\geq1$, such that:
\begin{itemize}
\item for $j=0,\dotsc,n-1$, $\dfrac  {(\mc{P}_i,\mc{T}_i,\xi_i)}  {(\mc{P}_{i+1},\mc{T}_{i+1},\xi_{i+1})}$ is an instance of one of the decomposition rules of Definition \ref{def: decomposition} such that $\xi_i$ is an $\alpha$/$\beta$ eventuality and $\xi_i\vtr\xi_{i+1}$,
\item for $j=0,\dotsc,n$, $w_{i+j}=w_i$ and $\psi_{i+j}=\xi_j$ and the triple $(\mc{P}_i,\mc{T}_i,\xi_i)$ is satisfied by $w_i$,
\item for $j=0,\dotsc,n$, $\mc{P}_j$ is the set $\{\xi_0,\dotsc,\xi_{j-1}\}$ and $\xi_j$ is not in $\mc{P}_j$, as otherwise there would be a subsequence $(w,\chi_1),\dotsc,(w,\chi_m)$ of $\fpath(\mc{M},w,\varphi,\vtr)$ such that $\chi_1=\chi_m$,
\item the pair $(\mc{T}_n,\xi_n)$ is in $\fd(\psi_i)$ which, according to Definitions \ref{def: vtrd} and \ref{def: reduction sets}, entails that $\psi_i\!\vtrd\!\xi_n$ (notice that $\xi_n \!\!=\! \psi_{i+n}$) and that the set $\mc{T}_n\!\cup\!\{\xi_n\}$ is a satisfiable reduction set of $\psi_i$,
\item $\xi_n$ is either a non-eventuality, or an eventuality of the form $\neg\forall B.\chi$ with $B\in\atp_\Omega$.
\end{itemize}
The formula of the last pair of $\fpath(\mc{M},w,\varphi,\vtr)$ is a non-eventuality. According to Definition \ref{def: decomposition}, all the principal formulas of the triples of the considered sequence for the $\alpha$/$\beta$ eventuality $\psi_i$ need to be eventualities, possibly except for the last one. Therefore, $i+n$ cannot be greater than the length of $\fpath(\mc{M},w,\varphi,\vtr)$.
By the previous properties, it should be clear that if we remove the appropriate pairs from $\fpath(\mc{M},w,\varphi,\vtr)$ in which the corresponding formulas are $\vtr$-related, then we can obtain $\fpath(\mc{M},w,\varphi,\vtrd)$.
\end{proof}

The language $\mc{L}_a$ of process terms of \tpdl is the language of the regular expressions on the alphabet $\Sigma$ (where $\Sigma$ is generated by the grammar $a\; (a\in\atp) \mid \varphi \mid \varphi\Ra \varphi\;(\varphi\in\mc{L}_s)$).
In the following, we show that each process/regular expression of $\mc{L}_a$ defines a language of processes of sequentially composed elements of $\atp_\Omega\cup\mc{L}_s$.

\begin{defi} 		\label{lang} 
	The language generated by some process $A\in\mc{L}_a$ is defined by the following rules: 
	$\mc{L}(a) = \{a\}$, 
	$\mc{L}(\varphi) = \{\varphi\}$, 
	$\mc{L}(\Omega) = \{\Omega\}$, 
	$\mc{L}(\varphi\Ra\psi) = \mc{L}(\neg\varphi\Omega + \Omega\psi)$, 
	$\mc{L}(A+B) = \mc{L}(A) \cup \mc{L}(B)$, 
	$\mc{L}(AB) = \mc{L}(A)\mc{L}(B)$, 
	$\mc{L}(A^*) = \mc{L}(A)^*$.
	As usual, if $\mc{L}_1$ and $\mc{L}_2$ are languages of processes, then $\mc{L}_1\mc{L}_2 = \{ A_1 A_2 \mid A_1\in\mc{L}_1 \text{ and } A_2\in\mc{L}_2\}$ and $\mc{L}_1^* = \bigcup_{n\geq0} \mc{L}_1^n$, where $\mc{L}_1^0= \{\varepsilon\}$ and $\mc{L}_1^{n+1} = \mc{L}_1\mc{L}_1^n$.
	By the previous rules, it is clear that the words/processes of a language $\mc{L}(A)$ are members of $(\atp_\Omega\cup \mc{L}_s)^*$.
	
	Let $A=A_1\dotsb A_k$ and $B= B_1\dotsb B_n$ be processes of $(\atp_\Omega\cup \mc{L}_s)^*$ different from the empty string such that  
	each one of their subprocesses (i.e the process $A_i$, where $1\leq i\leq k$, and $B_j$, where $1\leq j\leq n$) is in $\atp_\Omega$ or it is a sequence of sequentially composed formulas of $\mc{L}_s$, the largest possible 
	(in the sense that there is no $1\leq i< k$ such that $A_iA_{i+1}\in \mc{L}_s^*$ and there is no $1\leq i< n$ such that $B_iB_{i+1}\in \mc{L}_s^*$).
	We say that $A$ and $B$ are \textit{interchangeable} iff one of the following conditions hold:
	\begin{enumerate*}[i)]
	\item $A = B$, or
	\item $k=n$ and for $i=1,\dotsc,k$, 
	if $A_i\in\atp_\Omega$, then $A_i=B_i$, and 
	if $A_i=\varphi_1\dotsb\varphi_r$ with $r\geq 1$, then $B_i=\psi_1\dotsb\psi_l$ with $l\geq1$ such that the formula $\bigwedge\limits_{1\leq j\leq l}\psi_j$ is semantically equivalent with $\bigwedge\limits_{1\leq j\leq r}\varphi_j$.
	\end{enumerate*}
\end{defi}

\begin{lem}	\label{lem: lang properties} 
	In any $\mc{L}$-model, 
	for any process $A\in\mc{L}_a$, $\onto{A} \;= \bigcup\limits_{B\in\mc{L}(A)} \onto{B}$, and
	for any interchangeable processes $A$ and $B$ of $(\atp_\Omega\cup \mc{L}_s)^*$, $\onto{A}\,=\,\onto{B}$.
\end{lem}
\begin{proof}
	The proof of the first property is by induction on the structure of $A$. We distinguish cases on $A$ and we make use of the interpretation of the processes (see Table \ref{tab: int}) and of Corollary \ref{cor: =>2} ($\lonto{\varphi\Ra\psi} = \lonto{\neg\varphi\Omega + \Omega\psi}$). Then, the proof follows by the induction hypothesis and of the rules for the definition of the language that is generated by a process (see Definition \ref{lang}).
	The second property follows immediately from the semantics, as the interpretation of semantically equivalent sentential tests (i.e. formulas that act as programs) is the same.
\end{proof}

In the following of the paper, the definition of the size of a formula and of a process will be useful. The definition that follows is the same as that given in \cite{Hartonas1}. In accordance with \cite{Hartonas1}, the definition treats some constructs slightly differently in comparison to the case of \pdl. This is done to compensate the failure of a linear bound of the Fischer-Ladner closure for \tpdl and prove instead a quadratic bound.
For more details, we refer the reader to \cite{Hartonas1}.

\begin{defi} 		\label{def: size} 
	The size $|\varphi|$ of a formula $\varphi$ and the size $\siz{A}$ of a program $A$ are defined as follows:
	\begin{align*}
	& |p| = 1					&&	|\forall A.\varphi| = 1+\siz{A}+|\varphi|	&&	\siz{a} = 1 				&&	\siz{A+B} = 1 + \siz{A} + \siz{B} \\
	& |\neg\varphi|=1+|\varphi| &&	|\cpi A| = 1 + \siz{A}						&&	\siz{\varphi} = 1+|\varphi| &&	\siz{AB} =1+2\siz{A}+ \siz{B} \\
	& 							&&								&&	\siz{\varphi\Ra\psi} = 1+|\varphi|+|\psi|	&&	\siz{A^*} = 1+2\siz{A} 
	\end{align*}
	Moreover, we define the size of $\true$ and $\false$ to be equal to $1$ and thus, $\siz{\Omega}=\siz{\true\Ra\true}=3$. 
\end{defi}

\section{Hintikka Structures} \label{sec: Hintikka structures} 

In this section, we present the usual Hintikka structures (also called in the literature model graphs) and we show that the formulas which appear in them are satisfiable. Moreover, a Hintikka structure reveals the kind of properties that our algorithm should ensure in order to decide affirmatively for the satisfiability of a formula.

\begin{defi} \label{def: labelled structure}   
	A \emph{labelled structure} $\mc{A}$ is a triple $\langle S, (\leads{A})_{A\in\atp_\Omega},  L \rangle$, where: 
	\begin{itemize}
	\item $S$ is a nonempty set of states.
	\item $\leads{}: \atp_\Omega \lra 2^{S\times S}$ assigns a binary relation to each atomic process $a\in\atp$ and to the term $\Omega$.
	\item $L: S \lra 2^{\mc{L}_s}$ labels each state of $S$ with a set of formulas of the language $\mc{L}_s$.
	\end{itemize}
\end{defi}

\begin{defi}  		\label{def: structure path} 
	In a labelled structure $\mc{A} \!=\! \langle S, (\leads{A})_{A\in\atp_\Omega},  L \rangle$, a \textit{structure path} $\stp(s,\varphi, s',\psi)$ is a sequence of state-formula pairs $(s_0, \varphi_0), \dotsc, (s_k, \varphi_k)$ with $k\geq 0$, such that:
	\begin{itemize}
	\item $s_0 = s$, $\varphi_0 = \varphi$, $s_k = s'$, $\varphi_k = \psi$ and for $i=0,\dotsc,k$, $\varphi_i\in L(s_i)$,
	\item for $i=0,\dotsc,k-1$, if $\varphi_i$ is a formula of the form $\neg\forall A.\chi$ such that $A\in\atp_\Omega$, then $\varphi_{i+1}=\neg\chi$ and $s_i\leads{A}s_{i+1}$, otherwise, $\varphi_i\vtrd\varphi_{i+1}$ and $s_{i+1} = s_i$ and there is a reduction set $R$ of $\varphi_i$ such that $\varphi_{i+1}\in R$ and $R\subseteq L(s_{i+1})$.
	\end{itemize}
\end{defi}

\begin{defi} \label{def: Hintikka structure}  
	A \emph{Hintikka structure} $\mc{H} =\pos{S, (\leads{A})_{A\in\atp_\Omega}, L}$ is a labelled structure such that, for any state $s \in S$, the labelling function $L$ meets the following conditions:
	\begin{enumerate}[(H1)]
	\item \label{contr} For each $p\in\atf$, if $\neg p \in L(s)$, then $p\notin L(s)$, and 
	for each process $A\in\wt{\Sigma}$, for each agent $\imath\in I$, if $\neg\cpi A \in L(s)$, then $\cpi A\notin L(s)$.
	\item \label{cap1} No capability statement of the form $\neg\cpi\varphi$ is in $L(s)$.
	\item \label{alpha/beta} If $\varphi\in L(s)$ is an $\alpha$/$\beta$ formula, then one of its reduction sets is a subset of $L(s)$.
	\item \label{cap2} 
	If some capability statement $\neg\cpi A$ with $A\in\wt{\Sigma}$ and $\imath\in I$ is in $L(s)$ and $\{ \cpi(\varphi_1\Ra\psi_1),\dotsc,\cpi(\varphi_k\Ra\psi_k) \}$ with $k\geq0$ is the set of all the capabilities statements in $L(s)$ that concern precondition-effect processes for $\imath\in I$, 
	then there is a state $s'\in S$ such that $\{ \varphi_1,\dotsc,\varphi_k, \neg\forall A.\forall(\neg\psi_1).\dotsb\forall(\neg\psi_k).\false \} \subseteq L(s')$.
	\item \label{box} If $\forall a.\vartheta \in L(s)$, then $\forall s' \in S$, if $s \leads{a} s'$, then $\vartheta \in L(s')$.
	\item \label{Omega} If $\forall\Omega.\vartheta \in L(s)$, then 
	\begin{enumerate*}[i)]
	\item $\vartheta$ is a formula of the form $\forall\Omega^*.\varphi$ and it is in $L(s)$, and
	\item for each $A\in\atp_\Omega$, for each $s'\in S$, if $s\leads{A}s'$, then $\vartheta\in L(s')$.
	\end{enumerate*}
	\item \label{diamond} If $\neg\forall A.\chi \in L(s)$ such that $A\in\atp_\Omega$, then there is a state $s'\in S$ such that $s \leads{A} s'$ and $\neg\chi\in L(s')$.
	\item \label{ev} If an $\alpha$/$\beta$ eventuality $\varphi=\neg\forall A_1.\dotsb\forall A_k.\forall A^*.\chi$ with $k\geq0$ and $\neg\chi\notin\ev$ is in $L(s)$, then there is a structure path $\stp(s,\varphi, s',\neg\chi)$ with $s'\in S$.
	\end{enumerate}
	Finally, the transitions in $\mc{H}$ are restricted as follows:
	\begin{enumerate}[resume*]
	\item \label{transitions} For any process $A\in\atp_\Omega$, for any states $s_1,s_2\in S$, if $s_1\leads{A}s_2$, then there is no other process $A'\in\atp_\Omega$ such that $s_1\leads{A'}s_2$ (i.e. for any processes $A,A'\in\atp_\Omega$, if $s_1\leads{A}s_2$ and $s_1\leads{A'}s_2$, then $A=A'$).
	\end{enumerate}
	We say that there is a Hintikka structure $\mc{H}$ for a formula $\varphi$ iff there is a state $s$ of $\mc{H}$ such that $\varphi \in L(s)$.
\end{defi}

\begin{defi}   \label{def: Hintikka model}
	If $\mc{H} = \pos{S, (\leads{A})_{A\in\atp_\Omega}, L }$ is a Hintikka structure, then the $\mc{L}$-model which corresponds to $\mc{H}$, denoted by $\mc{M(H)}$, is the triple $\pos{\mc{F},\rho, (\imath^\mc{M(H)})_{\imath\in I} }$ with $\mc{F} = \pos{W, (\onto{a})_{a\in\atp}, I }$, where $W=S$, such that:
	\begin{itemize}
	\item The transition function $\onto{}$ is defined as follows: 
		\begin{itemize}
		\item for each $a\in\atp$, $\onto{a} \,=\, \leads{a}$,
		\item for the universal definable relation, $\onto{\Omega} \,=\, \onto{\true} \cup \{ (w,w') \mid \exists A_1\dotsb A_k\in\atp_\Omega^+ (w\leads{A_1}\leads{A_2}\dotsb\leads{A_k}w')\}$,
		\item for each process term $\varphi\Ra\psi$, $\lonto{\varphi\Ra\psi} \,=\, \lonto{\neg\varphi\Omega} \cup \lonto{\Omega\psi}$.
		\end{itemize}
	
	\item The interpretation function $\rho$ is defined as follows: for each $p\in\atf$,  $\rho(p) = \{s\in S \mid p\in L(s)\}$.
	
	\item For each $\imath\in I$, for each $w \in W$, we have that
	$\imath^\mc{M(H)}(w)  =  \bigcup\{  \onto{A}  \mid  A\in\wt{\Sigma}  \land  \cpi A\in L(w) \}$.
	\end{itemize}
\end{defi}

\begin{rem}
	We point out that the defined $\mc{L}$-model $\mc{M(H)}$ of Definition \ref{def: Hintikka model} is indeed an $\mc{L}$-model. 
	Definition \ref{def: l-model} does not explicitly define the interpretation of the terms $\varphi\Ra\psi$, but it provides a condition that their interpretation should satisfy (see Table \ref{tab: int}).
	The way that we interpreted the terms $\varphi\Ra\psi$ in Definition \ref{def: Hintikka model} is compatible with that condition. First, notice that the interpretation of $\Omega$ is defined in terms of $\leads{}$. It is indeed the universal definable relation as it considers all the possible transitions in the corresponding Hintikka structure, as well as the transitions of the relation $\onto{\true}$. Second, we defined the interpretation of a process term $\varphi\Ra\psi$ by using $\Omega$ in the way that Corollary \ref{cor: =>2} states.
	Finally, by the way that we defined the capabilities assignment function follows that the normality condition of equation \eqref{eq: normality} is satisfied  and thus, we have that $\imath^\mc{M(H)}(w)  =  \bigcup\{\onto{A} \mid A\in\wt{\Sigma} \text{ and } w\models^\mc{M(H)} \cpi A \}$ for each $\imath\in I$ and for each $w\in W$.
\end{rem}

\begin{lem} 		\label{lem: Hintikka->Model} 
	If there is a Hintikka structure $\mc{H} = \pos{S,(\onto{a})_{a\in \atp}, L}$ for a formula $\xi$, then the formula $\xi$ is satisfied by a state of the model $\mc{M(H)}$ which corresponds to $\mc{H}$.
\end{lem}
\begin{proof}
According to Definition \ref{def: Hintikka model}, $\mc{M(H)}$ is the triple $\pos{\mc{F},\rho, (\imath^\mc{M(H)})_{\imath\in I} }$ such that $\mc{F} = \pos{W, (\onto{a})_{a\in\atp}, I }$. Moreover, $W$ is the set of states $S$.
We need to show that for any formula $\xi$, for any state $s\in S$, if $\xi\in L(s)$, then $s\models^\mc{M(H)}\xi$.
We proceed by induction on the size of the formula $\xi$ (recall Definition \ref{def: size}). The induction hypothesis is stated as follows:
\begin{equation}	\label{IH}
\text{For any formula } \psi, \text{ if } |\psi|<|\xi|, \text{ then } \forall s\in S \big( \psi\in L(s) \,\Ra\, s\models^\mc{M(H)} \psi \big).
\tag{I.H.}
\end{equation}

In the following, for some state $s\in S$, we examine all the forms that a formula $\xi\in L(s)$ may have.
First, we assume that the formula $\xi$ is an atomic formula $p\in\atf$ or its negation. Since $p\in L(s)$, by the definition of $\mc{M(H)}$ (see Definition \ref{def: Hintikka model}), $s$ is in $\rho(p)$ and thus, $s\models^{\mc{M(H)}} p$. In the case that $\neg p\in L(s)$, we know that $p$ cannot be in $L(s)$ (see condition \ref{contr} of Definition \ref{def: Hintikka structure}). So, $s$ is not in $\rho(p)$ and thus, s does not satisfy $p$. In other words, $s$ satisfies $\neg p$.

In the sequel, we examine all the cases that concern the formulas of the form $\forall A.\vartheta$ and $\neg\forall A.\vartheta$. First, we examine the cases that concern an iterated process, as they are the most interesting and they are also used in subsequent cases.

We assume that $\xi = \forall A^*.\vartheta$. 
Suppose that there is a sequential transition $s_0\onto{A}s_1\onto{A}s_2 \dotsb s_{k-1}\onto{A}s_k$ such that $s_0 = s$ and $k\geq0$. We show that $\vartheta$, $\forall A.\forall A^*.\vartheta\in L(s_k)$. The proof proceeds by induction on $k$. 
In the case that $k=0$, since $\forall A^*.\vartheta\in L(s)$, by \ref{alpha/beta} of Definition \ref{def: Hintikka structure}, we know that $\vartheta, \forall A.\forall A^*.\vartheta\in L(s)$.
Now, for the composite transition $s_0\onto{A}s_1\onto{A}s_2 \dotsb s_{n}\onto{A}s_{n+1}$, we need to show that $\vartheta$, $\forall A.\forall A^*.\vartheta\in L(s_{n+1})$. By the induction hypothesis, we know that $\vartheta$, $\forall A.\forall A^*.\vartheta\in L(s_n)$. We proceed by induction on the structure of $A$:
\begin{enumerate}
	\item $A=a$: Since $s_n\onto{a}s_{n+1}$, by Definition \ref{def: Hintikka model}, we know that $s_n\leads{a}s_{n+1}$. By \ref{box} of Definition \ref{def: Hintikka structure}, since $\forall a.\forall a^*.\vartheta\in L(s_n)$ and $s_n\leads{a}s_{n+1}$, we can conclude that $\forall A^*.\vartheta\in L(s_{n+1})$. Furthermore, by \ref{alpha/beta}, we have that $\vartheta, \forall A.\forall A^*.\vartheta\in L(s_{n+1})$.
	
	\item $A =\varphi$: By semantics, the existence of the transition $s_0\onto{\varphi}s_1 \dotsb s_{n}\onto{\varphi}s_{n+1}$ shows that $s \models^\mc{M(H)} \varphi$ and for $i=1,\dotsc,n+1$, $s_i=s$. Since $s_n = s_{n+1}$, it is immediate that $\vartheta$, $\forall A.\forall A^*.\vartheta\in L(s_{n+1})$.
	
	\item $A=\Omega$: Since $s_n\onto{\Omega}s_{n+1}$, by Definition \ref{def: Hintikka model}, we know that either $s_n\onto{tt}s_{n+1}$, or $s_n\leads{A_1}\leads{A_2}\dotsb\leads{A_m}s_{n+1}$, where $m\geq1$ and for $i=1,\dotsc,m$, $A_i\in\atp_\Omega$. 
	In the first case, we can conclude that $s_n = s_{n+1}$ and that $\vartheta$, $\forall A.\forall A^*.\vartheta\in L(s_{n+1})$.
	In the second case, by induction on $m$ and by using conditions \ref{Omega} and \ref{alpha/beta} of Definition \ref{def: Hintikka structure}, we can show that $\vartheta, \forall\Omega.\forall\Omega^*.\vartheta\in L(s_{n+1})$.
	
	\item $A=\varphi\Ra\psi$: Since $s_n\lonto{\varphi\Ra\psi}s_{n+1}$, by Definition \ref{def: Hintikka model} (see also Corollary \ref{cor: =>2}), we have that either $s_n\onto{\neg\varphi}s_n\onto{\Omega}s_{n+1}$, or $s_n\onto{\Omega}s_{n+1}\onto{\psi}s_{n+1}$. Since in both cases $s_n\onto{\Omega}s_{n+1}$, we know that either $s_n\onto{tt}s_{n+1}$, or $s_n\leads{A_1}\leads{A_2}\dotsb\leads{A_m}s_{n+1}$, where $m\geq1$ and for $i=1,\dotsc,m$, $A_i\in\atp_\Omega$.
	In the first case, we conclude that $s_n = s_{n+1}$ and that $\vartheta$, $\forall A.\forall A^*.\vartheta\in L(s_{n+1})$.
	In the second, we first notice that since $\forall(\varphi\Ra\psi).\forall A^*.\vartheta\in L(s_n)$, by \ref{alpha/beta}, either $\{\varphi, \forall\Omega^*.\forall\psi.\forall A^*.\vartheta\}\subseteq L(s_n)$, or $\forall\Omega^*.\forall A^*.\vartheta\in L(s_n)$. So, by induction on $m$ and by using \ref{Omega} and \ref{alpha/beta}, we can show that either $\forall\Omega^*.\forall\psi.\forall A^*.\vartheta \in L(s_{n+1})$, or $\forall\Omega^*.\forall A^*.\vartheta\in L(s_{n+1})$.
	In the case that $\forall\Omega^*.\forall\psi.\forall A^*.\vartheta \in L(s_{n+1})$, recall that we know that $\varphi\in L(s_n)$ and by \ref{alpha/beta}, we conclude first that $\forall\psi.\forall A^*.\vartheta \in L(s_{n+1})$ and then that either $\neg\psi \in L(s_{n+1})$ or $\forall A^*.\vartheta \in L(s_{n+1})$. By the outer induction hypothesis, since $|\varphi|<|\xi|$, we conclude that $s_n$ satisfies $\varphi$. Now, since $s_n\models^\mc{M(H)}\varphi$ and $s_n\lonto{\varphi\Ra\psi}s_{n+1}$, we have to admit that $s_{n+1}$ satisfies $\psi$. Therefore, we have to reject the case that $\neg\psi \in L(s_{n+1})$ (as otherwise, since $|\neg\psi| <|\xi|$, $s_{n+1}$ would also satisfy $\neg\psi$) and admit that $\forall A^*.\vartheta$ is in $L(s_{n+1})$. Again, by \ref{alpha/beta}, we have that $\{\vartheta, \forall A.\forall A^*.\vartheta\}$ is a subset of $L(s_{n+1})$.
	Finally, in the case that $\forall\Omega^*.\forall A^*.\vartheta\in L(s_{n+1})$, by using \ref{alpha/beta}, we can reach the desired conclusion. 
	
	\item $A = A_1+A_2$: We know that $s_n\!\!\lonto{A_1+A_2}\!\!s_{n+1}$ and thus, either $s_n\onto{A_1}s_{n+1}$, or $s_n\onto{A_2}s_{n+1}$. In both cases, since $\forall(A_1\!+\!A_2).\forall A^*.\vartheta\in L(s_n)$, by \ref{alpha/beta} of Definition \ref{def: Hintikka structure}, we have that $\forall A_1.\forall A^*.\vartheta, \forall A_2.\forall A^*.\vartheta\in L(s_n)$ and by the induction hypothesis of the structural induction, we have that $\forall A^*.\vartheta\in L(s_{n+1})$. By condition \ref{alpha/beta}, we have that $\vartheta, \forall A.\forall A^*.\vartheta\in L(s_{n+1})$.
	
	\item $A = A_1A_2$: Since $s_n\lonto{A_1A_2}s_{n+1}$, there is a state $s'$ such that $s_n\onto{A_1}s'\onto{A_2}s_{n+1}$. Furthermore, as $\forall A_1A_2.\forall A^*.\vartheta\in L(s_n)$, by condition \ref{alpha/beta} of Definition \ref{def: Hintikka structure}, we can conclude that $\forall A_1.\forall A_2.\forall A^*.\vartheta\in L(s_n)$. By the induction hypothesis of the structural induction, first, we have that $\forall A_2.\forall A^*.\vartheta\in L(s')$ and then, that $\forall A^*.\vartheta\in L(s_{n+1})$. By \ref{alpha/beta}, we conclude that $\vartheta, \forall A.\forall A^*.\vartheta\in L(s_{n+1})$.
	
	\item $A = B^*$: Since $s_n\onto{B^*}s_{n+1}$, there are transitions such that $w_0\onto{B}w_1\dotsb w_{m-1}\onto{B}w_m$ with $w_0 = s_n$, $w_m = s_{n+1}$ and $m\geq0$. If $m=0$, then $s_n= s_{n+1}$ and as a result, $\vartheta$, $\forall B^*.\forall A^*.\vartheta\in L(s_{n+1})$. In the case that $m>0$, we assume that $\vartheta$, $\forall B^*.\forall A^*.\vartheta\in L(w_{m-1})$. By condition \ref{alpha/beta} of Definition \ref{def: Hintikka structure}, we have that $\forall B.\forall B^*.\forall A^*.\vartheta\in L(w_{m-1})$. Now, by the induction hypothesis of the structural induction, we can conclude that $\forall B^*.\forall A^*.\vartheta\in L(w_m)$. Again, by condition \ref{alpha/beta}, first, we have that $\forall A^*.\vartheta\in L(w_m)$ and then, that $\vartheta, \forall A.\forall A^*.\vartheta\in L(w_m)$.
\end{enumerate}

We have shown that if $\forall A^*.\vartheta \in L(s)$, then $\forall s' \in S \big(s\onto{A^*}s' \Ra \vartheta \in L(s')\big)$. As $|\vartheta|<|\forall A^*.\vartheta|$, by the \ref{IH}, we know that $\forall s' \in S \big(s\onto{A^*}s' \Ra s'\models^\mc{M(H)}\vartheta\big)$. Thus, by semantics, we have that $s\models^\mc{M(H)}\forall A^*.\vartheta$.

Let $\xi$ be an $\alpha$/$\beta$ eventuality $\neg\forall A_1.\dotsb\forall A_k.\chi$ such that $k\geq1$ and $A_k=A^*$ and $\neg\chi\notin\ev$. By \ref{ev} of Definition \ref{def: Hintikka structure}, we know that there is a structure path $\stp_0 (s,\xi, s',\neg\chi) = (s_1,\varphi_1),\dotsc,(s_n,\varphi_n)$ in $\mc{H}$ with $n\geq2$ and $s'\in S$. Without loss of generality, we assume that $\neg\chi$ appears only in the last pair of the structure path. 
By Lemma \ref{lem: properties of vtrd}, it follows that for $i=1,\dotsc,n-1$, the formula $\varphi_i$ is an eventuality of the form $\neg\forall B_1.\dotsb\forall B_r.\forall A^*.\chi$ with $r\geq0$.

Having in mind the sequence $\stp_0 (s,\xi, s',\neg\chi) = (s_1,\varphi_1),\dotsc,(s_n,\varphi_n)$, 
we define an appropriate process $\delta_m= B_1\dotsb B_m \in (\atp_\Omega \cup \mc{L}_s)^*$, with $1\leq m\leq n$.
Let $B_1 = \varepsilon$ (thus, $\delta_1=\varepsilon$) and for $i=2,\dotsc,m$, 
\begin{enumerate*}[i)]
	\item if $\varphi_{i-1}$ is a formula of the form $\neg\forall A.\chi$ with $A\in\atp_\Omega$, then $B_i = A$, and 
	\item if $\varphi_{i-1}\!\vtrd\!\varphi_i$ and $R$ is the reduction set of the $\alpha$/$\beta$ eventuality $\varphi_{i-1}$ such that $\varphi_i\!\in R$ and $R\subseteq L(s_i)$ and $(\mc{T},\psi)$ is the pair in $\fd(\varphi_{i-1})$ such that  $R=\mc{T}\!\cup\!\{\psi\}$, 
	then $B_i$ is the formula $\bigwedge\!\mc{T}$, i.e. the conjunction of the formulas in $\mc{T}$.
\end{enumerate*}

We show that for $i=1,\dotsc,n$,
\begin{enumerate}[i)]
\item if $i\neq n$, then $\varphi_i$ is an eventuality of the form $\neg\forall B_1.\dotsb\forall B_l.\chi$ such that $l\geq 1$ and for $j=1,\dotsc,l$, $\siz{B_j}<|\xi|$,
\item if $\varphi_i=\neg\forall C_1.\dotsb\forall C_r.\chi$ with $r\geq0$, then for each process of the language $\{\delta_i\} \mc{L}(C_1\dotsb C_r)$, there is an interchangeable process in $\mc{L}(A_1\dotsb A_k)$ (in the sense of Definition \ref{lang}), and 
\item there is the transition $s_1 \onto{\delta_i} s_i$ in $\mc{M(H)}$.
\end{enumerate}

We proceed by induction on $i$. For $i=1$, since $\varphi_1 = \xi$, it is clear that all the necessary properties are satisfied. Furthermore, since $\delta_1=\varepsilon$, it is immediate that the language $\{\varepsilon\}\mc{L}(A_1\dotsb A_k)$ is a subset of $\mc{L}(A_1\dotsb A_k)$. Finally, since $\varepsilon$ is essentially the trivial test $\true$, we have that $s_1\onto{\delta_1}s_1$.

Next, we assume that the necessary properties hold for $i$ and we show that they also hold for $i+1$.
We distinguish cases on $\varphi_i= \neg\forall B_1.\dotsb\forall B_r.\chi$. First, we assume that $B_1$ is an atomic process $b$.
By the definition of a structure path, we have that $\varphi_{i+1}=\neg\forall B_2.\dotsb\forall B_r.\chi$ with $r\geq0$. Thus, by the induction hypothesis for $\varphi_i$, we have to admit that all the necessary properties also hold for $\varphi_{i+1}$.
Notice that $\delta_{i+1} = \delta_ib$. Since, by the induction hypothesis, for each process of the language $\{\delta_i\} \mc{L}(bB_2\dotsb B_r)$, there is an interchangeable process in $\mc{L}(A_1\dotsb A_k)$, it is clear that the same also holds for the processes of $\{\delta_ib\} \mc{L}(B_2\dotsb B_r)$ (as by Definition \ref{lang} $\{\delta_ib\} \mc{L}(B_2\dotsb B_r) = \{\delta_i\} \mc{L}(bB_2\dotsb B_r)$). 
Since by the induction hypothesis $s_1 \onto{\delta_i} s_i$ and by the definition of a structure path, $s_i\leads{b}s_{i+1}$, we can conclude that $s_i\onto{b}s_{i+1}$ (by the definition of $\mc{M(H)}$), as well as that $s_1\onto{\delta_ib}s_{i+1}$.
We can argue for the case $B_1=\Omega$ in a similar way. Notice that, by the definition of $\mc{M(H)}$, the relation $\leads{\Omega}$ is a subset of $\onto{\Omega}$.

Now, we assume that $\varphi_i= \neg\forall B_1.\dotsb\forall B_r.\chi$ is an $\alpha$/$\beta$ formula. By the definition of a structure path, we have that $\varphi_i\vtrd\varphi_{i+1}$ and $s_{i+1}=s_i$. Moreover, there is a reduction set $R$ of $\varphi_i$ such that $\varphi_{i+1}\in R$ and $R\subseteq L(s_{i+1})$. 
According to Definitions \ref{def: reduction sets} and \ref{def: decomposition}, since $\varphi_i$ is an eventuality, there is a series of applications of the decomposition rules of Definition \ref{def: decomposition} which leads to the definition of $R$. The decomposition rules correspond to specific $\alpha$ and $\beta$ rules of Table \ref{tab: alpha/beta}. We distinguish cases on which decomposition rule is applied each time and we show that the necessary properties hold.

First, we assume that $B_1=\varphi\Ra\vartheta$ and as a result, the next formula that occurs is either the formula $\psi=\neg\forall(\neg\varphi).\forall\Omega. \forall B_2.\dotsb\allowbreak\forall B_r.\chi$, or the formula $\psi=\neg\forall\Omega.\forall\vartheta.\forall B_2.\dotsb\forall B_r.\chi$. In both cases, by the induction hypothesis for $\varphi_i$, the necessary properties also hold for $\psi$.
Now, since $\mc{L}(\varphi\Ra\vartheta) = \{\neg\varphi\Omega,\Omega\vartheta\}$, the languages $\mc{L}(\neg\varphi\Omega)$ and $\mc{L}(\Omega\vartheta)$ are subsets of $\mc{L}(\varphi\Ra\vartheta)$. Thus, the languages $\{\delta_i\} \mc{L}(\neg\varphi\Omega B_2\dotsb B_r)$ and  $\{\delta_i\} \mc{L}(\Omega\vartheta B_2\!\dotsb B_r)$ are subsets of $\{\delta_i\} \mc{L}(B_1B_2\!\dotsb B_r)$. So, by the induction hypothesis for $\{\delta_i\} \mc{L}(B_1B_2\!\dotsb B_r)$, the necessary property is also satisfied for $\{\delta_i\} \mc{L}(\neg\varphi\Omega B_2\dotsb B_r)$ and for $\{\delta_i\} \mc{L}(\Omega\vartheta B_2\dotsb B_r)$.
Finally, by the induction hypothesis, we know that $s_1\onto{\delta}s_i$. Since $s_{i+1}=s_i$, we have that $s_1\onto{\delta}s_{i+1}$.

In a similar way, we can examine the rest cases. We leave the details to the interested reader. So, no matter which decomposition rule is applied for the definition of the reduction set of $\varphi_i$ the necessary properties are satisfied.

Recall that we need to show that $\xi=\neg\forall A_1.\dotsb\forall A_k.\chi$ is satisfied by the state $s=s_1$. By the properties that we have shown earlier, we are interested in the case that $i=n$. Since $\varphi_n=\neg\chi$, there is a process $\delta_n$ such that $s\onto{\delta_n}s'$ (recall that $s'=s_n$) and there is an interchangeable process $B$ in $\mc{L}(A_1\dotsb A_k)$ (in the sense of Definition \ref{lang}). By Lemma \ref{lem: lang properties}, we know that $\onto{\delta_n}\,=\onto{B}$ and thus, $s\onto{B}s'$. Again by Lemma \ref{lem: lang properties}, we know that  $\lonto{A_1\dotsb A_k}$ is the union of all the relations which correspond to the processes that belong to the language $\mc{L}(A_1\dotsb A_k)$. Thus, since $B$ is in $\mc{L}(A_1\dotsb A_k)$ and $s\onto{B}s'$, we can conclude that $s\lonto{A_1\dotsb A_k}s'$. 
Moreover, by the structure path $\stp_0 (s,\xi, s',\neg\chi)$, we know that $\neg\chi$ is in $L(s')$. Therefore, by the \ref{IH} ($|\neg\chi| < |\neg\forall A_1.\dotsb\forall A_k.\chi|$), we know that $s'\models^\mc{M(H)} \!\neg\chi$. Finally, by semantics, we have that $s \models^\mc{M(H)} \neg\forall A_1.\dotsb\forall A_k.\chi$.

Now, we consider the case that $\xi=\forall a. \vartheta$.
By \ref{box} of Definition \ref{def: Hintikka structure}, we know that for any state $s'\in S$, if $s\leads{a}s'$ then $\vartheta\in L(s')$.
Now, by the definition of $\mc{M(H)}$ (i.e. we know that $\onto{a}\,=\leads{a}$) and by the \ref{IH} (as $|\vartheta|<|\forall a.\vartheta|$), we have that $\forall s'\in S(s\onto{a}s' \Ra s'\models^{\mc{M(H)}}\!\vartheta)$. Consequently, by semantics, we have that $s\models^{\mc{M(H)}}\forall a.\vartheta$.

In the case that $\xi=\forall \Omega. \vartheta$, by \ref{Omega} of Definition \ref{def: Hintikka structure}, we know that $\vartheta$ is a formula of the form $\forall\Omega^*.\varphi$ and that it is in $L(s)$. By Lemma \ref{lem: Omega}, it is immediate that $\onto{\Omega}\onto{\Omega^*}\,=\onto{\Omega^*}$. So, $\forall\Omega.\forall\Omega^*.\varphi$ is satisfiable iff $\forall\Omega^*.\varphi$ is satisfiable. Recall that we have already argued for the case $\forall A^*.\vartheta$ which includes the case $A=\Omega$.

Now, we examine the case that $\xi=\neg\forall A.\vartheta$ such that $A\in\atp_\Omega$. 
By \ref{diamond}, we know that there is a state $s'\in S$ such that $s\leads{A}s'$ and $\neg\vartheta\in L(s')$. 
By Definition \ref{def: Hintikka model}, we know that $\leads{\Omega}\, \subseteq\onto{\Omega}$ and $\forall a\in\atp$, $\onto{a}\,=\leads{a}$. Thus, we have that $s\onto{A}s'$. Since $|\neg\vartheta|<|\neg\forall A.\vartheta|$, by the \ref{IH}, we have that $s'\models^{\mc{M(H)}}\!\!\neg\vartheta$. 
So, by semantics, we conclude that $s\models^{\mc{M(H)}}\neg\forall A.\vartheta$.

If $\xi$ is a formula of the form $\forall(\varphi\Ra\psi).\vartheta$,
according to Definition \ref{def: reduction sets}, its reduction sets are the sets $\{ \varphi, \forall\Omega^*.\forall\psi.\vartheta \}$ and $\{ \forall\Omega^*.\vartheta \}$. Hence, according to \ref{alpha/beta} of Definition \ref{def: Hintikka structure}, either $\{ \varphi, \forall\Omega^*.\forall\psi.\vartheta \}$ is a subset of $L(s)$, or $\{ \forall\Omega^*.\vartheta \}$ is a subset of $L(s)$.
In the first case, since $|\varphi| < |\forall(\varphi\Ra\psi).\vartheta|$, by the induction hypothesis, $s$ satisfies $\varphi$. Moreover, notice that $|\forall\psi.\vartheta| < |\forall(\varphi\Ra\psi).\vartheta|$ and that we have already argued for the more general case $\forall A^*.\vartheta$. Therefore, we can conclude that $s$ satisfies the set $\{ \varphi, \forall\Omega^*.\forall\psi.\vartheta \}$ and by Proposition \ref{prop: equiv alpha beta}, it also satisfies $\forall(\varphi\Ra\psi).\vartheta$.
In the second case, since we have already argued for the more general case $\forall A^*.\vartheta$, we can conclude that $s$ satisfies $\forall\Omega^*.\vartheta$ and again, by Proposition \ref{prop: equiv alpha beta}, it also satisfies $\forall(\varphi\Ra\psi).\vartheta$.

We assume that $\xi$ is the $\beta$ formula $\neg\forall(\varphi\Ra\psi).\vartheta$ and that it is not an eventuality. Thus, $\neg\vartheta$ cannot be an eventuality as well. According to Definition \ref{def: reduction sets}, the reduction sets of $\xi$ are the singletons $\{ \neg\forall(\neg\varphi).\forall\Omega.\vartheta \}$ and $\{ \neg\forall\Omega.\forall\psi.\vartheta \}$. Hence, according to \ref{alpha/beta} of Definition \ref{def: Hintikka structure}, either $\{ \neg\forall(\neg\varphi).\forall\Omega.\vartheta \}$ is a subset of $L(s)$, or $\{ \neg\forall\Omega.\forall\psi.\vartheta \}$ is a subset of $L(s)$.

In the first case, by \ref{alpha/beta}, we know that $\neg\varphi\in L(s)$ and $\neg\forall\Omega.\vartheta \in L(s)$. Now, by \ref{diamond}, we know that there is a state $s'\in S$ such that $s\leads{\Omega}s'$ and $\neg\vartheta\in L(s')$.
As $|\neg\varphi|, |\neg\vartheta|<|\xi|$, by the \ref{IH}, we know that $s\models^{\mc{M(H)}} \!\!\neg\varphi$ and $s' \models^{\mc{M(H)}} \!\!\neg\vartheta$. So, as by Definition \ref{def: Hintikka model}, $\leads{\Omega}\, \subseteq\onto{\Omega}$, by Corollary \ref{cor: =>2}, we conclude that $s\models^{\mc{M(H)}}\!\!\neg\forall(\varphi\Ra\psi).\vartheta$.
We can argue similarly for the second case.

We assume that $\xi$ is an $\alpha$ formula of one of the following forms, but not an eventuality: $\neg\neg\vartheta$, $\neg\forall\psi.\vartheta$, $\forall A_1A_2.\vartheta$, $\neg\forall A_1A_2.\vartheta$ and $\forall(A_1\!+\!A_2).\vartheta$. Since $\alpha\in L(s)$, by Definition \ref{def: Hintikka structure} (condition \ref{alpha/beta}), we have that $\{\alpha_1,\alpha_2\} \subseteq L(s)$. Now, by the induction hypothesis, we have that $s\models^{\mc{M(H)}}\alpha_1$ and $s\models^{\mc{M(H)}}\alpha_2$ and as a result, by Proposition \ref{prop: equiv alpha beta}, $s\models^{\mc{M(H)}}\alpha$.
Now, we assume that $\xi$ is a $\beta$ formula of one of the following forms, but not an eventuality: $\forall \psi.\vartheta$ and $\neg\forall(A_1\!+\!A_2).\vartheta$. Since $\beta\in L(s)$, by Definition \ref{def: Hintikka structure} (condition \ref{alpha/beta}), we have that $\beta_1\in L(s)$ or $\beta_2\in L(s)$. Now, by the induction hypothesis, we have that $s\models^{\mc{M(H)}}\beta_1$ or $s\models^{\mc{M(H)}}\beta_2$ and as a result, by Proposition \ref{prop: equiv alpha beta}, $s\models^{\mc{M(H)}}\beta$.

Now, we examine all the forms of a capability statement. First, observe that, by semantics, any statement $\cpi\varphi$ is logically valid. By \ref{cap1} of Definition \ref{def: Hintikka structure}, $\xi$ cannot be a statement $\neg\cpi\varphi$. In the case that $\xi = \cpi A$ with $A\in\wt{\Sigma}$, since $\cpi A$ is in $L(s)$, by the definition of $\imath^\mc{M(H)}(s)$,  we know that $\onto{A} \,\subseteq \imath^\mc{M(H)}(s)$. Therefore, it is immediate that $s \models^\mc{M(H)} \cpi A$.

The most interesting case is when $\xi = \neg\cpi A \in L(s)$ with $A\in\wt{\Sigma}$. By semantics, in order for $s$ to satisfy $\neg\cpi A$, the relation $\onto{A}$ must not be a subset of $\imath^\mc{M(H)}(s)$. 
So, we need to show that $\onto{A} \,\neq \emptyset$ (obviously, the empty set is a subset of $\imath^\mc{M(H)}(s)$) and by the definition of $\imath^\mc{M(H)}(s)$,  that $\onto{A}$ is not a subset of  $\bigcup\{ \onto{A'} \mid \cpi A'\in L(s) \text{ and } A'\in\wt{\Sigma} \}$. In other words, we need to show that there is a transition $s_1\onto{A}s_2$ such that there is no positive capability statement $\cpi A'$ with $A'\in\wt{\Sigma}$ in $L(s)$ such that $s_1\onto{A'}s_2$. Observe that by condition \ref{contr}, $\cpi A$ cannot be in $L(s)$.

Let $\Gamma = \{ \cpi(\varphi_1\Ra\psi_1),\dotsc,\cpi(\varphi_k\Ra\psi_k) \}$ be the set of all the statements of the form $\cpi(\varphi\Ra\psi)$ in $L(s)$. So, what we actually need to show is that $\neg\cpi A$ is satisfiable in conjunction with the set $\Gamma$.
By condition \ref{cap2}, we know that there is a state $s_1\in S$ such that the set $\{ \varphi_1,\dotsc,\varphi_k, \neg\forall A.\forall(\neg\psi_1).\dotsb\forall(\neg\psi_k).\false \}$ is a subset of $L(s_1)$. Observe that by the induction hypothesis, $s_1$ satisfies the preconditions $\varphi_1,\dotsc,\varphi_k$.

Now, we distinguish cases on $A\in\wt{\Sigma}$. First, we assume that $A$ is an atomic process. In this case, by condition \ref{diamond}, we know that there is a state $s_2\in S$ such that $s_1\leads{A}s_2$ and $\neg\forall(\neg\psi_1).\dotsb\forall(\neg\psi_k).\false \in L(s_2)$. Now, by condition \ref{alpha/beta}, we can conclude that $\{ \neg\psi_1,\dotsc,\neg\psi_k,\neg\false \}$ is a subset of $L(s_2)$. By the induction hypothesis, we know that $s_2$ does not satisfy the effects $\psi_1,\dotsc,\psi_k$. 
Since $A$ is an atomic process, by Definition \ref{def: Hintikka model}, we know that $\onto{A}\,=\,\leads{A}$. Since all the preconditions $\varphi_1,\dotsc,\varphi_k$ hold at $s_1$, but at the same time the effects $\psi_1,\dotsc,\psi_k$ do not hold at $s_2$, by Corollary \ref{cor: =>2}, we can conclude that for $i=1,\dotsc,k$, the transition $s_1\onto{A}s_2$ is not of type $\varphi_i\Ra\psi_i$ (i.e. it does not hold that $s_1\lonto{\varphi_i\Ra\psi_i}s_2$).

Now, assume that there is an atomic process $a$ different from $A$ such that $\cpi a\in L(s)$ and $s_1\onto{a}s_2$. By the definition of $\mc{M(H)}$, we can conclude that $s_1\leads{a}s_2$. So, there are the transitions $s_1\leads{a}s_2$ and $s_1\leads{A}s_2$ such that $A\neq a$. However, by \ref{transitions}, this cannot be true.
Therefore, there is no statement $\cpi A'$ with $A'\in\wt{\Sigma}$ in $L(s)$ such that $s_1\onto{A'}s_2$.

In the case that $A$ is a process $\varphi\Ra\psi$, by condition \ref{alpha/beta}, either  $\neg\forall(\neg\varphi).\forall\Omega. \forall(\neg\psi_1).\dotsb\forall(\neg\psi_k).\false \in L(s_1)$ or $\neg\forall\Omega.\forall\psi. \forall(\neg\psi_1).\dotsb\forall(\neg\psi_k).\false \in L(s_1)$.
So, in the first case, by \ref{alpha/beta}, we have that $\neg\varphi$ and $\neg\forall\Omega. \forall(\neg\psi_1).\dotsb\forall(\neg\psi_k).\false$ are in $L(s_1)$. By the induction hypothesis, we have that $s_1$ satisfies $\neg\varphi$. By \ref{diamond}, we know that there is a state $s_2\in S$ such that $s_1\leads{\Omega}s_2$ and $\neg\forall(\neg\psi_1).\dotsb\forall(\neg\psi_k).\false$ is in $L(s_2)$. By condition \ref{alpha/beta} and the induction hypothesis, we can conclude that $s_2$ does not satisfy the effects $\psi_1,\dotsc,\psi_k$ whose negations are in $L(s_2)$.
In the second case, similarly to the first, there is a state $s_2\in S$ such that $s_1\leads{\Omega}s_2$ and $s_2$ satisfies $\psi$ which is in $L(s_2)$, but it does not satisfy the effects $\psi_1,\dotsc,\psi_k$ whose negations are in $L(s_2)$.
In both cases, by the definition of $\mc{M(H)}$, we have that $s_1\onto{\Omega}s_2$ (notice that $\leads{\Omega}$ is a subset of $\onto{\Omega}$) and $s_1\lonto{\varphi\Ra\psi}s_2$, as $s_1\models^\mc{M(H)} \!\!\neg\varphi$ or $s_2\models^\mc{M(H)} \!\psi$ (see Corollary \ref{cor: =>2}). Moreover, since all the preconditions $\varphi_1,\dotsc,\varphi_k$ hold at $s_1$, but at the same time the effects $\psi_1,\dotsc,\psi_k$ do not hold at $s_2$, we can conclude that for $i=1,\dotsc,k$, the transition $s_1\onto{A}s_2$ is not of type $\varphi_i\Ra\psi_i$.
Similarly to the case of an atomic process $A$, we can assume that there is an atomic process $a$ such that $\cpi a\in L(s)$ and $s_1\onto{a}s_2$ and reach a contradiction by using \ref{transitions}. Consequently, we have shown that there is no statement $\cpi A'$ with $A'\in\wt{\Sigma}$ in $L(s)$ such that $s_1\onto{A'}s_2$.

Finally, the remaining cases for the capabilities statements can be shown similarly to previous cases of this proof.
\end{proof}

\section{The Tableau-based Satisfiability Algorithm} 		\label{sec: algorithm}

Here, we present the algorithm which decides the satisfiability of \tpdl formulas. 
More specifically, in Subsection \ref{subsec: tableau calculus}, we present the necessary tableau rules that the algorithm uses to expand the nodes of a tableau, in Subsection \ref{subsec: algorithm}, we present how our algorithm constructs a tableau and in Subsection \ref{subsec: examples}, we give examples of tableaux.

\subsection{A Tableau Calculus for Type PDL}	\label{subsec: tableau calculus}

In order to give a tableau-based satisfiability algorithm, we need to define the appropriate tableau calculus, always according to the semantics of the system (see \cite{Gore-Handbook} for a review on the tableaux calculi for modal logics). We present the tableau rules as transformation functions of a set of formulas to other sets of formulas. In the satisfiability algorithm, a tableau node is not just a set of formulas. Nevertheless, we choose to present the rules here independently of the way that the other attributes of a tableau node are defined, in order to abstract at this point from the corresponding details.

A \textit{tableau calculus} is defined as a finite collection of tableau rules.
A \textit{tableau rule} $\rho$ is a function $\rho(M) \subseteq 2^{\mc{L}_s}$ which accepts as input a set $M \subseteq \mc{L}_s$ of formulas, where one or more formulas are distinguished as the \textit{principal formulas}, and delivers a finite but nonempty set $\{M_1, \dotsc, M_k\}$ of sets of formulas. 
The set $M$ is usually called  the premise, or the numerator of the rule, while the sets $M_1, \dotsc, M_k$ are called  the conclusions, or the denominators of the rule.
The tableau rules are usually separated in two types:
\begin{enumerate*}[i)]
	\item or-rule:	$(\rho_\lor) \;\;\;\dfrac{M}{M_1 \mid M_2 \mid \dotsb \mid M_k}$	
	which is \textit{sound} iff if its premise is satisfiable then at least one of its conclusions is also satisfiable,
	
	\item and-rule:
	$(\rho_\land) \;\;\;\dfrac{M}{M_1 \;\&\;  M_2 \;\&\;  \dotsb \;\&\;  M_k}$
	which is \textit{sound} iff if its premise is satisfiable then all of its conclusions are also satisfiable.
\end{enumerate*}

\paragraph{The Static Rule.}
Due to Definition \ref{def: reduction sets}, all the $\alpha$ and $\beta$ rules of Table \ref{tab: alpha/beta} can be summarized in the following single rule:
\begin{equation}	\label{static rule} 
(static)\;
\dfrac
{\varphi}
{\mc{R}^\varphi_1 \mid \dotsb \mid \mc{R}^\varphi_k}
\quad
\text{where $\varphi$ is an $\alpha$/$\beta$ formula and $k\geq1$ is its reduction degree}
\end{equation}

\begin{prop} \label{prop: static sound}
	The unique static rule is sound.
\end{prop}
\begin{proof}
	We assume that the $\alpha$/$\beta$ premise $\varphi$ is satisfied by a state $w$ of an $\mc{L}$-model $\mc{M}$ and we distinguish cases on whether $\varphi$ is an eventuality.
	If $\varphi$ is a non-eventuality, then according to Definition \ref{def: reduction sets} the static rule is nothing more than one of the $\alpha$/$\beta$ rules of Table \ref{tab: alpha/beta}. Thus, by Proposition \ref{prop: equiv alpha beta}, at least one of the reduction sets of $\varphi$ is also satisfiable.
	Now, in the case that $\varphi$ is an eventuality, since $\varphi$ is satisfied by $w$ of $\mc{M}$, by Lemma \ref{lem: fpaths}, we know that there is a fulfilling path $\fpath(\mc{M},w,\varphi,\vtrd)$. By its properties, since $\varphi$ is an $\alpha$/$\beta$ formula, the sequence $\fpath(\mc{M},w,\varphi,\vtrd)$ begins with two pairs $(w,\varphi),(w,\psi)$ such that $\varphi\vtrd\psi$ and there is a reduction set $\mc{R}$ of $\varphi$ such that $\psi\in\mc{R}$ and $w\models^\mc{M}\mc{R}$.
\end{proof}

\paragraph{The Transitional Rule.}
For a set $\Gamma$ of formulas, let $\Gamma_\diamondsuit$ consist of all the formulas in $\Gamma$ of the form $\neg\forall A.\chi$ such that $A$ is either an atomic program term $a$, or the program term $\Omega$.
If $\Gamma_\diamondsuit$ is the set $\{ \neg\forall A_1.\chi_1, \dotsc, \neg\forall A_n.\chi_n \}$, then for each $\neg\forall A_i.\chi_i$ in $\Gamma_\diamondsuit$ where $1\leq i\leq n$, define a set $\Delta_i$ of formulas by distinguishing cases as follows:
\begin{align}
\text{if } A_i \text{ is some } a\in\atp, \text{ then } 
&\Delta_i = \{\neg\chi_i\} \cup \{\psi \mid \forall a.\psi\in\Gamma\}  \cup  \{\psi \mid \forall\Omega.\psi\in\Gamma\}\\
\text{if } A_i = \Omega, \text{ then } 
&\Delta_i = \{\neg\chi_i\} \cup \{\psi \mid \forall\Omega.\psi\in\Gamma\}
\end{align}
The transitional rule can then be stated as an and-rule:
\begin{equation}		\label{trans}
(trans)\;\dfrac{\Gamma}{\Delta_1 \;\&\;  \dotsb \;\&\; \Delta_n}
\end{equation}
If $\Gamma_\diamondsuit = \emptyset$, then the transitional rule cannot be applied.

\begin{prop}	\label{prop: trans sound}
	The $(trans)$-rule is sound.
\end{prop}
\begin{proof}
	Suppose that the set $\Gamma$ is satisfied by a state $w$ of an $\mc{L}$-model $\mc{M}$. Moreover, let $\neg\forall A.\chi$ be a formula in $\Gamma$ such that $A$ is either an atomic program $a$, or the program term $\Omega$. By semantics, since $w\models^\mc{M}\neg\forall A.\chi$, we can conclude that there is a state $w'$ of $\mc{M}$, such that $w\onto{A}w'$ and $w'\models^\mc{M} \neg\chi$. Furthermore, for each formula $\forall A.\psi$ of $\Gamma$, since $w\models^\mc{M}\forall A.\psi$ and $w\onto{A}w'$, it is immediate that $w'\models^\mc{M} \psi$.
	More specifically, in the case that $A=a$, for each formula $\forall\Omega.\psi$ of $\Gamma$, since $w\models^\mc{M}\forall\Omega.\psi$ and $w\onto{a}w'$ and $a\vmodels\Omega$ (recall Lemma \ref{lem: Omega}), we have again that $w'\models^\mc{M} \psi$.
\end{proof}

\paragraph{The Capability Rule.}

The $\alpha$/$\beta$ rules decompose the processes of the capabilities statements until processes of $\wt{\Sigma}$ are defined. This is not enough. For example, in the case of a negated capability statement $\neg\cpi A$ with $A\in\wt{\Sigma}$ along with a non-negated statement $\cpi B$ with $B\in \wt{\Sigma}$, the algorithm should make sure that the interpretation of $A$ is not a subrelation of that of $B$, as otherwise the agent $\imath$ should also have the capability to execute $A$ besides $B$. 

For a set $\Gamma$ of formulas, let $\Gamma_{N\mc{C}}$ consist of all the formulas in $\Gamma$ of the form $\neg\cpi A$ such that $A$ is in $\wt{\Sigma}$ and
for some agent $\imath\in I$, let $\Gamma_{\mc{C}_\imath\Ra}$ be the set of all the non-negated capabilities statements in $\Gamma$ of the form $\cpi(\varphi\Ra\psi)$.
Now, let $\Gamma_{N\mc{C}}$ be the set $\{  \neg\cp_{\imath_1} A_1, \dotsc, \neg\cp_{\imath_n} A_n \}$. So, for each formula $\neg\cp_{\imath_i} A_i$ where $1\leq i\leq n$, if $\Gamma_{\mc{C}_{\imath_i}\Ra}$ is the set $\{ \cp_{\imath_i}(\varphi_1\Ra\psi_1), \dotsc, \cp_{\imath_i}(\varphi_k\Ra\psi_k) \}$, then define a set $\Delta_i$ of formulas as follows:
\begin{equation}	\label{eq: cap Delta}
\Delta_i = \{ \varphi_1,\dotsc,\varphi_k, \neg\forall A_i.\forall(\neg\psi_1).\dotsb\forall(\neg\psi_k).\false \}
\end{equation}
In the case that $\Gamma_{\mc{C}_{\imath_i}\Ra}$ is empty, $k=0$. Thus, we use $\false$ to have a syntactically correct formula, (i.e. $\neg\forall A_i.\false$). Of course, $\neg\false=\true$ holds in any state of any $\mc{L}$-model.
The capability rule is stated as an and-rule:
\begin{equation}		\label{eq: capability rule}
(cap)\;\dfrac{\Gamma}{\Delta_1 \;\&\;  \dotsb \;\&\; \Delta_n}
\end{equation}
If $\Gamma_{N\mc{C}}=\emptyset$, then the capability rule cannot be applied.

\begin{prop}	\label{prop: cap sound}
	The $(cap)$-rule is sound.
\end{prop}
\begin{proof}
	Suppose that the numerator $\Gamma$ is satisfied by a state $w$ of an $\mc{L}$-model $\mc{M}$. We consider a formula $\neg\cpi A \in \Gamma$ and the set $\Gamma_{\mc{C}_\imath\Ra} = \{ \cpi(\varphi_1\!\Ra\!\psi_1),\dotsc,\cpi(\varphi_k\!\Ra\!\psi_k) \}$. 
	We show that the set $\{ \varphi_1,\dotsc,\varphi_k, \neg\forall A.\forall(\neg\psi_1).\dotsb\allowbreak\forall(\neg\psi_k).\false \}$ is satisfiable.
	So, by semantics, we know that for $i=1,\dotsc,k$, $\lonto{\varphi_i\Ra\psi_i} \,\subseteq\, \imath^\mc{M}(w)$. Moreover, since $A$ is in $\wt{\Sigma}$, it is necessarily true that $\onto{A}$ is not a subset of 
	$\bigcup\limits_{1\leq i\leq k}\lonto{\varphi_i\Ra\psi_i}$,
	as otherwise $w$ would also satisfy $\cpi A$. Additionally, since $w\models^\mc{M}\neg\cpi A$, by Lemma \ref{lem: cap}, we have that $\onto{A}\,\neq\emptyset$.
	Hence, there is a transition $w_1\onto{A}w_2$ which cannot be in the interpretations of the processes $\varphi_1\Ra\psi_1,\dotsc,\varphi_k\Ra\psi_k$. In other words, by Corollary \ref{cor: =>2}, $w_1$ satisfies the preconditions $\varphi_1,\dotsc,\varphi_k$, but $w_2$ does not satisfy the effects $\psi_1,\dotsc,\psi_k$. 
	It follows that $w_1$ satisfies the set $\{ \varphi_1,\dotsc,\varphi_k, \neg\forall A. \forall(\neg\psi_1).\dotsb\forall(\neg\psi_k).\false \}$.
\end{proof}

\subsection{The Algorithm} \label{subsec: algorithm}

\subsubsection{Basic Definitions} \label{subsubsec: basic definitions}

The following technical definitions facilitate the presentation of the algorithm.

\begin{defi}	 		\label{def: label} 
A \textit{label} $\lab$ is a pair $\pos{\Phi, \rd}$, where $\Phi$ is a set of formulas and $\rd$ is a reduction function which assigns to the $\alpha$/$\beta$ eventualities of $\Phi$ one of the values $1$ and $0$
such that for each $\alpha$/$\beta$ eventuality $\varphi\in\Phi$, if $\rd(\varphi)=1$, then at least one of the reduction sets of $\varphi$ is a subset of $\Phi$. 

The \textit{active set} $\actf$ of a label $\lab=\pos{\Phi, \rd}$ is the subset of $\Phi$ which consists of
\begin{enumerate*}[i)]
\item the formulas of the form 
	$p$, $\neg p$, 
	$\forall A.\psi$, $\neg\forall A.\psi$,
	$\cpi B$, $\neg\cpi B$, 
		where $A\in\atp_\Omega$ and $B\in\Sigma$,
\item the $\alpha$/$\beta$ non-eventualities for which none of their reduction sets is a subset of $\Phi$ and
\item the $\alpha$/$\beta$ eventualities $\varphi$ such that $\rd(\varphi)=0$.
\end{enumerate*}

The \textit{reduced set} $\rdf$ of a label $\lab=\pos{\Phi, \rd}$ consists of the $\alpha$/$\beta$ formulas $\varphi$ of $\Phi$ such that:
\begin{enumerate*}[i)]
\item if $\varphi$ is not an eventuality, then at least one of its reduction sets is a subset of $\Phi$, and
\item if $\varphi$ is an eventuality, then $\rd(\varphi)=1$.
\end{enumerate*}

The active and the reduced set of a label form a partition of its set of formulas. A \textit{partial label} is a label such that its active set contains $\alpha$/$\beta$ formulas. A \textit{saturated label} is a label which is not a partial one.
\end{defi}

\begin{defi} 			\label{def: similar labels}
We say that two labels $\lab_1$, $\lab_2$ are \textit{similar} and we write $\lab_1 \approx \lab_2$ iff
their active sets are equal (thus, both of them are either partial labels or saturated ones) and 
if they are partial labels, then their reduced sets are also equal.
\end{defi}

\begin{defi} 			\label{def: fully reduced}
We say that a formula $\varphi$ is \textit{fully reduced} to a formula $\psi$ within a label $\lab$ and we write $\varphi\vtrd_\lab\psi$ iff $\varphi\vtrd\psi$ and there is a reduction set $\mc{R}$ of $\varphi$ such that $\psi\in\mc{R}$ and $\mc{R}$ is a subset of $\Phi_\lab$.
\end{defi}

\begin{defi}	 		\label{def: reach} 
Let $\lab$ be a label  and $\varphi\in\Phi_\lab\cap\ev$ an eventuality.
The set $\reach(\varphi,\lab)$ is defined as the set of formulas $\psi\in\Phi_\lab$ such that there is a sequence of formulas $\psi_1,\dotsc,\psi_k$ of $\Phi_\lab$, with $k\geq 1$, such that 
\begin{enumerate}[i)]
\item $\psi_1=\varphi$ and $\psi_k=\psi$,
\item for $i=1,\dotsc,k-1$, $\psi_i\in\rdf_\lab\cap\ev$ (thus, $\psi_i\notin\actf_\lab$) and $\psi_i$ is fully reduced to $\psi_{i+1}$, i.e. $\psi_i \vtrd_\lab \psi_{i+1}$,
\item $\psi_k$ is either an active eventuality of $\Phi_\lab$, or a non-eventuality formula of $\Phi_\lab$.
\end{enumerate}
\end{defi}

\begin{defi}	\label{def: node}
A \textit{tableau node} or just \textit{node} $v$ of a rooted directed graph $G$ is a triple $\pos{\lab_v, D_v, \sts_v}$ where
$\lab_v$ is a label,
$D_v$ is a set of tableau nodes of $G$, called \textit{dependency set} and its elements dependency nodes, such that they are ancestors of $v$ (i.e. closer to the root) and 
$\sts_v$ is the \textit{status} of $v$ which is assigned one of the values \sat, \tempsat and \unsat.
We call a node \textit{partial} iff its label is partial and \textit{saturated} or \textit{state} iff its label is saturated.
\end{defi}

Intuitively speaking, the status of a node reveals the satisfiability status of the set of formulas of its label. The value \unsat indicates that the label cannot be satisfiable, while the value \sat indicates the opposite. 
Loops (i.e. graph cycles) are formed as usual due to the iteration operator, as the nodes of a tableau are reused. Dependencies are created in the sense that a node becomes dependent on ancestor nodes, as at the moment that a loop is formed, the latter nodes do not have the required information yet so that the status of a node can be determined.
Thus, the value \tempsat is used which indicates temporary satisfiability, as we cannot be certain yet whether a label is satisfiable. The ancestor nodes on which a label is dependent belong to the corresponding dependency set.
In the cases of the values \sat and \unsat, the label is not dependent on any node. 
At the end of the construction of a tableau, the satisfiabilitiy of all the nodes must have been determined and thus, their statuses are eventually assigned one of the values \sat and \unsat.

\begin{defi}	 		\label{def: tableau} 
A \emph{tableau} for a formula $\varphi$ is a directed graph $G = (V,E_f,E_b,E_c)$, where
\begin{itemize}
\item $V$ is a set of tableau nodes,
\item $E_f$, $E_b$ and $E_c$ are sets of forward, backward and cyclic edges, respectively, which may be labelled with some formula $\neg\forall A.\chi$ with $A\in\atp_\Omega$ or with some capability statement $\neg\cpi A$ with $A\in\wt{\Sigma}$,
\item its root is a node whose label is the pair $\pos{\{\varphi\},\rd}$, where $\rd$ is the reduction function which assigns to all the formulas of its domain the value $0$.
\end{itemize}
For simplicity of notation, we assume that $E$ is the set of all the edges of $G$ and we equivalently write $G=(V,E)$.
\end{defi}

\begin{defi}	 		\label{def: fulfillment relation} 
	A \emph{fulfillment relation} $\rat_G \subseteq (V\times\ev) \times (V\times\mc{L}_s)$ for a tableau $G=(V,E)$ relates 
	pairs consisting of a node and one of its label's active eventualities
	to pairs of the same form, with the difference that the involved formula can be any formula of the corresponding label, even a non-eventuality. In the following, we drop the index $G$ whenever $G$ is clear from the context.
\end{defi}

\begin{defi}	\label{def: fulfilled dependent}
	Let $G=(V,E)$ be a tableau and $\rat$ its fulfillment relation. 
	Let $v$ be a node of $G$ and $\varphi= \neg\forall A_1.\dotsb.\forall A_n.\forall A^*.\chi$ an active eventuality of $v$ with $n\geq0$ and $\neg\chi\notin\ev$.
	We say that
	\begin{enumerate}[i)]
		\item $\varphi$ is \textit{$\rat$-fulfilled} iff
		there is a sequence $(v_1,\varphi_1),\dotsc,(v_k,\varphi_k)$ of pairs of nodes and of active formulas of their labels with $k\geq2$ such that 
		no pair appears more than once in it and 
		$v_1=v$ and $\varphi_k = \neg\chi$ and
		for $i=2,\dotsc,k$, if $v_i\neq v$, then $\sts_{v_i} \in \{\sat, \tempsat\}$, and 
		for $i=1,\dotsc,k-1$, $\varphi_i\in\ev$ and $(v_i,\varphi_i)\rat(v_{i+1},\varphi_{i+1})$,
		
		\item $\varphi$ is \textit{$\rat$-dependent} on $v'$ iff
		there is a sequence $(v_1,\varphi_1),\dotsc,(v_k,\varphi_k)$ of pairs of nodes and of active eventualities of their labels with $k\geq2$ such that
		no pair appears more than once in it and $v_1=v$ and $v_k=v'$ and
		for $i=2,\dotsc,k-1$, if $v_i\neq v$, then $\sts_{v_i} = \tempsat$, and
		for $i=1,\dotsc,k-1$, $(v_i,\varphi_i)\rat(v_{i+1},\varphi_{i+1})$, and 
		the status of $v_k$ has not been defined yet and at the same time, $(v_k,\varphi_k)$ is not $\rat$-related to any other pair.
	\end{enumerate}
Finally, we say that $\varphi$ is \textit{$\rat$-unfulfilled} iff it is not $\rat$-fulfilled and it is not $\rat$-dependent on any node.
\end{defi}

Intuitively speaking, the fulfillment relation of a tableau records the fulfilling paths of the active eventualities of the nodes. If the value \sat has been assigned to the status of a node, then all of its active eventualities are $\rat$-fulfilled. On the other hand, if an active eventuality of a node is $\rat$-unfulfilled, then the status of a node is assigned the value \unsat. Finally, roughly speaking, if an active eventuality of a node is not $\rat$-fulfilled and it is $\rat$-dependent on a node $v$ and there are no $\rat$-unfulfilled eventualities, then the status is assigned the value \tempsat and it is natural to expect that $v$ belongs to the corresponding dependency set.

\begin{rem}[Notational and Other Conventions] \label{rem not}
	When we refer to the attributes of a label $\lab_v$ of a specific node $v$, we may write $\actf_v$, $\rdf_v$ and $\Phi_v$, without further disambiguation, unless they are explicitly denoted otherwise. Moreover, if $\varphi$ and $\psi$ are formulas such that $\varphi\vtrd_{\lab_v}\psi$, we may equivalently write $\varphi\vtrd_v\psi$ (recall Definition \ref{def: fully reduced}).
	To deal with the undefined attributes of the nodes of a tableau, we let the `undefined value' be denoted by $\bot$. 
	Additionally, for some value $x$ of a set $B$, we denote as $f^x:A\to B$ the function which assigns to all the values of its domain the value $x$. 
	Finally, for nodes $v_1$ and $v_2$  of a tableau $G$, we write $(v_1,v_2)^t$ for the edge $(v_1,v_2)$ labelled with $t$.
\end{rem}

\subsubsection{The Procedures}	\label{subsubsec: the procedures}

Here, we present how the algorithm constructs a tableau and we describe all the necessary procedures in detail. The algorithm expands the nodes in a depth-first, left-to-right fashion and defines their statuses in a postorder manner.

The procedure \tref{isSat} takes as input the formula whose satisfiability we want to examine and it merely initializes a tableau and defines the label of its root node.
Then, the procedure \tref{ct} is used which roughly speaking builds a tableau and defines the status of the root node which determines the satisfiability of the input formula.

\begin{procedure}[H]
\caption{isSat($\varphi_0$)}\label{isSat}
\KwIn{A formula $\varphi_0 \in \mc{L}_s$.}
\KwOut{Whether the formula $\varphi_0$ is satisfiable or not.}
\BlankLine

Define a global variable $G$ to hold a graph structure: 
$G = (V, E)$ with $V := \emptyset$ and $E := \emptyset$\;

Define a global variable $\rat$ to hold the fulfillment relation for $G$.\;

Define the node $v = \pos{\lab_v,D_v,\sts_v}$ as follows:\quad
	$\lab_v := \pos{\{\varphi_0\}, \rd^0}$,\qquad
	$D_v := \bot$,\qquad
	$\sts_v := \bot$\;																									\label{issat: root1}
\tref{ct}$(\bot, \bot, v)$ 																								\label{call}\;
\leIf 	{$\sts_v = \sat$} 	{\Return{\texttt{true}}} 	{\Return{\texttt{false}}}
\end{procedure}

\begin{procedure}[H]
\caption{constructTableau($v_0, t, v_1$)}	\label{ct} 
\KwIn{A node $v_0\!\in\! V\cup\!\{\bot\}$, an edge tag $t$ (a formula), both of them potentially undefined, and a node $v_1\!\notin\!V$.}
\KwOut{-}
\BlankLine

\eIf(\tcc*[f]{see Definition \ref{def: similar labels} for similar labels})
	{$\exists v_1'\in V$ such that $\lab_{v_1'} \approx \lab_{v_1}$}{														\label{ct: caching}
	
	\leIf{$v_1'$ is a forward ancestor of $v_0$ (i.e. $(v_1',v_0)\!\in\! E_f^+$)}{
		$E_c \!:= E_c	 \!\cup\! \{(v_0,v_1')^t\!\}$
	}
	{ 
		$E_b \!:= E_b \!\cup\! \{(v_0,v_1')^t\}$																			\label{ct: edge}
	}
	
	\lIf{$\rdf_{v_1} \!\!\neq \rdf_{v'_1}$}{
		extend $\rdf_{v_1'}$ with $\rdf_{v_1}$:\;		
			\hspace{3.1 cm} $\Phi_{v_1'} \!:= \Phi_{v_1'} \!\cup \rdf_{v_1}$ and for each $\alpha$/$\beta$ formula $\varphi\in \Phi_{v_1'} \!\cap \ev$, $\rd_{v_1'}(\varphi) := 1$																											\label{ct: extrdf}
	}
}		
{
	Extend $G$ with $v_1$: 
		$V := V \cup \{v_1\}$ and if $v_0\neq\bot$, then add a labelled forward edge: $E_f := E_f \cup \{(v_0,v_1)^t\}$			\label{ct: add node}
	 
	\uIf{$\exists \varphi \in \Phi_{v_1}$ such that $\neg\varphi\in\Phi_{v_1}$ or $\varphi=\neg\cpi\psi$}{						\label{ct: contr}
		Assign to $D_{v_1}$ the empty set and to $\sts_{v_1}$ the value \unsat.													\label{ct: unsat}
	}
	
	\uElseIf{$v_1$ is a partial node}{																							\label{ct: partial} 
		\tref{asr}$(v_1)$\;																										\label{ct: expand partial} 
		$\tref{csp}(v_1)$\;																										\label{ct: sts partial}
	}
	
	\Else{ 																														\label{ct: state}
		\tref{ansr}$(v_1)$\;
		$\tref{css}(v_1)$\;																										\label{ct: sts state}
	}
	
	\lIf{$v_1$ has cyclic parents}{																								\label{ct: update}
		$\tref{upd}(v_1)$
	}

}

\end{procedure}

The procedure \tref{ct} is the backbone of our algorithm, as it handles all the cases that occur. It accepts as input two nodes, $v_0$ and $v_1$, and an edge tag $t$. Notice that $v_0$ and $t$ might be undefined. In the case that $v_0\neq\bot$, the node $v_0$ is already part of the tableau, whereas $v_1$ is treated as a candidate node of the graph under construction and, in particular, as a candidate child of $v_0$. 
In lines \ref{ct: caching}-\ref{ct: extrdf}, the algorithm has discovered a pre-existing node $v_1'$ with a label similar with that of $v_1$. The procedure ignores $v_1$ and it adds a backward or a cyclic edge labelled with the edge tag $t$.
Now, in the opposite case, 
lines \ref{ct: add node}-\ref{ct: update} are used. First, the procedure adds $v_1$ to the tableau as a forward child of $v_0$. 
Next, if there are contradictory formulas in its label, its status is assigned the value \unsat. Otherwise, the procedure proceeds and examines what type of node $v_1$ is. One of the procedures \tref{asr} and \tref{ansr} expands $v_1$ and one of the procedures \tref{csp} and \tref{css} defines its status and its dependency set.
Finally, the nodes which are dependent on $v_1$ are updated by the procedure \tref{upd}.

The procedure \tref{asr} expands a partial node by applying the static rule (see Section \ref{subsec: tableau calculus}).
An active $\alpha/\beta$ formula $\varphi$ is chosen as the principal formula and the number of the defined children is at most the reduction degree of $\varphi$ (recall Definition \ref{def: reduction sets}). 
The set of formulas of a child is defined by adding to the set of formulas of the parent one of the reduction sets of $\varphi$ and the reduction function is defined so that $\varphi$ is considered as reduced and all the new $\alpha$/$\beta$ eventualities as active. 
We point out that in the case that the status of one of the defined children has been assigned the value \sat, \tref{asr} does not proceed to the definition of the remaining children.

\begin{procedure}[H]
\caption{applyStaticRule($v$)}\label{asr} 
\KwIn{A partial node $v$.}
\KwOut{-}
\vspace{0.1 cm}

Let $\varphi$ be an $\alpha$/$\beta$ formula in $\actf_v$, of minimal reduction degree $k$, and $\mc{R}_1^\varphi, \dotsc,\mc{R}_k^\varphi$ its reduction sets.\;																											\label{asr: principal} 
\For{$i:=1$ \KwTo $k$}{
	
	\uIf{there is no child $v_c$ of $v$ such that $\sts_{v_c} = \sat$}{
		Create a node $v_i = \pos{\lab_i, D_i, \sts_i}$ with its attributes undefined.\;
		Define the label $\lab_i = \pos{\Phi_{\lab_i}, \rd_{\lab_i}}$ as follows:\; 
					\hspace{0.5 cm} 	$\Phi_{\lab_i} := \Phi_v \cup \mc{R}_i^\varphi$ and
					\hspace{0 cm} 	for each $\alpha$/$\beta$ $\psi\in \ev\cap\Phi_{\lab_i}$,
							$\rd_{\lab_i}(\psi) :=
							\begin{cases}
							\rd_v(\psi)		&		\text{if $\psi\in \Phi_v\setminus\{\varphi\}$}\\
							1					&		\text{if $\psi = \varphi$}\\
							0					&		\text{if $\psi\in \mc{R}_i^\varphi\setminus\Phi_v$}
							\end{cases}$\;																					\label{asr: deflabel}
		\tref{ct}$(v, \bot, v_i)$
	}
	\lElse{
		break the for loop.
	}
}

\end{procedure}

\begin{procedure}[H]
\caption{calcStsPartial($v$)}\label{csp} 
\KwIn{A partial node $v$.}
\KwOut{-}
\BlankLine 

\uIf{there is a child of $v$ such that the value \sat is assigned to its status}{											\label{csp: children1}
	let $v_1,\dotsc,v_k$ be the children of $v$ such that for $i=1,\dotsc,k$, $\sts_{v_i}=\sat$.
}
\lElse{
	let $v_1,\dotsc,v_k$ be the children of $v$ such that for $i=1,\dotsc,k$, $\sts_{v_i}\in \{\bot,\tempsat\}$.			\label{csp: children2}
}

\uIf{$k>0$}{
	$D_v := \Big( \bigcup\limits_{i=1,\dotsc,k}$  
							\textbf{if} ($D_{v_i}\neq\bot$) \textbf{then} $D_{v_i}$ \textbf{else} $\{ v_i \}$ 
					$\Big) \setminus \{ v \}$ \;																			\label{csp: D}
	
	\ForEach{active eventuality $\varphi$ of $v$}{																			\label{csp: for} 
		\For{$i:=1$ \KwTo $k$}{
			\lForEach{$\psi\in \reach(\varphi,\lab_{v_i})$}{
				define $(v,\varphi) \rat (v_i,\psi)$:\quad 
				$\rat := \rat \cup \{ ((v,\varphi), (v_i,\psi)) \}$															\label{csp: rat}
			}
		}
	}

	\eIf{$\exists\varphi\in \actf_v\cap\ev$ which is $\rat$-unfulfilled}{													\label{csp: if_unfulfilled}
		Assign to $D_v$ the empty set and to $\sts_v$ the value \unsat.														\label{csp: unsat2}
	}
	{
		\leIf{$D_v = \emptyset$}{
			$\sts_v := \sat$
		}
		{
			$\sts_v := \tempsat$
		}																													\label{csp: sat_tempsat}
	}

}
\lElse{
	Assign to $D_v$ the empty set and to $\sts_v$ the value \unsat.															\label{csp: unsat1}
}
\end{procedure}

After the expansion of a partial node, the algorithm calculates its dependency set, the fulfillment relation for its active eventualities and its status by using the procedure \tref{csp}. 
The procedure uses only the children whose statuses have not been assigned the value \unsat (see lines \ref{csp: children1}-\ref{csp: children2}) and if there is no such child, then it assigns the value \unsat.
First, it defines the dependency set by using the dependency sets of the appropriate children and by distinguishing cases on whether they have been defined. 
Then, according to lines \ref{csp: for}-\ref{csp: rat}, it defines the fulfillment relation for the active eventualities of the node.
Finally, according to lines \ref{csp: if_unfulfilled}-\ref{csp: sat_tempsat}, by using the defined dependency set and the defined fulfillment relation, it determines the status of the node.

The procedure \tref{ansr} expands a state by applying the transitional rule in lines \ref{ansr: trans1B}-\ref{ansr: trans2E} and the capability rule in lines \ref{ansr: capB}-\ref{ansr: capE} (see also Section \ref{subsec: tableau calculus}).
In the case that one of the already defined children has been assigned the value \unsat, the procedure stops the expansion of its input node.

\begin{procedure}[H]
\caption{applyNonStaticRules($v$)}\label{ansr} 
\KwIn{A state node $v$.}
\KwOut{-}
\BlankLine

Let $\Delta_1$ be the set $\big\{\neg\forall A.\chi\in\Phi_v \mid A\in\atp_\Omega \big\}$ and 
$\Delta_2$ the set $\big\{ \neg\cpi A \in\Phi_v \mid A\in\wt{\Sigma} \big\}$.

\ForEach{formula $\varphi$ in $\Delta_1\cup\Delta_2$}{
	\uIf{there is no child $v_c$ of $v$ such that $\sts_{v_c} = \unsat$}{
		
		Create a node $v' = \pos{\lab,D,\sts}$ with its attributes undefined.
		
		\uIf{$\varphi$ is some formula $\neg\forall a.\chi$}{																	\label{ansr: trans1B}
			Define $\lab$ as follows:
			\hspace{0.3 cm} $\Phi_\lab:= \{\neg\chi\}  \cup  \{\vartheta \mid \forall a.\vartheta\in\actf_v \}  \cup  \{\vartheta  \mid \forall \Omega.\vartheta \in \actf_v \}$
			\hspace{0.3 cm} $\rd_\lab := \rd^0$\;
			\tref{ct}$(v, \varphi, v')$
		}																														\label{ansr: trans1E}

		\uElseIf{$\varphi$ is some formula $\neg\forall\Omega.\chi$}{															\label{ansr: trans2B}
			Define $\lab$ as follows:
			\hspace{0.3 cm} $\Phi_\lab := \{\neg\chi\} \cup \{\vartheta  \mid \forall \Omega.\vartheta\in\actf_v \}$
			\hspace{0.3 cm} $\rd_\lab := \rd^0$\;
			\tref{ct}$(v, \varphi, v')$
		}																														\label{ansr: trans2E}

		\ElseIf{$\varphi$ is some formula $\neg\cpi A$}{																		\label{ansr: capB}
			Let $\{ \cpi(\varphi_1\Ra\psi_1),\dotsc,\cpi(\varphi_k\Ra\psi_k) \}$ be the set of all the capabilities statements for the agent $\imath \in I$ in $\Phi_v$ that concern precondition-effect processes.\;
			Define $\lab$ as follows:
			\hspace{0.3 cm} $\Phi_\lab := \{ \varphi_1,\dotsc,\varphi_k, \neg\forall A.\forall(\neg\psi_1).\dotsb\forall(\neg\psi_k).\false \}$
			\hspace{0.3 cm} $\rd_\lab := \rd^0$\;
			\tref{ct}$(v, \varphi, v')$
		}																														\label{ansr: capE}
		
	}
	\lElse{
		break the for loop.
	}
}

\end{procedure}

Similarly to \tref{csp}, the procedure \tref{css} defines the dependency set of a state (see line \ref{css: D}), the fulfillment relation for its active eventualities (see lines \ref{css: for}-\ref{css: rat}) and its status (see lines \ref{css: if_unfulfilled}-\ref{css: sat_tempsat}).
If there is a child whose status has been assigned the value \unsat, then it directly assigns the value \unsat to the status of the input state.

\begin{procedure}[H]
\caption{calcStsState($v$)}\label{css} 
\KwIn{A state node $v$.}
\KwOut{-}
\BlankLine

Let $v_1,\dotsc,v_k$ the children of $v$ with $k\geq0$.\;
\uIf{for $i=1,\dotsc,k$, $\sts_{v_i} \neq \unsat$}{																		\label{css: if}
	
	$D_v := \Big( \bigcup\limits_{i=1,\dotsc,k}$  
		\textbf{if} ($D_{v_i}\neq\bot$) \textbf{then} $D_{v_i}$ \textbf{else} $\{ v_i \}$ 
		$\Big) \setminus \{ v \}$ \;																					\label{css: D}
	
	\ForEach{active eventuality $\neg\forall A.\chi$ of $v$ ($A\in\atp_\Omega$)}{										\label{css: for}
		Let $v'$ be the child of $v$ such that $(v,v')$ is labelled with $\neg\forall A.\chi$.\;
		Define $(v,\neg\forall A.\chi) \rat (v',\neg\chi)$: \quad
		$\rat := \rat  \cup  \{ (v,\neg\forall A.\chi), (v',\neg\chi) \}$												\label{css: rat}
	}

	\eIf{$\exists\varphi\in \actf_v\cap\ev$ which is $\rat$-unfulfilled}{												\label{css: if_unfulfilled}
		Assign to $D_v$ the empty set and to $\sts_v$ the value \unsat.													\label{css: unsat2}
	}
	{
		\leIf{$D_v = \emptyset$}{
			$\sts_v := \sat$
		}
		{
			$\sts_v := \tempsat$
		}																												\label{css: sat_tempsat}
	}

}
\lElse{
	Assign to $D_v$ the empty set and to $\sts_v$ the value \unsat.														\label{css: unsat1}
}
	
\end{procedure}

The procedure \tref{upd} takes as input a node $v$ and it updates only the nodes whose statuses have been assigned the value \tempsat and at the same time, their dependency sets contain $v$.
If the status of the input node or one of the updated nodes has been assigned one of the values \sat and \unsat, then the procedure \tref{propag} is used to propagate the same status value to the appropriate ancestor nodes.
In lines \ref{upd: while_Unsat}-\ref{upd: propag_Unsat}, the procedure updates those nodes which have an eventuality which is $\rat$-unfulfilled. 
Then, in lines \ref{upd: while_Sat_TempSat}-\ref{upd: TempSat}, it updates the dependency sets of the appropriate remaining nodes by using the dependency set of the input node and it defines their statuses accordingly.

\begin{procedure}[H]
	\caption{updDepNodes($v$)}		\label{upd}
	\KwIn{A node $v$ which has at least one cyclic parent.}
	\KwOut{-}
	\BlankLine
	
	\lIf{$\sts_{v} \in \{\sat,\unsat\}$}{																	\label{upd: propag_Sat_Unsat}
		$\tref{propag}(v,v)$
	}
	
	\While{$\exists v'\in V$ such that $\sts_{v'} = \tempsat$ and $v\in D_{v'}$ and 
		$\exists\varphi\in\actf_{v'}\cap\ev$ which is $\rat$-unfulfilled}{									\label{upd: while_Unsat}
		Assign to $D_{v'}$ the empty set and to $\sts_{v'}$ the value \unsat.\;								\label{upd: unsat}
		$\tref{propag}(v,v')$																				\label{upd: propag_Unsat}
	}
	
	\While{there is a node $v'\in V$ such that $\sts_{v'} = \tempsat$ and $v\in D_{v'}$}{					\label{upd: while_Sat_TempSat}
		$D_{v'} := \left(D_{v'} \setminus \{v\}\right) \cup D_v$\;											\label{upd: updateDep}
		\uIf{$D_{v'} = \emptyset$}{																			\label{upd: Sat_TempSat}
			$\sts_{v'} := \sat$\;																			\label{upd: Sat}
			$\tref{propag}(v,v')$																			\label{upd: propag_Sat}
		}
		\lElse{
			$\sts_{v'} := \tempsat$																			\label{upd: TempSat}
		}
	}
	
\end{procedure}

The procedure call $\tref{propag}(v_0,v_1)$ updates only the nodes whose statuses have been assigned the value \tempsat and their dependency sets contain $v_0$. Additionally, the update begins from the parents of $v_1$ whose status has been assigned one of the values \sat and \unsat (see \tref{upd}).
Each time that the procedure \tref{propag} is used, it exclusively propagates the status value of its second (input) argument. It follows two simple rules. 
If the status of a child of a partial node or the statuses of all the children of a state have been assigned the value \sat, then the value \sat is also assigned to the status of the parent. 
On the other hand, if the status of a child of a state or the statuses of all the children of a partial node have been assigned the value \unsat, then the value \unsat is also assigned to the status of the parent.

\begin{procedure}[H]
	\caption{propagateSts($v_0,v_1$)}		\label{propag}
	\KwIn{A node $v_0$ whose status has been defined and a node $v_1$ such that $\sts_{v_1}$ is in $\{\sat, \unsat\}$.}
	\KwOut{-}
	\BlankLine
	
	Let $X$ be a set of nodes which is initialized as the singleton $\{ v_1 \}$.\;
	\While{$X\neq\emptyset$}{
		Choose a node $v$ from $X$ and remove it from $X$.\;
		\ForEach{parent $v'$ of $v$ such that $\sts_{v'}  = \tempsat$ and $v_0\in D_{v'}$}{									\label{propag: for}
			\uIf{$\sts_v \!\!=\! \sat$  $\land$ 
				\big($v'$ is a partial node  or  ($v'$ is a state and $\forall v''( (v',v'')\!\in\! E \Rightarrow \sts_{v''}\!=\!\sat)$)\big)}{
				Assign to $D_{v'}$ the empty set and to $\sts_{v'}$ the value \sat.\;										\label{propag: sat}
				$X : = X \cup \{v'\}$
			}		
			\ElseIf{$\sts_v \!\!=\! \unsat$  $\!\land\!$ 
				\big($v'$ is a state  or  ($v'$ is a partial node and $\forall v''( (v'\!,v'')\!\in\! E \Rightarrow \sts_{v''}\!=\!\unsat)$)\big) }{
				Assign to $D_{v'}$ the empty set and to $\sts_{v'}$ the value \unsat.\;										\label{propag: unsat}
				$X : = X \cup \{v'\}$
			}
		}	
	}
	
\end{procedure}

\vspace{0.0 cm}
\begin{rem}
	In the following of this paper, whenever we refer to a tableau for some formula $\varphi$, we consider the tableau that is built by the procedure \tref{isSat} with input the formula $\varphi$.
\end{rem}

\subsection{Examples} 		\label{subsec: examples}

In this subsection, we give three examples of tableaux with which we present different aspects of our algorithm. We adopt some notational conventions in order to simplify their presentation.
The tableau nodes are presented within a rectangle. To easily distinguish between states and partial nodes, the latter ones are presented within a dashed rectangle.
We present only the set of formulas of a label of a node without its reduction function and we do not distinguish the reduced formulas from the active ones. We trust the reader to fill in the details.
We may present consecutive applications of the static rule in a single step. Additionally, the principal $\alpha$/$\beta$ formula of the static rule (which is applied only on partial nodes) is underlined.
For sake of simplicity, we omit the dependency sets of the nodes, under the condition that they are the empty set.
As the dependency set and the status of a node may be calculated several times during the construction of a tableau, the different values are listed from the oldest to the most recent one and they are separated by a vertical bar `$\mid$'. 
Assuming the steps of the algorithm to be numbered and each step being either about creating a node, or (re)defining the status and the dependency set of a node, we write $v_{i_1.\dotsb.i_k}$ with $k\geq2$ for the node which becomes created at step $i_1$ and gets its dependency set and its status defined for the first time at step $i_2$ and redefined at steps $i_3,\dotsc,i_k$. We simply refer to any node $v_{i_1.\dotsb.i_k}$ by writing $v_{i_1}$ and we refer to its label as $\lab_{i_1}$.
Finally, backward edges are depicted as dashed arrows, whereas cyclic edges as dotted arrows.

The first example, see Figure \ref{fig: ex1}, has to do with the conjunction of the capabilities statements $\cpi(p\land q \Ra r)$ and $\neg\cpi(p\Ra r)$.
The specific formula cannot be satisfiable. Notice that, by Corollary \ref{cor: =>2} and by semantics, the process $p\Ra r$ is of type $p\land q \Ra r$, in symbols $\lonto{p\Ra r} \subseteq \lonto{p\land q \Ra r}$. Thus, there is no $\mc{L}$-model $\mc{M}$, there is no state $w$ of $\mc{M}$ such that $\lonto{p\land q \Ra r}$ is a subset of $\imath^\mc{M}(w)$ and at the same time, $\lonto{p\Ra r}$ is not a subset of $\imath^\mc{M}(w)$, so that $w$ satisfies $\cpi(p\land q \Ra r)$ and $w$ does not satisfy $\cpi(p\Ra r)$. In other words, the formula $\cpi(p\land q \Ra r) \ra \cpi(p\Ra r)$ is valid.
The algorithm applies the capability rule through the procedure \tref{ansr}. Intuitively speaking, the algorithm tries, through the formulas of $v_2$, to define a transition of type $p\Ra r$, but not of type $p\land q \Ra r$. As argued earlier, this is impossible and as a result, it fails and it concludes that the input set of formulas is unsatisfiable.

\begin{figure}[ht] \label{fig: ex1}
	\centering
	\includegraphics[width=\columnwidth]{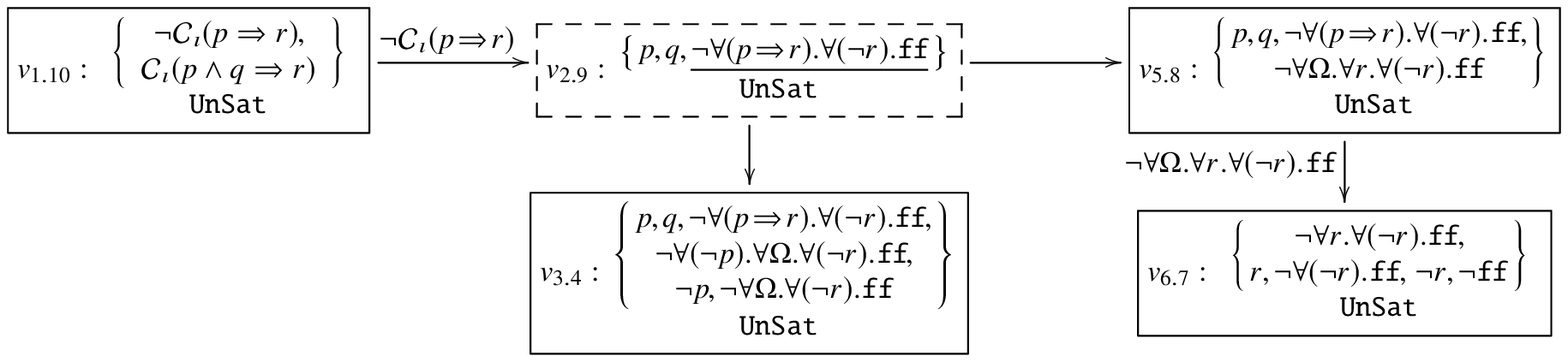}
	\caption{The tableau for the unsatisfiable set of formulas $\{ \cpi(p\land q \Ra r), \neg\cpi(p\Ra r) \}$}
\end{figure}

The second example, see Figure \ref{fig: ex2}, demonstrates how the algorithm handles the preconditions-effects processes.
The formula $\neg\forall((p \Ra r) +a).p$ requires the existence of a transition either of a process of type $p \Ra r$ or of $a$. In both cases, the reachable state should not satisfy the atomic formula $p$. On the other hand, the formula $\forall(p\land q \Ra r).p$ requires that after the execution of a process of type $p\land q \Ra r$, the atomic formula $p$ should be satisfied. So, in the case of a transition of type $p \Ra r$, as $\lonto{p\Ra r} \subseteq \lonto{p\land q \Ra r}$ (as argued in the previous example), the reachable state should satisfy $p$ and $\neg p$. This is not possible and the algorithm reaches the same conclusion as the status of $v_2$ indicates. 
Despite the previous result, the algorithm also examines the formulas $\neg\forall a.p$ and $\forall(p\land q \Ra r).p$ through the node $v_{20}$. Notice that the transition that the formula $\neg\forall a.p$ requires might not be of type $p\land q \Ra r$ and as a result, the problematic situation of the formulas $\neg\forall(p \Ra r).p$ and $\forall(p\land q \Ra r).p$ is avoided. Therefore, there is a transition of $a$ which leaves a state which satisfies the preconditions $p$ and $q$ and reaches a state which does not satisfy the effects $r$. Thus, the specific transition is not of type $p\land q \Ra r$ and the reachable state is not necessary to satisfy $p$ due to the formula $\forall(p\land q \Ra r).p$. The algorithm reaches the same conclusion through the nodes $v_{20}$, $v_{21}$, $v_{22}$ and $v_{23}$. Observe that in the case of the principal formula $\forall r.p$ of $v_{22}$, the procedure \tref{asr} examines only the reduction set $\{\neg r\}$ and not the reduction set $\{p\}$, as the status of $v_{23}$ was assigned the value \sat.
Similarly, the algorithm examines only one of the two reduction sets of the $\beta$ formula $\forall(p\land q \Ra r).p$ of the partial node $v_{20}$.

\begin{figure}[ht] \label{fig: ex2}
	\centering
	\includegraphics[width=\columnwidth]{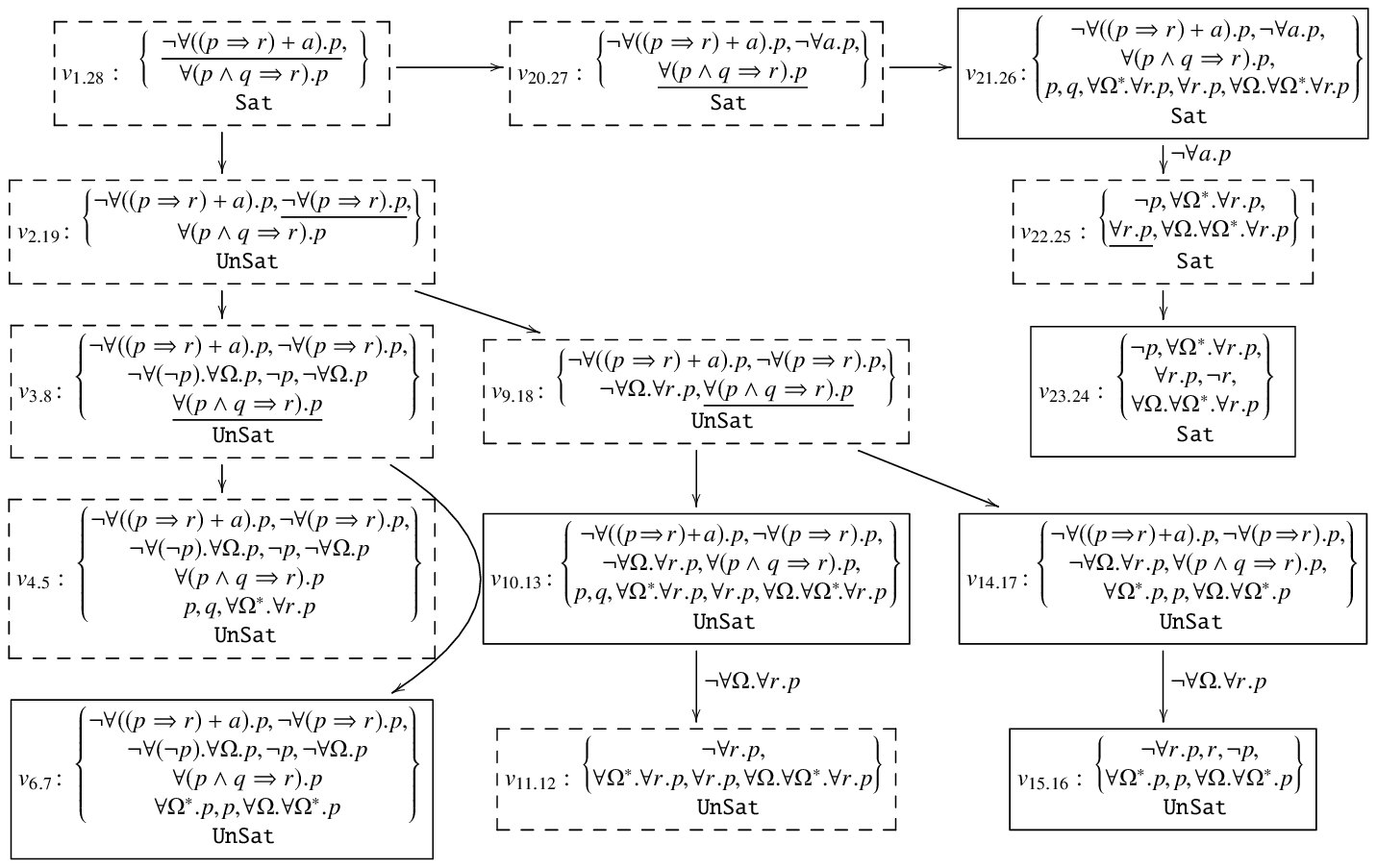}
	\caption{The tableau for the satisfiable set of formulas $\{ \neg\forall((p \Ra r) +a).p, \forall(p\land q \Ra r).p \}$}
\end{figure}

The last example, see Figure \ref{fig: ex3}, deals with the eventuality $\neg\forall(a+b)^*.p$ and with the formula $\forall a^*.p$. The latter formula requires that after any number of executions of $a$, including zero, the atomic formula $p$ must hold. The eventuality $\neg\forall(a+b)^*.p$ requires that after a finite number of executions of $a+b$, including zero, the formula $p$ must not be satisfied.
According to Definitions \ref{def: reduction sets} and \ref{def: decomposition}, the reduction sets of $\neg\forall(a+b)^*.p$ are three singletons: $\{\neg p\}$, $\{\neg\forall a.\forall(a+b)^*.p\}$ and $\{\neg\forall b.\forall(a+b)^*.p\}$.
In the case of the first set, the algorithm fails to fulfill the eventuality due to $\forall a^*.p$ (see the node $v_3$). In the case of the second reduction set, a loop is formed, i.e. a graph cycle, and as the status of $v_1$ has not been defined yet, the node $v_5$ is dependent on $v_1$ and thus, $\sts_{v_5}$ is initially assigned the value \tempsat. In the case of the last reduction set, the eventuality is fulfilled through the nodes $v_7$, $v_8$ and $v_9$. Observe that the procedure \tref{asr} examines only the first reduction set of the eventuality of $v_8$ as the status of $v_9$ was assigned the value \sat. 
Eventually, after the definition of $\sts_{v_1}$, the status of $v_5$ is recalculated and it is assigned the value \sat.
The fulfillment relation of the tableau is also illustrated in Figure \ref{fig: ex3} as a graph in which the nodes are the related pairs.

\begin{figure}[ht] \label{fig: ex3}
	\centering
	\includegraphics[width=\columnwidth]{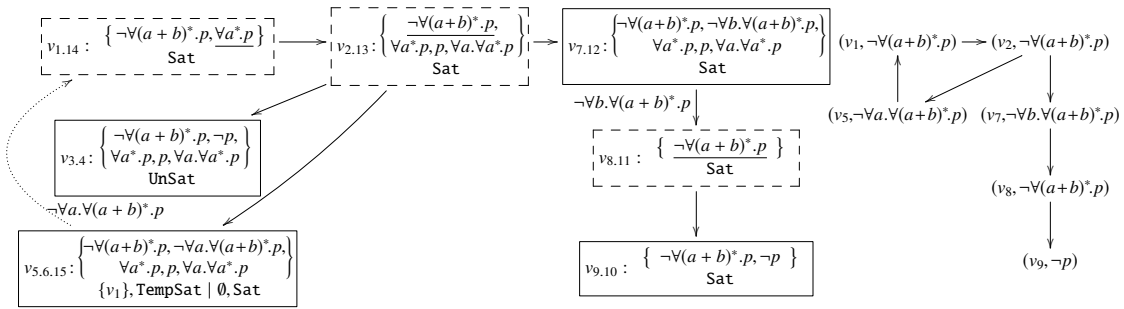}
	\caption{The tableau for the set of formulas $\{ \neg\forall(a+b)^*.p, \forall a^*.p \}$ and its fulfillment relation as a graph.}
\end{figure}

\section{Correctness and Complexity} 		\label{sec: correctness and complexity}

In this section, we show that the algorithm works properly.
More specifically, in Subsection \ref{subsec: tableau properties}, we give some tableaux properties which will be useful in the subsequent subsections. In Subsection \ref{subsec: soundness}, we show that the algorithm is sound and in Subsection \ref{subsec: completeness}, that it is complete. Finally, in Subsection \ref{subsec: complexity}, we show that it runs in exponential time.

\subsection{Tableaux Properties} 		\label{subsec: tableau properties}

\begin{lem}	\label{lem: saturation properties}
	In a tableau $G=(V,E)$, the following properties hold:
	\begin{enumerate}[i)]
		\item The set of formulas of a partial node is a subset of the set of formulas of any of its children.
		\item There is no sequence $v_1,\dotsc,v_k$ of partial nodes of $G$ such that $v_1=v_k$ and for $i=1,\dotsc, k-1$, $(v_i,v_{i+1})\in E$.
		\item For each partial node of $G$, there is a path of nodes which reaches a state. 
	\end{enumerate}
\end{lem}
\begin{proof}
	The properties of the lemma are an immediate consequence of the way that the algorithm applies the static rule through the procedure \tref{asr}.
	Moreover, by the way that the active and the reduced sets of a node have been defined, the static rule cannot be applied more than once for a specific $\alpha$/$\beta$ formula. Hence, a state is eventually defined.
\end{proof}

\begin{lem}	\label{lem: sts dep} \label{lem: unchanged sat unsat}
	Let $G=(V,E)$ be a tableau and $v$ one of its nodes. 
	\begin{enumerate}[i)]
		\item If the status of $v$ has been assigned the value \sat or \unsat, then the dependency set $D_v$ is empty. 
		\item If the status of $v$ has been assigned the value $\tempsat$, then $D_v \neq \emptyset$.
		\item From the moment that the status of a node is assigned the value \sat or \unsat, it never changes.
	\end{enumerate}
\end{lem}
\begin{proof}
	In all the cases that the procedures \tref{ct}, \tref{csp}, \tref{css}, \tref{upd} and \tref{propag} define the status and the dependency set of a node, they are always compatible with the first two properties.
	As far as the third property is concerned, we point out that the procedures \tref{upd} (see lines \ref{upd: while_Unsat} and \ref{upd: while_Sat_TempSat}) and \tref{propag} (see line \ref{propag: for}) update the nodes whose statuses have been assigned the value \tempsat.
\end{proof}

\begin{lem}	\label{lem: rat}
	Let $G=(V,E)$ be a tableau, $\rat$ the fulfillment relation of $G$, $v$ a node of $G$ and $\varphi = \neg\forall A_1.\dotsb\forall A_n.\forall A^*.\chi$ an active eventuality of $v$ with $n\geq0$ and $\neg\chi\notin\ev$. 
	If $(v,\varphi) \rat (v',\psi)$, then the following properties hold: 
	\begin{enumerate}[i)]
		\item $v'$ is a child of $v$ and either $\psi$ is an active eventuality of $v'$ of the form $\neg\forall B_1.\dotsb\forall B_r.\forall A^*\!.\chi$ with $r\geq0$, or $\psi$ is the non-eventuality $\neg\chi$ which is in the set of formulas of the label of $v'$,
		\item if $v$ is a state, then $\varphi$ is some eventuality $\neg\forall B.\vartheta$ with $B\in\atp_\Omega$ and $\psi$ is the eventuality $\neg \vartheta$.
	\end{enumerate}
\end{lem}
\begin{proof}
	The relation $\rat$ is defined exclusively by the procedures \tref{csp} and \tref{css}. Thus, the lemma is an immediate consequence of how these procedures work.
	First, we assume that $v$ is a state. Since $\varphi$ is an active eventuality of $v$, it is some formula $\neg\forall B.\vartheta$ with $B \in\atp_\Omega$. By the definition of \ev, we have that $\neg\vartheta\in \ev$. We assume that $v'$ is the child that \tref{ansr} defines such that the edge $(v,v')$ is labelled with $\varphi$. Notice that according to \tref{ansr}, the formula $\neg\vartheta$ is an active eventuality of $v'$. According to line \ref{css: rat} of \tref{css}, we have that $(v,\varphi) \rat (v',\neg\vartheta)$. Therefore, the related pairs meet the requirements of the lemma.
	
	Now, we presume that $v$ is a partial node and that $v'$ is one of the children of $v$ which \tref{csp} considers (see lines \ref{csp: children1}-\ref{csp: children2}). According to line \ref{csp: rat} of \tref{csp}, for a formula $\psi$ in $\reach(\varphi,\lab_{v'}) \subseteq \Phi_{v'}$, we have that $(v,\varphi) \rat (v',\psi)$. By Definition \ref{def: reach}, we know that there is a sequence $\varphi_1,\dotsc,\varphi_k$ of formulas with $k\geq1$ such that $\varphi_1=\varphi$ and $\varphi_k=\psi$ and for $i=1,\dotsc,k-1$, $\varphi_i \in \rdf_{v'}$ is a reduced $\alpha$/$\beta$ formula and $\varphi_i\vtrd\varphi_{i+1}$.
	In the case that $k=1$, we can conclude that the formula $\varphi=\psi$	is also an active eventuality of $v'$.
	In the case that $k>1$, since $\varphi_1\vtrd\varphi_2$, by Lemma \ref{lem: properties of vtrd}, we know that there are two cases for $\varphi_2$. In the first case, it is the non-eventuality $\neg\chi$ (note that $k=2$) and thus, the lemma is satisfied. In the other case, $\varphi_2$ is an active eventuality of $v'$ of the form $\neg\forall B_1.\dotsb\forall B_r.\forall A^*.\chi$ with $r\geq1$ and $B_1\in\atp_\Omega$ (again, $k=2$).
\end{proof}

\begin{lem}	\label{lem: dependency set}
	Let $G=(V,E)$ be a tableau, $v$ a node of $G$ such that $\sts_v\neq\unsat$ and $\varphi$ an active eventuality of $v$.
	The following properties hold:
	\begin{enumerate}[i)]
		\item Each node $v'$ of the dependency set $D_v$ is a forward ancestor node of $v$
		and at the same time, it is a descendant of $v$ through a path of nodes in which the last edge is cyclic. 
		\item If $\varphi$ is not $\rat$-fulfilled and it is $\rat$-dependent on a node $v'$, then $v'$ is in $D_v$.
		\item If $\varphi$ is $\rat$-fulfilled only through sequences of $\rat$-related pairs in which there is at least one node whose status has been assigned the value \tempsat, 
		then for each sequence $(v_1,\vartheta_1) \rat \dotsb \rat (v_k,\vartheta_k)$ through which $\varphi$ is $\rat$-fulfilled, for $i=2,\dotsc,k$, if $\sts_{v_i} = \tempsat$, then $D_{v_i} \setminus \{v\}$ is a subset of $D_v$.
	\end{enumerate}
\end{lem}
\begin{proof}
	We proceed by induction on the inductive definition of $G$. A tableau for a formula $\varphi$ can be viewed as the graph which has as root a node whose label has as only active formula the input formula $\varphi_0$ (see line \ref{issat: root1} of \tref{isSat}) and then, according to the procedure call $\tref{isSat}(\varphi_0)$, it is inductively defined either by expanding some node, or by (re)defining the status and the dependency set of some node. 
	In the following, we examine all the cases that concern the definition of the dependency set of a node $v$ and we examine all the properties of the lemma simultaneously. Additionally, we assume that the lemma holds for all the nodes whose dependency sets have been defined before the one that we examine.
	
	The dependency set of a node is calculated by one of the procedures \tref{csp} and \tref{css} and it might be recalculated due to line \ref{ct: update} of \tref{ct} with the use of \tref{upd}. In the following, we examine the previous cases.
	
	There is no difference on the definition of the dependency set of a partial node with that of a state (see line \ref{csp: D} of \tref{csp} and line \ref{css: D} of \tref{css}). 
	Thus, we consider a node $v$ without distinguishing cases on its type. The procedures \tref{csp} and \tref{css} use the appropriate children of $v$. Let $v_c$ be one of them.	
	If the dependency set of $v_c$ is undefined, then its status is also undefined, as both attributes are calculated in a postorder manner. 
	Thus, $v_c$ must be a forward ancestor of $v$ and a cyclic child of $v$ (see also line \ref{ct: edge} of \tref{ct}). The algorithm adds $v_c$ to $D_v$. 
	Additionally, if $(v,\varphi)\rat(v_c,\psi)$, where $\varphi$ is an active eventuality of $v$ and $\psi$ an active eventuality of $v_c$ (recall Lemma \ref{lem: rat}), then $\varphi$ is $\rat$-dependent on $v_c$ which is in $D_v$.
	
	Now, if $D_{v_c}$ has been calculated, then its status should also have been calculated. 
	In this case, the nodes of $D_{v_c}$ are added to $D_v$. By the induction hypothesis, we know that the nodes of $D_{v_c}$ are forward ancestors of $v_c$ in which the node $v$ might be included. Thus, $v$ is removed from the set $D_v$ and as a result, all the nodes of $D_v$ are forward ancestors of $v$.
	Moreover, by the induction hypothesis, we also know that the nodes of $D_{v_c}$ are descendants of $v_c$ as the lemma requires. Since $v_c$ is a descendant of $v$, all the nodes which are added to $D_{v}$ also satisfy the lemma.
	Next, we examine the second property of the lemma. If $(v,\varphi)\rat(v_c,\psi)$ and $\psi$ is $\rat$-dependent on a node $v''$, then $\varphi$ is also $\rat$-dependent on the same node. Again, by the induction hypothesis, we know that $v''$ is in $D_{v_c}$ and as argued earlier, it is also in $D_v$. Notice that if $v'' = v$, then, by Definition \ref{def: fulfilled dependent}, $\varphi$ cannot be dependent on $v$ as $(v,\varphi)$ is related to the pair $(v_c,\psi)$.
	
	Finally, if $(v,\varphi)\rat(v_c,\psi)$ and the active eventuality $\psi$ of $v_c$ is $\rat$-fulfilled, then $\varphi$ is also $\rat$-fulfilled. By assumption, we know that if $\varphi$ is $\rat$-fulfilled, then it is through a sequence $(v,\varphi)\rat(v_c,\psi) \rat (v_1,\vartheta_1) \rat\dotsb \rat (v_k,\vartheta_k)$ of pairs, where $k\geq1$, in which there is at least one node whose status has been assigned the value \tempsat. Since $\varphi$ is $\rat$-fulfilled through $\psi$, it is clear that the same also holds for $\psi$. By the induction hypothesis, we know that for $i=1,\dotsc,k$, if $\sts_{v_i} = \tempsat$, then $D_{v_i} \setminus \{v_c\}$ is a subset of $D_{v_c}$. 
	After the definition of the status of $v_c$, there is no node whose dependency set contains $v_c$. If there was such a node, then by the first property of this lemma, $v_c$ would have at least one cyclic parent and the procedure \tref{upd} would be triggered due to line \ref{ct: update} of \tref{ct}. According to \tref{upd} and \tref{propag}, the node $v_c$ is removed from all the dependency sets. So, the desirable subset property remains true and more specifically, we have that for $i=1,\dotsc,k$, if $\sts_{v_i} = \tempsat$, then $D_{v_i} \subseteq D_{v_c}$.
	Since the nodes of $D_{v_c} \setminus \{v\}$ are added to $D_v$ and the active eventuality $\varphi$ of $v$ is $\rat$-fulfilled through the active eventuality $\psi$ of $v_c$, it is clear that the third property of the lemma is satisfied.

	Now, we examine the case of a node $v$ whose dependency set is recalculated by the procedure \tref{upd} in line \ref{upd: updateDep}. 
	Its dependency set is recalculated after the definition of the status of a node $v_1$ and thus, we need to admit that $\sts_v = \tempsat$ and $v_1\in D_v$ (see line \ref{upd: while_Sat_TempSat}).
	According to line \ref{upd: updateDep}, $D_v$ is updated as follows: $D_v := (D_v\setminus\{v_1\}) \cup D_{v_1}$.
	
	By the induction hypothesis, we know that the lemma holds for the nodes of $D_v$ before its redefinition and it also holds for the nodes of $D_{v_1}$. Therefore, since $v_1\in D_v$, $v_1$ is a forward ancestor of $v$ and the nodes of $D_{v_1}$ are forward ancestors of $v_1$. It is clear that the nodes of $D_{v_1}$ are also forward ancestors of $v$. Therefore, the nodes of the updated $D_v$ are forward ancestors of $v$. 
	Furthermore, by the induction hypothesis, we also know that $v_1$ is a descendant of $v$ and the nodes of $D_{v_1}$ are descendants of $v_1$. Thus, the nodes of $D_{v_1}$ which are added to the updated $D_v$ are descendants of $v$.
	
	Finally, after the definition of the status of $v_1$ and the definition of the relation $\rat$ for the active eventualities of $v_1$, no eventuality of $v$ can be $\rat$-dependent on $v_1$. If some active eventuality of $v$ was dependent on $v_1$ before the definition of $\sts_{v_1}$, then, after the definition of $\sts_{v_1}$, it is dependent on the same nodes with some active eventuality of $v_1$. By the induction hypothesis, we know that these dependency nodes are in $D_{v_1}$ which are added to $D_v$. Therefore, the second property of the lemma is satisfied.
	
	In a similar way, we can argue for the third property of the lemma. By the induction hypothesis, we know that it is satisfied for the active eventualities of $v$ and for those of $v_1$. After the redefinition of $D_v$, if an active eventuality of $v$ is $\rat$-fulfilled through an active eventuality of $v_1$, then we know that the necessary nodes are in the updated $D_v$.
\end{proof}

\begin{lem}		\label{lem: sat unsat} 
	After the completion of the construction of a tableau $G$, the status of each node $v$ is assigned either the value \unsat or the value \sat (i.e. there is no node $v$ such that $\sts_v$ is undefined or $\sts_v = \tempsat$).
\end{lem}
\begin{proof}
	By the way that the procedure \tref{ct} works, it is clear that after the end of the construction of a tableau, the status of each node has been calculated and defined. 
	In the following, we examine a node $v$ whose status has been assigned the value \tempsat.	By Lemma \ref{lem: unchanged sat unsat}, we know that the status of a node which has been assigned one of the values \sat and \unsat never changes. Thus, the status of $v$ was assigned the value \tempsat the first time it was calculated by one of the procedures \tref{csp} and \tref{css}.
	By Lemma \ref{lem: sts dep}, we have that $D_v\neq\emptyset$ and by Lemma \ref{lem: dependency set}, we know that the nodes of $D_v$ are forward ancestors of $v$. Moreover, at the moment of the definition of $\sts_v$ and of $D_v$, since the statuses are defined in a postorder manner, the statuses of the nodes of $D_v$ have not been defined yet. 
	We consider a node $v_1$ which belongs to $D_v$. By Lemma \ref{lem: dependency set}, we know that $v_1$ has at least one cyclic parent. After the definition of its status, according to line \ref{upd: updateDep} of \tref{upd} which is triggered due to line \ref{ct: update} of \tref{ct}, $v_1$ is removed from $D_v$ and the nodes of $D_{v_1}$ are added to $D_v$. Again, since $D_{v_1}$ is a set of forward ancestors whose statuses have not been defined yet, we conclude that the same also holds for $D_v$.
	Eventually, after the completion of the construction of a tableau, the statuses of all the nodes cannot remain undefined. Therefore, $D_v$ is eventually defined as the empty set. By Lemma \ref{lem: sts dep}, we know that the dependency set of a node whose status is assigned the value \tempsat cannot be the empty set. Thus, $\sts_v$ is eventually assigned one of the values \sat and \unsat.
\end{proof}

\begin{lem} \label{lem: tempsat rat}
	Let $G=(V,E)$ be a tableau and $v$ a node of $G$. 
	If $\sts_v = \tempsat$, then each active eventuality $\varphi$ of $v$ either is $\rat$-fulfilled or is $\rat$-dependent on some node.
\end{lem}
\begin{proof}
	The procedures \tref{csp} and \tref{css} assign the value \tempsat to $\sts_v$ iff the following two conditions hold: 
	\begin{enumerate*} [i)]
		\item each one of the active eventualities of $v$ is $\rat$-fulfilled or it is $\rat$-dependent on some node and at the same time, 
		\item the dependency set of $v$ is not the empty set.
	\end{enumerate*}
	Additionally, the status of a node $v$ might be recalculated by the procedure \tref{upd}, under the condition that $\sts_v = \tempsat$. We assume that the status of $v$ is reassigned again with the value \tempsat in line \ref{upd: TempSat}. Since lines \ref{upd: while_Unsat}-\ref{upd: propag_Unsat} are avoided, we can conclude that no eventuality is $\rat$-unfulfilled. Thus, each eventuality of $v$ is $\rat$-fulfilled or is $\rat$-dependent on some node.
	
	Finally, we show that the algorithm keeps the status and the dependency set of a node up-to-date. 
	Suppose that $\varphi$ is an active eventuality of a node $v$ which is $\rat$-fulfilled through a sequence $(v_1,\vartheta_1) \rat\dotsb\rat (v_k,\vartheta_k)$ of related pairs, where $v_1=v$ and $\vartheta_1 = \varphi$. Moreover, we assume that there is a node $v_i$, where $1 <i\leq k$, such that $\sts_{v_i} = \tempsat$. We additionally assume that the status of $v_i$ changes. This has as a result to possibly affect the $\rat$-fulfillment of $\varphi$. 
	We show that the status of $v$ is updated accordingly. The status of $v_i$ can change only after the definition of the status of a node (see lines \ref{ct: contr}-\ref{ct: sts state} of \tref{ct}) through line \ref{ct: update} of \tref{ct} with the use of \tref{upd}. This procedure affects only the nodes whose dependency sets contain a specific node. By Lemma \ref{lem: dependency set}, we know that $D_{v_i}\setminus\{v\}$ is a subset of $D_v$. Thus, the node $v$ is also updated by the same procedure call of \tref{upd}.
	According to lines \ref{upd: while_Unsat}-\ref{upd: propag_Unsat} of \tref{upd}, if some active eventuality of $v_i$ is $\rat$-unfulfilled, then $\sts_{v_i}$ is assigned the value \unsat. Now, by the same lines of \tref{upd}, since $\varphi$ cannot be $\rat$-fulfilled anymore, the status of $v$ is also defined accordingly.
\end{proof}

\begin{lem}		\label{lem: sat children}
	For any node $v$ of a tableau $G=(V,E)$ such that $\sts_v = \sat$, 
	if it is a partial node, then it has at least one child such that its status has also been assigned the value \sat, and 
	if it is a state node, then the statuses of all its children have also been assigned the value \sat.
\end{lem}
\begin{proof}
	We consider a node $v$ such that $\sts_v = \sat$ and we show that the properties of the lemma hold. 
	Have in mind that the status of a node is initially defined by one of the procedures \tref{csp} and \tref{css} and then it may be updated several times with the use of \tref{upd} due to line \ref{ct: update} of \tref{ct}. By Lemma \ref{lem: sat unsat}, we know that $\sts_v$ is eventually assigned one of \sat and \unsat.
	
	
	We consider a node $v$ and we examine all the cases on how $\sts_v$ is assigned the value \sat. First, we assume that one of the procedures \tref{csp} and \tref{css} directly assigns the value \sat to $\sts_v$.
	By the way that these procedures work, since $\sts_v = \sat$, we can conclude that $D_v$ is defined as the empty set. Moreover, by Lemma \ref{lem: dependency set}, we can conclude that all the active eventualities of $v$ are $\rat$-fulfilled, as otherwise $D_v$ would not be the empty set.
	Both procedures consider the necessary and appropriate children of $v$, i.e only those whose status has not been assigned the value \unsat. Since $D_v = \emptyset$, by the way that the procedures define $D_v$, we can conlude that there is a child $v_c$ of $v$ for which either $D_{v_c} = \emptyset$ or $D_{v_c} = \{v\}$.
	In the first case, by Lemma \ref{lem: sts dep}, we have that $\sts_{v_c} = \sat$.
	In the second case, by the same lemma, we have to admit that $\sts_{v_c} = \tempsat$. By Lemma \ref{lem: dependency set}, since $v$ is in $D_{v_c}$, we know that $v$ has at least one cyclic parent. Thus, after the definition of the status of $v$, the necessary nodes are updated by the procedure call $\tref{upd}(v)$ (see line \ref{ct: update} of \tref{ct}).
	According to line \ref{upd: propag_Sat_Unsat}, since $\sts_{v} = \sat$, the procedure \tref{propag} is used and due to $\sts_v$, it assigns exclusively the value \sat. So, since $D_{v_c} = \{v\}$ and $\sts_{v_c} = \tempsat$, the node $v_c$ might be updated  and hence, we have that $\sts_{v_c}=\sat$.
	
	Next, we consider that $v_c$ is updated through lines \ref{upd: while_Unsat}-\ref{upd: TempSat}.
	By Lemma \ref{lem: tempsat rat}, we know that each active eventuality of $v_c$ is $\rat$-fulfilled or is $\rat$-dependent on some node.
	Since $D_{v_c}$ was initially the singleton $\{v\}$, by Lemma \ref{lem: dependency set}, an active eventuality of $v_c$ could be dependent only on $v$. As $\sts_v = \sat$ and all the active eventualities of $v$ are $\rat$-fulfilled, we conclude that all the active eventualities of $v_c$ are also $\rat$-fulfilled. Hence, lines \ref{upd: while_Unsat}-\ref{upd: propag_Unsat} are definitely not used. 
	Moreover, for the same reasons, no node whose dependency set contains $v$ could have an active eventuality which is $\rat$-unfulfilled. 
	Thus, $v_c$, as well as all the nodes whose dependency sets contain $v$, is updated through lines \ref{upd: while_Sat_TempSat}-\ref{upd: TempSat}. According to line \ref{upd: updateDep}, since $D_{v_c}=\{v\}$ and $D_v=\emptyset$, $D_{v_c}$ is also defined as the empty set. Therefore, the value \sat is assigned to $\sts_{v_c}$ in line \ref{upd: Sat}. 
	Notice that some other node might be updated first through the same lines of \tref{upd}. The value \sat is assigned to the status of any node $v'$ with the same initial properties with those of $v_c$ (i.e. $\sts_{v'}=\tempsat$ and $D_{v'} = \{v\}$) and then, the procedure \tref{propag} is used which assigns exclusively the value \sat (as $\sts_{v'}=\sat$). Thus, the node $v_c$ might be updated by \tref{propag} with the same result.
	
	
	Now, we assume that $\sts_v$ was initially assigned the value \tempsat and eventually, it is assigned the value \sat by the procedure \tref{propag} in one of lines \ref{upd: propag_Sat_Unsat} and \ref{upd: propag_Sat} of \tref{upd}. Observe that \tref{propag} assigns the value \sat iff the properties of the lemma are satisfied.
	Next, we assume that the value \sat is assigned to $\sts_v$ by the procedure $\tref{upd}$ in line \ref{upd: Sat}. More specifically, after the definition of the status of a node $v_0$, the procedure call $\tref{upd}(v_0)$ is triggered in line \ref{ct: update} of \tref{ct} which has led to the definition of $\sts_v$. 
	Thus, we can conclude that $\sts_v$ was initially assigned the value \tempsat and $v_0$ was in $D_v$. Moreover, since $\sts_v$ is eventually assigned the value \sat, by Lemma \ref{lem: sts dep}, we can conclude that $D_v$ is defined as the empty set in line \ref{upd: updateDep} of \tref{upd}. 
	Finally, since the update process was triggered after the definition of the status of $v_0$ and in line \ref{upd: updateDep} the node $v_0$ is removed from $D_v$ and the nodes of $D_{v_0}$ are added to it, we can conclude that $D_v$ was initially the set $\{v_0\}$ and $D_{v_0}$ was empty.
	
	We know that the first time that $\sts_v$ was defined by one of the procedures \tref{csp} and \tref{css}, the value $\tempsat$ has been assigned to $\sts_v$. 
	For the definition of $D_v$, both procedures consider those children of $v$ whose statuses have not been defined yet or have been assigned the value \tempsat. 
	So, after the definition of $D_v$, for each child $v_c$ of $v$, if $D_{v_c} = \bot$, then $v_c$ is in $D_v$, or else the set $D_{v_c} \setminus \{v\}$ is a subset of $D_v$.
	We show that this property is maintained throughout the construction of the tableau.
	
	Let $v_c$ be a child of $v$ whose status and dependency set have not been defined yet ($v_c$ is a cyclic child of $v$ and a forward ancestor of $v$) and thus, $v_c$ is in $D_v$. 
	After the definition of the status of $v_c$, $v$ is updated by the procedure \tref{upd} (due to line \ref{ct: update} of \tref{ct}). 
	If $\sts_{v_c}= \sat$, then we have shown what we wanted. 
	If $\sts_{v_c} = \unsat$ and at the same time, 
	the procedure \tref{propag} assigns the value \unsat to $\sts_v$, then we have reached a contradiction. Thus, if $v$ is partial node, then there is at least one child whose status has not been assigned the value \unsat and if $v$ is a state, then all of its children have not been assigned the value \unsat.
	If $\sts_{v_c} = \tempsat$, then $v_c$ is removed from $D_v$ and the nodes of $D_{v_c}\neq\emptyset$ are added to $D_v$. Thus, $D_v$ cannot be the empty set and as a result, it cannot be assigned the value \sat. Notice that by Lemma \ref{lem: dependency set}, $v$ cannot be in $D_{v_c}$, as $v_c$ is a forward ancestor of $v$ (and not the other way around). So, we have that $D_{v_c} \setminus \{v\} \subseteq D_v$.
	
	Now, we consider a child $v_c$ of $v$ such that $\sts_{v_c} = \tempsat$ and $D_{v_c} \setminus \{v\}$ is a subset of $D_v$.	
	By Lemma \ref{lem: dependency set}, we know that all the nodes of $D_{v_c}$ are forward ancestors of $v_c$ and their statuses have not been defined at the moment of the definition of $D_{v_c}$ and $\sts_{v_c}$.
	We also consider a node $v'$ which is in $D_{v_c}$ (notice that $v'$ could even be the node $v$). After the definition of its status, by Lemma \ref{lem: dependency set}, the procedure \tref{upd} is triggered by line \ref{ct: update} of \tref{ct}.
	
	If $\sts_{v'}=\sat$, then the procedure \tref{propag} is used (see line \ref{upd: propag_Sat_Unsat} of \tref{upd}). In this case, if $\sts_{v_c}$ is assigned the value \sat, then we have shown what we wanted. 
	If $\sts_{v'}=\unsat$, then \tref{propag} is used again. If the previous procedure assigns the value \unsat to $v_c$ and at the same time, it also assigns the same value to $\sts_v$, then we have reached a contradiction. Thus, if $v$ is partial node, then there is at least one child whose status has not been assigned the value \unsat and if $v$ is a state, then all of its children have not been assigned the value \unsat.
	Notice that we reach the same conclusions if the value \unsat is assigned to $\sts_{v_c}$ due to lines \ref{upd: while_Unsat}-\ref{upd: propag_Unsat} of \tref{upd}.
	
	On the other hand, if $v_c$ is not affected by lines \ref{upd: propag_Sat_Unsat}-\ref{upd: propag_Unsat} of \tref{upd}, then lines \ref{upd: while_Sat_TempSat}-\ref{upd: TempSat} are used, as $v'$ is in $D_{v_c}$ and $\sts_{v_c} = \tempsat$. 	
	According to line \ref{upd: updateDep}, $v'$ is removed from $D_{v_c}$ and all the nodes of $D_{v'}$ are added to $D_{v_c}$. Notice that if $v'=v$, then since $D_{v_c} \setminus \{v\}$ is a subset of $D_v$ we have that $D_{v_c} = D_v$.
	If $v'\neq v$, then $v'$ is also in $D_v$. Again due to line \ref{upd: updateDep}, we have that $v'$ is removed from $D_v$ and all the nodes of $D_{v'}$ are added to $D_v$. Thus, the desirable property, that is $D_{v_c} \setminus \{v\}$ is a subset of $D_v$, remains true.
	
	So, we have shown that during the construction of the tableau, before the assignment of $\sts_v$ with the value \sat, there is at least one child $v_c$ of $v$ whose status is assigned the value \tempsat and if $v$ is a state, then the statuses of its children have not been assigned the value \unsat. Notice that by Lemma \ref{lem: sts dep}, we have that $D_{v_c}\neq \emptyset$. Moreover, for each child $v_c$ of $v$, if $\sts_{v_c}=\tempsat$, then $D_{v_c} \setminus \{v\}$ is a subset of $D_v$. 
	Finally, since the status of $v$ has been defined, the dependency set of any node cannot contain $v$. Thus, for each child $v_c$ of $v$, if $\sts_{v_c}=\tempsat$, then $D_{v_c} \subseteq D_v$. 
	
	Now, recall that the value $\sat$ is assigned to $\sts_v$, after the definition of $v_0$, by the procedure \tref{upd} (see line \ref{upd: Sat}). 
	A necessary condition for the previous assignment is that $D_v$ is defined as the empty set in line \ref{upd: updateDep}. Recall that $D_v$ was initially the singleton $\{ v_0 \}$. Since $D_v$ is updated to the set $(D_v \setminus \{v_0\})  \cup D_{v_0}$, we can conclude that $D_{v_0}$ is the empty set.
	Moreover, all the active eventualities of $v$ are $\rat$-fulfilled, as otherwise $\sts_v$ would be assigned the value \unsat due to lines \ref{upd: while_Unsat}-\ref{upd: unsat}.
	Notice that the fulfillment relation $\rat$ for the active eventualities of $v$ is defined by one of the procedures \tref{csp} and \tref{css} by considering its appropriate children (recall also Lemma \ref{lem: rat}). So, the active eventualities of $v$ are $\rat$-fulfilled through a child of $v$ whose active eventualities are also $\rat$-fulfilled and whose status has been assigned the value \tempsat. 
	We denote this child as $v_c$. We have shown that $D_{v_c}\neq\emptyset$ is a subset of the initial dependency set of $v$ which was the singleton $\{v_0\}$. Therefore, we have that $D_{v_c} = \{v_0\}$. 
	Since $\sts_{v_c} = \tempsat$ and $D_{v_c} = \{v_0\}$, the status and the dependency set of $v_c$ are also updated by lines \ref{upd: while_Sat_TempSat}-\ref{upd: TempSat}.
	Again, lines \ref{upd: while_Unsat}-\ref{upd: unsat} are avoided, as we have shown that the active eventualities of $v_c$ are $\rat$-fulfilled. So, by the previous properties, we have that $D_{v_c}$ is defined as the empty set and $\sts_{v_c}$ is also assigned the value \sat.
\end{proof}

\begin{lem} \label{lem: rat fulfillment}
	Let $G=(V,E)$ be a tableau and $v$ a node such that $\sts_v = \sat$. After the assignment of the value \sat to $\sts_v$ and after the update of the appropriate nodes by the procedure \tref{upd}, each active eventuality $\varphi$ of $v$ is $\rat$-fulfilled through a sequence $(v_1,\vartheta_1)\rat\dotsb\rat(v_k,\vartheta_k)$ of $\rat$-related pairs, such that $v_1=v$ and $\vartheta_1=\varphi$ and for $i=1,\dotsc,k$, $\sts_{v_i} = \sat$.
\end{lem}
\begin{proof}
	We proceed by induction on the inductive definition of $G$. A tableau for a formula $\varphi$ can be viewed as the graph which has as root a node whose label has as only active formula the input formula $\varphi_0$ (see \tref{isSat}) and then, according to the procedure call $\tref{isSat}(\varphi_0)$, it is inductively defined either by expanding a node, or by (re)defining the status and the dependency set of a node. 
	In the sequel, we consider a node $v$ whose status has been assigned the value \sat and we examine all the cases that concern the definition of its status to show that its active eventualities are $\rat$-fulfilled. Additionally, we assume that the lemma holds for all the nodes whose statuses have been defined and have been assigned the value \sat before the status of $v$.

	
	First, we assume that one of the procedures \tref{csp} and \tref{css} directly assigns the value \sat to $\sts_v$. Both procedures assign the value \sat iff the following conditions hold: 
	\begin{enumerate*} [i)]
		\item each one of the active eventualities of $v$ is $\rat$-fulfilled or it is $\rat$-dependent on some node and 
		\item the dependency set of $v$ is the empty set.
	\end{enumerate*}
	By Lemma \ref{lem: dependency set}, we can conclude that no active eventuality is $\rat$-dependent on some node, as otherwise $D_v$ would not be the empty set. Therefore, all the active eventualities of $v$ are $\rat$-fulfilled.
	
	We have shown that an active eventuality $\varphi$ of $v$ is $\rat$-fulfilled through a sequence $(v_1,\vartheta_1)\rat\dotsb\rat(v_k,\vartheta_k)$ of related pairs, where $v_1=v$ and $\vartheta_1 = \varphi$. By Definition \ref{def: fulfilled dependent}, we know that for $i=1,\dotsc,k$, $\sts_{v_i} \in \{\sat,\tempsat\}$. At the moment of the definition of $\sts_v$, we assume that there is a node $v_i$, where $1< i\leq k$, such that $\sts_{v_i} = \tempsat$ and we show that $\sts_{v_i}$ is eventually assigned the value \sat. By Lemma \ref{lem: sts dep}, we know that $D_v = \emptyset$ and $D_{v_i} \neq \emptyset$ and by the third property of Lemma \ref{lem: dependency set}, that the set $D_{v_i} \setminus \{v\}$ is a subset of $D_v$. 
	Hence, $D_{v_i}$ is the singleton $\{v\}$, as $D_v=\emptyset$. So, by Lemma \ref{lem: dependency set}, $v$ has at least one cyclic parent and the procedure \tref{upd} is used after the definition of $\sts_v$. Moreover, since $\sts_{v_i} = \tempsat$ and $D_{v_i} = \{v\}$, we know that $v_i$ is updated. 
	We need to show that lines \ref{upd: while_Unsat}-\ref{upd: propag_Unsat} are avoided, i.e. no active eventuality of $v_i$ is $\rat$-unfulfilled. All the nodes are up-to-date except for those whose dependency sets contain $v$ and thus, the procedure call $\tref{upd}(v)$ is used. By Lemma \ref{lem: tempsat rat}, we know that all the active eventualities of $v_i$ are $\rat$-fulfilled or $\rat$-dependent. Since $\sts_v=\sat$, no active eventuality of $v_i$ can be affected so that it would be determined as $\rat$-unfulfilled.
	Moreover, by Lemma \ref{lem: dependency set}, if some active eventuality of $v_i$ is $\rat$-dependent on a node, then that node is $v$, as $D_{v_i}=\{v\}$. As the active eventualities of $v$ are $\rat$-fulfilled, we have that the active eventualities of $v_i$ are also $\rat$-fulfilled. So, $\sts_{v_i}$ is assigned the value \sat.
	
	
	Besides the cases that concern the procedures \tref{csp} and \tref{css}, the status of $v$ may be assigned the value \sat by the procedure \tref{upd}. The necessary conditions for the update of $v$ are the initial assignment of the value \tempsat to $\sts_v$ and the existence of a node $v_0$ in $D_v$. After the definition of $\sts_{v_0}$, the procedure call $\tref{upd}(v_0)$ is triggered in line \ref{ct: update} of \tref{ct} which leads to the assignment of the value \sat to $\sts_v$.
	First, we assume that the value \sat is assigned to $\sts_v$ by line \ref{upd: Sat} of \tref{upd}. Since $\sts_v=\sat$, $D_v$ is defined as the empty set in line \ref{upd: updateDep}. Hence, we have to admit that $D_v$ was initially the singleton $\{v_0\}$ and at the same time, $D_{v_0} = \emptyset$. As lines \ref{upd: while_Unsat}-\ref{upd: propag_Unsat} are avoided, we conclude that no eventuality of $v$ is $\rat$-unfulfilled. Moreover, as $D_v = \emptyset$, by Lemma \ref{lem: dependency set}, we also conclude that no eventuality of $v$ is $\rat$-dependent on a node. Therefore, all the active eventualities of $v$ are $\rat$-fulfilled. 
	
	We have shown that an active eventuality $\varphi$ of $v$ is $\rat$-fulfilled through a sequence $(v_1,\vartheta_1)\rat\dotsb\rat(v_k,\vartheta_k)$ of related pairs, where $v_1=v$ and $\vartheta_1 = \varphi$. By Definition \ref{def: fulfilled dependent}, we know that for $i=1,\dotsc,k$, $\sts_{v_i} \in \{\sat,\tempsat\}$. 
	At the moment of the definition of $\sts_v$, we assume that there is a node $v_i$, where $1< i\leq k$, such that $\sts_{v_i} = \tempsat$ and we show that $\sts_{v_i}$ is eventually assigned the value \sat. 
	By Lemma \ref{lem: sts dep}, we know that $D_v = \emptyset$ and $D_{v_i} \neq \emptyset$ and by the third property of Lemma \ref{lem: dependency set}, that the set $D_{v_i} \setminus \{v\}$ is a subset of $D_v$. 
	Hence, $D_{v_i}$ is the singleton $\{v_0\}$, as $D_v$ was initially the singleton $\{v_0\}$. Thus, $v_i$ is updated after the definition of the status of $v_0$, as $v$ was updated. 
	We need to show that lines \ref{upd: while_Unsat}-\ref{upd: propag_Unsat} are avoided, i.e. no active eventuality of $v_i$ is $\rat$-unfulfilled. These lines are avoided as they have been avoided in the case of $v$. Do not forget that $\varphi$ of $v$ is $\rat$-fulfilled through $v_i$ and if $\sts_{v_i}$ was assigned the value \unsat due to some $\rat$-unfulfilled eventuality, then $\varphi$ would also be $\rat$-unfulfilled. Thus, all the active eventualities of $v_i$ are also $\rat$-fulfilled. So, the status of $v_i$ is assigned the value \sat.
	
	
	Next, we assume that the value \sat is assigned to $\sts_v$ by the procedure \tref{propag} in one of lines \ref{upd: propag_Sat_Unsat} and \ref{upd: propag_Sat} of \tref{upd}.
	The previous assignment takes place in line \ref{upd: propag_Sat_Unsat} (through the procedure call $\tref{propag}(v_0,v_0)$), only if we assume that $\sts_{v_0}=\sat$. 
	As argued earlier, since $\sts_{v_0}=\sat$ has been defined by one of \tref{csp} and \tref{css}, we can conclude that the active eventualities of $v_0$ are $\rat$-fulfilled. 
	Additionally, $v$ might also be updated through another node $v'$ which initially satisfied the same properties as $v$ (i.e. $\sts_{v'} = \tempsat$ and $v_0 \in D_{v'}$). The node $v'$ is updated through lines \ref{upd: while_Sat_TempSat}-\ref{upd: TempSat} and the value \sat is assigned to $\sts_{v'}$ in line \ref{upd: Sat} (as $D_{v'}= \emptyset$). As argued earlier, the active eventualities of $v'$ are $\rat$-fulfilled. We assume that $v$ is updated by the procedure call $\tref{propag}(v_0,v')$ in line \ref{upd: propag_Sat}.
	
	In both previous cases, there is a node $v_1$ whose status has been assigned the value $\sat$ and its active eventualities are $\rat$-fulfilled. Also, it has triggered the update of the status of $v$ through the procedure \tref{propag}. Thus, there is a sequence $v_1,\dotsc,v_k$ of nodes with $k\geq2$ such that $v_k=v$ and for $i=2,\dotsc,k$,
	$v_i$ is a parent of $v_{i-1}$ and  
	the status $\sts_{v_i}$ was initially assigned the value $\tempsat$ and $v_0$ was in $D_{v_i}$ and eventually, after the application of \tref{propag}, we have that $\sts_{v_i}=\sat$ and $D_{v_i}=\emptyset$. 
	Moreover, according to \tref{propag}, for $i=2,\dotsc,k$, if $v_i$ is a partial node, then the status of at least one of its children has been assigned the value \sat (notice that $\sts_{v_{i-1}} = \sat$) and if $v_i$ is a state node, then the statuses of all of its children have been assigned the value \sat.
	
	We know that all the active eventualities of $v_1$ are $\rat$-fulfilled. We show that the same also holds for the active eventualities of $v_k$.
	First, we assume that $v_k$ is a partial node. Have in mind that the status of $v_k$ was initially assigned the value \tempsat. So, by the way that the procedure \tref{csp} defines the relation $\rat$, we can conclude that for each active eventuality $\psi$ of $v_k$, we have that $(v_k,\psi)\rat(v_{k-1},\psi')$, where $\psi'$ is some formula of $v_{k-1}$.
	By Lemma \ref{lem: rat}, we know that there are two cases. In the first, $\psi'$ is the necessary non-eventuality in order for $\psi$ to be $\rat$-fulfilled. In the second case, $\psi'$ is an active eventuality of $v_{k-1}$. Since $\sts_{v_{k-1}}$ is assigned the value \sat before $v_k$, by the induction hypothesis, we know that $\psi'$ is $\rat$-fulfilled. Therefore, since $(v_k,\psi) \rat (v_{k-1},\psi')$, $\psi$ is also $\rat$-fulfilled. 
	
	Now, we assume that $v_k$ is a state. An active eventuality $\varphi$ of $v_k$ is some formula $\neg\forall A.\xi$ with $A\in\atp_\Omega$. By the way that \tref{ansr} defines the children of a state, we know that there is a child $v'$ of $v_k$ such that $\neg\xi$ is an active eventuality of $v'$. By the way that \tref{css} defines the fulfillment relation $\rat$, we know that $(v_k,\varphi) \rat (v',\neg\xi)$. No matter whether $v'=v_{k-1}$ or not, $v'$ is a child of $v_k$ whose status has been assigned the value \sat before $v_k$ (otherwise, \tref{propag} would not assign the value \sat to $\sts_{v_k}$) and by the induction hypothesis, we know that $\neg\xi$ is $\rat$-fulfilled. Since $(v_k,\varphi)\rat (v',\neg\xi)$, we can conclude that $\varphi$ is also $\rat$-fulfilled.
\end{proof}

\subsection{Soundness} 		\label{subsec: soundness}

\begin{defi} \label{def: fat path} 
The \textit{fat path} $\fp(G)=(V',E')$ of a tableau $G=(V,E)$ such that the status of its root node has been assigned the value \sat is a subgraph of $G$ and it is inductively defined as follows:
\begin{itemize}
\item The root node of $G$ belongs to the fat path $\fp(G)$.
\item If $v\in V'$ is a partial node, then all of its children $v'\in V$ such that $\sts_{v'} = \sat$ are added to $\fp(G)$ as children of $v$ with the same type of edges as in $G$.
\item If $v\in V'$ is a state, then all of its children are added to $\fp(G)$ as its children with the same type of edges as in $G$.
\end{itemize}
Finally, we say that a partial node $v'$ \textit{ends up at} a state $v$ in a fat path $\fp(G)=(V,E)$ iff there is a sequence $v_1,\dotsc,v_k$ of nodes of $V$ with $k>1$ such that $v_1=v'$, $v_k=v$ and for $i=1,\dotsc,k-1$, $v_i$ is a partial node and $(v_i,v_{i+1})\in E$.
\end{defi}

\begin{lem}	\label{lem: fat path properties} 
Let $\fp(G)$ be the fat path of a tableau $G$.
The status of each node of $\fp(G)$ has been assigned the value \sat and each partial node of $\fp(G)$ has at least one child node.
\end{lem}
\begin{proof}
The proof is an immediate consequence of Definition \ref{def: fat path} and of Lemma \ref{lem: sat children}.
\end{proof}

\begin{defi} 		\label{def: tableau path}
Let $G=(V,E)$ be a tableau or a subgraph of a tableau, $v,v'\in V$ nodes and formulas $\varphi\in\Phi_v$ and $\varphi'\in\Phi_{v'}$.
A \emph{tableau path} $\tp(v,\varphi,v',\varphi')$ in $G$ is a sequence $(v_0,\psi_0), \dotsc, (v_k,\psi_k)$ of pairs of nodes and formulas with $k\geq0$ such that $v_0 = v$, $\psi_0=\varphi$, $v_k=v'$, $\psi_k=\varphi'$ and for $i =0,\dotsc,k$, $\psi_i \in \Phi_{v_i}$ and for $i=0,\dotsc,k-1$,
	\begin{enumerate}
	\item if $v_{i+1} \!= v_i$, then $\psi_i\in\rdf_{v_i}$ and $\psi_i$ is fully reduced to $\psi_{i+1}$ within $v_{i+1}$, in symbols $\psi_i\vtrd_{v_{i+1}}\!\psi_{i+1}$,
	
	\item if $v_i\neq v_{i+1}$ and $v_i$ is a partial node, then $(v_i,v_{i+1})\in E$ and either $\psi_i \vtrd_{v_{i+1}} \psi_{i+1}$, or $\psi_{i+1}=\psi_i$,
	
	\item if $v_i\neq v_{i+1}$ and $v_i$ is a state, then $(v_i,v_{i+1})\in E$ and $\psi_i=\neg\forall A.\chi \in\actf_{v_i}$ with $A\!\in\atp_\Omega$ and $\psi_{i+1}=\neg\chi$.
\end{enumerate}
\end{defi}

\begin{lem}	\label{lem: tableau path}
	Let $\fp(G)=(V',E')$ be a fat path of a tableau $G=(V,E)$.
	For any node $v$ of $\fp(G)$, for any active eventuality $\varphi = \neg\forall A_1.\dotsb\forall A_k.\forall A^*.\chi$ of $\lab_v$ with $k\geq0$ and $\neg\chi\notin\ev$,
	there is a tableau path $\tp(v,\varphi,v_1,\neg\chi)$ in $\fp(G)$, where $v_1$ is a node of $V'$.
\end{lem}
\begin{proof}
	Since $v$ is a node of $\fp(G)$, by Lemma \ref{lem: fat path properties}, we know that $\sts_v=\sat$. Hence, by Lemma \ref{lem: rat fulfillment}, we can conclude that $\varphi$ is $\rat$-fulfilled through nodes whose statuses have been assigned the value \sat (notice that all of them are in $\fp(G)$). Therefore, there is a sequence $(v_1,\psi_1) \rat (v_2,\psi_2) \rat\dotsb\rat (v_n,\psi_n)$ of $\rat$-related pairs of nodes of $\fp(G)$ and of formulas such that $v_1 = v$, $\psi_1=\varphi$ and $\psi_k=\neg\chi$. By the properties of Lemma \ref{lem: rat}, we can conclude that for $i=1,\dotsc,n-1$, $\psi_i$ is an eventuality of the form $\neg\forall B_1.\dotsb\forall B_r.\forall A^*.\chi$ with $r\geq0$.
	
	In the following, we argue that the sequence $(v_1,\psi_1),\dotsc,(v_n,\psi_n)$ is the necessary tableau path. We consider a pair $(v_i,\psi_i)$ with $1\leq i \leq n-1$, we examine all the necessary cases and we show that the subsequence $(v_i,\psi_i), (v_{i+1},\psi_{i+1})$ is compatible with the requirements of Definition \ref{def: tableau path}.
	
	We distinguish cases based on what kind of node $v_i$ is. 
	If $v_i$ is a state, then by Lemma \ref{lem: rat} we can conclude that 
		$v_{i+1}$ is a child of $v_i$, 
		$\psi_i$ is an active eventuality of $v_i$ of the form $\neg\forall C.\vartheta$ with $C\in\atp_\Omega$ and 
		$\psi_{i+1}$ is the formula $\neg\vartheta$.
	The pairs $(v_i,\psi_i),(v_{i+1},\psi_{i+1})$ are compatible with the requirements of Definition \ref{def: tableau path}.
	
	Now, we assume that $v_i$ is a partial node. By Lemma \ref{lem: rat}, we know that $v_{i+1}$ is a child of $v_i$. 
	The procedure \tref{csp} uses the set $\reach(\psi_i,\lab_{v_{i+1}})$ for the definition of $(v_i,\psi_i) \rat (v_{i+1},\psi_{i+1})$. More specifically, the formula $\psi_{i+1}$ belongs to the set $\reach(\psi_i,\lab_{v_{i+1}})$.
	By Definition \ref{def: reach}, we know that there is a sequence $\vartheta_1,\dotsc,\vartheta_m$ of formulas with $m\geq1$ such that $\vartheta_1=\psi_i$ and $\vartheta_m=\psi_{i+1}$ and for $j=1,\dotsc,m-1$, $\vartheta_j \in \rdf_{v_{i+1}}$ is a reduced $\alpha$/$\beta$ formula and $\vartheta_j\vtrd_{v_{i+1}}\vartheta_{j+1}$.
	In the case that $m=1$, we can conclude that $\psi_i$ and $\psi_{i+1}$ are the same formula which is also an active eventuality of $v_{i+1}$. Once more, the pairs $(v_i,\psi_i),(v_{i+1},\psi_{i+1})$ are compatible with the requirements of Definition \ref{def: tableau path}.
	In the case that $m>1$, since $\vartheta_1\vtrd\vartheta_2$, by Lemma \ref{lem: properties of vtrd}, we know that there are two cases for $\vartheta_2$. The first case is that it is the non-eventuality $\neg\chi$ (notice that $m=2$). In the other case, $\vartheta_2$ is an active eventuality of $v'$ of the form $\neg\forall B_1.\dotsb\forall B_r.\forall A^*.\chi$ with $r\geq1$ and $B_1\in\atp_\Omega$. Since $B_1\in\atp_\Omega$, $\vartheta_2$ cannot be an $\alpha$/$\beta$ formula. Again, we have that $m=2$. 
	In both cases, we have that $\vartheta_1\vtrd\vartheta_2$ and $m=2$ and hence, the requirements of Definition \ref{def: tableau path} are satisfied.
\end{proof}

\begin{thm}[Soundness] \label{thm: soundness}
	If the status of the root node of the tableau $G = (V, E)$ for a formula $\varphi_0$ was assigned the value \sat (i.e. the procedure call $\tref{isSat}(\varphi_0)$ returns true), then $\varphi_0$ is satisfiable.
\end{thm}
\begin{proof}
	By Lemma \ref{lem: Hintikka->Model}, we know that if there is a Hintikka structure for $\varphi_0$, then $\varphi_0$ is satisfiable. Hence, it is enough to show that there is a Hintikka structure for $\varphi_0$.
	We consider the fat path $\fp(G) = (V',E')$ of $G$ and based on it, we define a labelled structure $\mc{H} =\langle S, (\leads{A})_{A\in\atp_\Omega}, L \rangle$, as follows:
	\begin{itemize}
	\item $S$ is defined as the set of all the states of $\fp(G)$:	$S = \{ v \mid  v \text{ is a state of $\fp(G)$}\}$
	
	\item For any states $v,v'\in S$, for any $A \in\atp_\Omega$, $v \leads{A}v'$ iff there is a sequence $v_0, \dotsc, v_k$ of nodes of $\fp(G)$ with $k\geq1$ such that:
	\begin{enumerate*}[i)]
	\item $v_0 = v$ and $v_k = v'$,
	\item $(v_0,v_1)\in E'$ is an edge which is labelled with a formula of the form $\neg\forall A.\chi$,
	\item for $i=1,\dotsc,k-1$, $v_i$ is a partial node and $(v_i,v_{i+1})\in E'$, i.e. $v_i$ ends up at $v_k$ (recall Definition \ref{def: fat path}).
	\end{enumerate*}
	
	\item For each state $v \in S$, we define $L(v)  =  \Phi_v$.
	\end{itemize}

	To show that $\mc{H}$ is a Hintikka structure, we need to show that all conditions \ref{contr}-\ref{transitions} of Definition \ref{def: Hintikka structure} are fulfilled.
	

To show that \ref{contr} is satisfied, we need to show that for some state $v\in S$, if $\neg\varphi\in L(v)$, then $\varphi\notin L(v)$, where $\varphi$ is an atomic formula of $\atf$ or a capability statement $\cpi A$ such that $A\in\wt{\Sigma}$. Equivalently, we show that for any formula $\varphi$ of this particular form, the set $\{\varphi,\neg\varphi\}$ cannot be a subset of $L(v)$. 
Additionally, to show that \ref{cap1} is satisfied, we need to show that for any state $v$ in $S$, $L(v)$ does not contain capabilities statements of the form $\neg\cpi\varphi$. 
The formulas that $\ref{contr}$ and \ref{cap1} concern are not $\alpha$/$\beta$ ones and thus, they can only appear in the active set of a node.
Recall that $L(v)$ is defined as the set $\Phi_v$ and that by Lemma \ref{lem: fat path properties}, we know that the status of each node of $\fp(G)$ is different from the value \unsat. Thus, we conclude that $\actf_v$ does not have any contradictory formulas or formulas of the form $\neg\cpi\varphi$, as otherwise, according to lines \ref{ct: contr}-\ref{ct: unsat} of \tref{ct}, the status of $v$ would be assigned the value \unsat.


For condition \ref{alpha/beta}, we consider some state $v\in S$ and we show that for some $\alpha$/$\beta$ formula $\varphi$ in $L(v)$, one of its reduction sets is a subset of $L(v)$. 
Since $L(v) = \Phi_v$ and $v$ is a state, by Definition \ref{def: label}, $\actf_v$ does not contain $\alpha$/$\beta$ formulas and for each $\alpha$/$\beta$ formula in $\rdf_v$, there is a reduction set which is a subset of $\Phi_{v} = L(v)$.
Due to line \ref{ct: extrdf} of \tref{ct}, $\Phi_v$ might be extended with the reduced $\alpha$/$\beta$ formulas of another state with the same active set with that of $v$.  Since the active sets of the involved states are the same, the new label remains saturated.


Next, we show that condition \ref{cap2} is satisfied. In particular, we show that for a state $v\in S$, if a capability statement $\neg\cpi A$ with $A\in\wt{\Sigma}$ and $\imath\in I$ is in $L(v)$ and $\{ \cpi(\varphi_1\Ra\psi_1),\dotsc,\cpi(\varphi_k\Ra\psi_k) \}$ is the set of all the capabilities statements in $L(v)$ that concern precondition-effect processes for $\imath\in I$, then there is a state $v'\in S$ such that $\{ \varphi_1, \dotsc, \varphi_k, \neg\forall A.\forall(\neg\psi_1).\dotsb\forall(\neg\psi_k).\false \}$ is a subset of $L(v')$. 
Observe that the previously mentioned capabilities statements of that particular form are not $\alpha$/$\beta$ formulas and thus, they can only be active formulas of $v$. 
We show that \ref{cap2} is satisfied as a consequence of the procedure \tref{ansr} which applies the capability rule. 
First, notice that by the definition of $\fp(G)$, since $v$ is a state, all of its children in $G$ also belong to $\fp(G)$.
The procedure \tref{ansr} defines a child node $v_1$ of $v$ whose set of formulas is the set $\{ \varphi_1, \dotsc, \varphi_k, \neg\forall A.\forall(\neg\psi_1).\dotsb\forall(\neg\psi_k).\false \}$. 
If $v_1$ is a state, then we have shown what we wanted, as $v_1$ is a state of $S$ and $L(v_1) = \Phi_{v_1}$.
If $v_1$ is not a state, then it is enough to show that the partial node $v_1$ ends up at a state of $\fp(G)$. Have in mind that, by the definition of $\fp(G)$ and by Lemma \ref{lem: fat path properties}, each partial node in $V'$ has at least one child. So, by Lemma \ref{lem: saturation properties}, there is a sequence $v_1,\dotsc,v_k$ of nodes of $\fp(G)$ with $k>1$ such that $v_k$ is a state and for $i=1,\dotsc,k-1$, $v_i$ is a partial node and $(v_i,v_{i+1})\in E'$ and $\Phi_{v_i} \subseteq \Phi_{v_{i+1}}$. By the way that $\mc{H}$ has been defined, $v_k$ is in $S$ and $\Phi_{v_1}$ is a subset of $\Phi_{v_k} = L(v_k)$.


In the following, we show that \ref{box} is satisfied. Suppose that there are states $v,v'\in S$ such that $\forall a.\vartheta\in L(v)$ and $v \leads{a} v'$. We show that $\vartheta\in L(v')$. 
Since $v\leads{a}v'$, we know that there is a sequence $v_0,\dotsc,v_k$ of nodes of $\fp(G)$, with $k\geq1$, such that $v_1,\dotsc,v_{k-1}$ are partial nodes and $v_0 = v$ and $v_k = v'$ and $v_1$ ends up at the state $v_k$. 
Moreover, by Lemma \ref{lem: saturation properties}, we also know that for $i=1,\dotsc,k-1$, $\Phi_{v_i}$ is a subset of $\Phi_{v_{i+1}}$ and thus, $\Phi_{v_1}\subseteq \Phi_{v_k}$.
Also, since $v \leads{a} v'$, the edge $(v_0,v_1)\in E'$ is labelled with some formula $\neg\forall a.\chi$. Thus, the transitional rule was applied to $v_0$ due to $\neg\forall a.\chi\in \actf_v$ which led to the definition of $v_1$. According to \tref{ansr}, we have that $\vartheta\in\Phi_{v_1}$. By the definition of $S$, we have that $v_k$ is a state of $S$. Therefore, since $\Phi_{v_1}$ is a subset of $\Phi_{v_k} = L(v_k)$, $\vartheta$ is also in $L(v_k)$.


Next, we show that \ref{Omega} holds. We consider a state $v\in S$ such that $\forall\Omega.\vartheta$ is in $L(v)$. 
First, we assume that there is a process $A\in\atp_\Omega$ and a state $v'\in S$ such that $v\leads{A}v'$ and we show that $\vartheta\in L(v')$. This case is similar to \ref{box}, as it is a consequence of the transitional rule which is applied by \tref{ansr} due to a formula $\neg\forall A.\chi$ in $\Phi_v$.

Now, we deal with the first part of \ref{Omega}.
Recall that $\Omega$ is a syntactic element which was introduced (see Lemma \ref{lem: Omega}) to make precondition-effect processes easier to handle (see Corollary \ref{cor: =>2}). So, observe that a non-negated boxed formula with $\Omega$ occurs in a node of a tableau only as the result of the application of the static rule on a $\beta$-formula of the form $\forall(\chi_1\Ra\chi_2).\psi$. Hence, either the formula $\chi_1 \land \forall\Omega^*.\forall\chi_2.\psi$  or the formula $\forall\Omega^*.\psi$ is added to the partial node. Next, in both cases, again due to the static rule, formulas of the form $\forall\Omega.\forall\Omega^*.\xi$ occur, where $\xi$ is some appropriate formula. 

Now, suppose that there is a state $v\in S$ such that $\forall\Omega.\vartheta$ is in $L(v)$. From what we have mentioned so far, the existence of this kind of formula is justified only by the existence of a state $v'$ in $S$ such that a formula of the form $\forall(\chi_1\Ra\chi_2).\psi$ is in $L(v')$. It is possible that $v'$ and $v$ are the same state. As argued earlier, the formula $\forall(\chi_1\Ra\chi_2).\psi$ leads to a formula of the form $\vartheta = \forall\Omega^*.\varphi$ in $L(v')$ which in turn leads to $\forall\Omega.\Omega^*.\varphi$. 
Finally, having in mind \ref{alpha/beta} and the proven second part of \ref{Omega}, we can conclude that there is a sequence of transitions $s_1 \leads{A_1} s_2 \dotsb s_{k-1} \leads{A_{k-1}} s_k$ in $\mc{H}$, with $k\geq1$, such that $s_1=v'$ and $s_k=v$ and for $i=1,\dotsc,k$, $\forall\Omega^*.\varphi, \forall\Omega.\forall\Omega^*.\varphi \in L(s_i)$.


As far as \ref{diamond} is concerned, we assume that for a state $v\in S$, there is a formula $\neg\forall A.\chi$ in $L(v)$ such that $A\in\atp_\Omega$. We show that there is a state $v'\in S$ such that $v\leads{A}v'$ and $\neg\chi\in L(v')$. 
Similarly to \ref{cap2}, as a consequence of \tref{ansr}, we can argue that a child $v_1\in V'$ of $v$ is defined such that $\neg\chi\in\Phi_{v_1}$, and the edge $(v,v_1)\in E'$ is labelled with the formula $\neg\forall A.\chi$. 
Moreover, in the case of \ref{cap2} by using Lemma \ref{lem: saturation properties}, we have shown that there is a sequence $v_1,\dotsc,v_k$ of nodes of $\fp(G)$ such that 
either $k=1$ and $v_1=v_k$ is a state
or $k>1$ and $v_1$ is a partial node which ends up at the state $v_k$.
In both cases, by the definition of $\mc{H}$, $v_k\in S$ and $v \leads{A} v_k$ and $\neg\chi\in L(v_k)$.


Next, we argue that \ref{ev} holds. Suppose that for a state $v\in S$, the $\alpha$/$\beta$ eventuality $\varphi\!=\!\neg\forall A_1.\!\dotsb\forall A_k.\forall A^*\!.\chi$ with $k\geq0$ and $\neg\chi\notin\ev$ is in $L(v)$. We show that there is a structure path $\stp(v,\varphi,v'',\neg\chi)$ in $\mc{H}$, with $v''\in S$. 

We have shown in the case of \ref{alpha/beta} that there is a reduction set $\mc{R}$ of $\varphi$ which is a subset of $L(v)$. By Definition \ref{def: reduction sets}, we know that there is a pair $(\mc{T},\psi)$ in $\mc{FD}(\varphi)$ such that $\mc{R} = \mc{T}\cup\{\psi\}$. According to Definition \ref{def: vtrd}, since $\varphi$ is an $\alpha$/$\beta$ eventuality, we have that $\varphi\vtrd\psi$.
By Lemma \ref{lem: properties of vtrd}, we know that there are two cases for $\psi$. The first case is that it is the non-eventuality $\neg\chi$ and thus, the necessary structure path is the sequence $(v,\varphi),(v,\neg\chi)$.
In the other case, $\psi$ is an active eventuality of $v$ of the form $\neg\forall B_1.\dotsb\forall B_r.\forall A^*.\chi$ with $r\geq1$ and $B_1\in\atp_\Omega$. 
By Lemma \ref{lem: tableau path}, we know that there is a tableau path $\tp(v,\psi, v',\neg\chi)$, where $v'$ is some node of $\fp(G)$. 
In the following, we define the appropriate structure path by using the tableau path $\tp(v,\varphi, v',\neg\chi) = (v,\varphi), \tp(v,\psi, v',\neg\chi)$.

First, we extend $\tp(v,\varphi, v',\neg\chi)$ to ensure that the node of its last pair is a state of $\fp(G)$.
In the case of \ref{cap2}, we argued that a partial node ends up at a state of $\fp(G)$. Therefore, there is a sequence of nodes $v_1,\dotsc,v_k$ of $\fp(G)$ with $k\geq1$, such that $v_1=v'$ and $v_k$ is a state and for $i=1,\dotsc,k-1$, $v_i$ is a partial node and $(v_i,v_{i+1})\in E'$. 
In the following, we examine the tableau path $\tp(v,\varphi,v',\neg\chi)$, $(v_2, \neg\chi),\dotsc,(v_{k}, \neg\chi)$, denoted by $\tp(v,\varphi, v'',\neg\chi)$.

We define the structure path $\stp(v,\varphi, v'',\neg\chi)$, by modifying $\tp(v,\varphi,\allowbreak v'',\neg\chi)$ in the appropriate way. 
First, we replace the node of any pair of the tableau path $\tp(v,\varphi, v'',\neg\chi)$ which is not a state of $\fp(G)$ with the first state of $\fp(G)$ that is met in a subsequent pair of $\tp(v,\varphi, v'',\neg\chi)$. 

We say that the transitional rule is applied to a node of a pair of a tableau path of the form $(v_1,\neg\forall A.\psi)$ with $A\in\atp_\Omega$  iff the next pair has the form $(v_2,\neg\psi)$ such that $(v_1,v_2)\in E'$. 
In the sequel, we examine a subsequence $(v_1,\xi_1),\dotsc,(v_k,\xi_k)$ of $\tp(v,\varphi, v'',\neg\chi)$, with $k\geq2$. We assume that $v_1$ and $v_k$ are the only states of the subsequence in which the transitional rule is applied, except, of course, for $v_k$ in the case that $(v_k,\xi_k)$ is the last pair of $\tp(v,\allowbreak\varphi, v'',\neg\chi)$. 
Recall that $v_1$ and $v_k$ belong to $S$ and for $i=2,\dotsc,k-1$, we replace $v_i$ with $v_k$. 
Before applying these replacements, by the definition of a tableau path, we know that there is $j\in \{1,\dotsc,k-1\}$, such that the followings hold:
\begin{enumerate*}[i)]
\item for $i=2,\dotsc,j$, 
	$v_i$ is a partial node which ends up at the state $v_k$ and 
	$\Phi_{v_i}$ is a subset of $\Phi_{v_{i+1}}$ (see Lemma \ref{lem: saturation properties}) and 
	either $\xi_i \vtrd_{v_{i+1}} \xi_{i+1}$ (recall Definition \ref{def: fully reduced}), or $\xi_i=\xi_{i+1}$,
\item for $i=j+1,\dotsc,k-1$,	
	$v_i= v_{i+1} = v_k$ and $\xi_i \vtrd_{v_{i+1}} \xi_{i+1}$.
\end{enumerate*}
By the definition of $L$, we conclude that for $i=2,\dotsc,k-1$, $\Phi_{v_i}$ is a subset of $L(v_k)$.
Finally, by the definition of the relation $\leads{}$ in $\mc{H}$, if we assume that $\xi_1$ is some formula of the form $\neg\forall B.\vartheta$ with $B\in\atp_\Omega$, then we have that $v_1\leads{B}v_k$.

After the replacement of the appropriate nodes in $\tp(v,\varphi, v'',\neg\chi)$ with the appropriate states of $\mc{H}$, we merge any consecutive identical pairs that occur. From what we have mentioned in the previous paragraph, it is straightforward to conclude that the new sequence $(s_1,\varphi_1),\dotsc,(s_r,\varphi_r)$ that occurs satisfies the following properties: 
\begin{enumerate*}[i)]
\item $(s_1,\varphi_1) = (v,\varphi)$ and $(s_r,\varphi_r) = (v'',\neg\chi)$,
\item for $i=1,\dotsc,r$, $s_i$ is a state of $\mc{H}$ such that $\varphi_i\in L(s_i)$,
\item for $i=1,\dotsc,r-1$, if $\varphi_i = \neg\forall B.\vartheta$ with $B\in\atp_\Omega$, then $\varphi_{i+1}=\neg\vartheta$ and $s_i\leads{B}s_{i+1}$, else $\varphi_i\vtrd\varphi_{i+1}$ and $s_{i+1} = s_i$ and there is a reduction set $\mc{R}$ of ${\varphi_i}$ such that $\varphi_{i+1}\in\mc{R}$ and $\mc{R} \subseteq L(s_{i+1})$.
\end{enumerate*}
By Definition \ref{def: structure path}, the defined sequence forms the desired structure path $\stp(v,\varphi, v'',\neg\chi)$. 


While most of the properties have been shown to be true based on the way that the algorithm has built the tableau $G$ and on the way that we defined $\mc{H}$, \ref{transitions} might not be satisfied. Thus, in the following, we explicitly modify $\mc{H}$ in two steps so that the specific property is true.
First, for any two different processes $A_1$ and $A_2$ of $\atp_\Omega$, 
if there are states $s_1$ and $s_2$ in $S$ such that $s_1\leads{A_1}s_2$ and $s_1\leads{A_2}s_2$,
then we remove the transition $s_1\leads{A_2}s_2$ and we add a new state $\tilde{s}_2$ in $S$, called shadow state of $s_2$, with the following properties:
\begin{enumerate*}[i)]
\item $L(\tilde{s}_2) = L(s_2)$,
\item we define the transition $s_1\leads{A_2}\tilde{s}_2$.
\end{enumerate*}
Now, it is immediate that for any two different processes $A_1$ and $A_2$ of $\atp_\Omega$, there are not states $s_1$ and $s_2$ in $S$ such that $s_1\leads{A_1}s_2$ and $s_1\leads{A_2}s_2$. In other words, if $s_1\leads{A_1}s_2$ and $s_1\leads{A_2}s_2$, then $A_1=A_2$.
Finally, in order for the other properties to remain true and unaffected, we add the appropriate transitions in $\mc{H}$.
For any shadow state $\tilde{s}$ of a state $s$, for any process $A\in\atp_\Omega$, for any state $s'\in S$, if $s\leads{A}s'$, then we define $\tilde{s}\leads{A}s'$.


To conclude this proof, we need to show that there is a state $s$ of $\mc{H}$ such that the input formula $\varphi_0$ of \tref{isSat} is in $L(s)$.
Since the root of $G$ is in $\fp(G)$ (see Definition \ref{def: fat path}), by Lemmas \ref{lem: saturation properties} and \ref{lem: fat path properties} (as shown in similar cases in this proof), there is a state in $\fp(G)$ (as well as in $\mc{H}$) which contains $\varphi_0$.
\end{proof}

\subsection{Completeness} 		\label{subsec: completeness}

In the following, we say that a state $w$ of an $\mc{L}$-model $\mc{M}$ satisfies a label $\lab$ and write $w\models^\mc{M}\lab$ iff it satisfies all the formulas of $\Phi_\lab$. Also, we say that $w$ satisfies a node $v$ and write $w\models^\mc{M}\!v$ iff it satisfies the label $\lab_v$.

\begin{thm}[Completeness]
If a formula $\varphi_0$ is satisfied by a state of an $\mc{L}$-model $\mc{M}$,
then the status of the root of the tableau $G=(V,E)$ for $\varphi_0$ is assigned the value \sat (i.e. $\tref{isSat}(\varphi_0)$ returns true).
\end{thm}
\begin{proof}
	In the following, we define an appropriate subgraph of the tableau $G$ in which the root node is included and then, we show that the statuses of its nodes are all assigned the value \sat.
	However, first, we need to make a remark about the procedure \tref{asr}. The specific procedure uses the reduction sets of some active $\alpha$/$\beta$ formula of a partial node in order to define the necessary children. By Proposition \ref{prop: static sound}, we know that if the parent partial node is satisfiable, then at least one of the reductions sets of the principal $\alpha$/$\beta$ formula is also satisfiable. The question that arises is whether \tref{asr} defines the child which is the result of the addition of the satisfiable reduction set. Observe that if the value \sat has been assigned to the status of one of the already defined children, then \tref{asr} does not proceed to the definition of the remaining children.
	Nevertheless, in the aforementioned case, we can be certain that the status of the parent partial node is also assigned the value \sat. This is an immediate consequence of how the procedure \tref{csp} works. Due to the status value \sat of the last defined child, by Lemmas \ref{lem: sts dep} and \ref{lem: rat fulfillment}, we know that the dependency set of the child is the empty set and all of its active eventualities are $\rat$-fulfilled. Now, by \tref{csp}, we can conclude that the dependency set of the parent partial node is also defined as the empty set and by the way that the fulfillment relation $\rat$ is defined, all of its active eventualities are $\rat$-fulfilled. Thus, the procedure \tref{csp} assigns the value \sat. 
	
	To sum up, we can be certain that one of the following properties holds for a satisfiable partial node in a tableau: 
	\begin{enumerate*}[i)]
		\item it has been expanded with all of its possible satisfiable children (which are at least one), or
		\item its status has been assigned the value \sat.
	\end{enumerate*}
	Of course, the previous properties are not mutually exclusive.
	
	Now, we define a subgraph $G'$ of the tableau $G$ inductively, as follows:
	\begin{enumerate}[i)]
		\item The satisfiable root node of $G$ belongs to $G'$.
		
		\item Each satisfiable node $v$ of $G'$ is expanded by distinguishing cases on its type:
		\begin{itemize}
			\item $v$ is a partial node: if all its satisfiable children belong to $G$, then all of them are added to $G'$, 
			or else, its status has been assigned the value \sat and $v$ is not expanded, i.e. it is a leaf of $G'$.
			\item $v$ is a state node: all of its children are added to $G'$.
		\end{itemize}
	\end{enumerate}
	
	By the definition of $G'$, by the way that \tref{asr} and \tref{ansr} expand a node and by Propositions \ref{prop: static sound}, \ref{prop: trans sound} and \ref{prop: cap sound}, we can conclude that all the nodes of $G'$ are satisfiable. Furthermore, we know that each partial node has been expanded with all of its possible satisfiable children or its status has been assigned the value \sat.
	
	
	We show that the statuses of all the nodes of $G'$ take the value \sat. In the following, we assume that the status of a node $v$ of $G'$ is assigned the value \unsat and we reach a contradiction. Without loss of generality, we also assume that it is the first node of $G'$ which takes the status value \unsat.	
	Since all the nodes of $G'$ are satisfiable, contradictory formulas or even capabilities statements of the form $\neg\cpi \varphi$ cannot occur in $\Phi_v$. Thus, lines \ref{ct: contr}-\ref{ct: unsat} of \tref{ct} are avoided.
	Additionally, one of the reasons that the procedures \tref{csp}, \tref{css} and \tref{propag} assign the value \unsat to the status of a node is that there are the appropriate children with this status value. Since $v$ is the first node of $G'$ to which the status value \unsat is assigned, its children cannot affect the definition of its status in this way. 
	Also, recall that the status of a leaf partial node of $G'$ takes the value \sat. 
	
	
	According to the procedures \tref{csp}, \tref{css} and \tref{upd}, the last reason which can justify the assignment of the value \unsat to $\sts_v$ is that there is an active eventuality $\varphi = \neg\forall A_1.\dotsb\forall A_n.\forall A^*.\chi$ of $v$ with $n\geq0$ and $\neg\chi\notin\ev$ which is not $\rat$-fulfilled and at the same time, it is not $\rat$-dependent on any node.	
	As $v$ belongs to $G'$, we know that it is satisfied by some state $w$ of $\mc{M}$. Additionally, by Lemma \ref{lem: fpaths}, we know that there is a fulfilling path $\fpath(\mc{M},w,\varphi,\vtrd) = (w_1,\vartheta_1),\dotsc,(w_k,\vartheta_k)$ for the active eventuality $\varphi$ of $v$. Without loss of generality, we assume that $\neg\chi$ appears for the first time in the last pair of the sequence.
	We define a sequence $(v_1,\xi_1,j_1),\dotsc,(v_r,\xi_r,j_r)$ of triples in which $v_1 = v$ and $\xi_1 = \varphi$ and the involved nodes are nodes of $G'$. The sequence has the following properties: 
	\begin{enumerate}[i)]
		\item for $i=1,\dotsc,r-1$, we have that $\xi_i$ is an active eventuality of $v_i$ of the form $\neg\forall B_1.\dotsb\forall B_m .\forall A^*.\chi$ with $m\geq0$ and $(v_i,\xi_i) \rat (v_{i+1},\xi_{i+1})$,
		\item for $i=2,\dotsc,r$, $\sts_{v_i} = \tempsat$
		\item the indexes $j_1,\dotsc,j_r$ indicate pairs of $\fpath(\mc{M},w,\varphi,\vtrd)$ which satisfy the following properties: 
		\begin{itemize}
			\item $j_1 = 1$ and $j_r = k$ and for $i=1,\dotsc,r-1$, $j_i\leq j_{i+1}$ and if $j_i\neq j_{i+1}$, then $j_{i+1}= j_i +1$, 
			\item for $i=1,\dotsc,r$, $w_{j_i} \models^\mc{M} v_i$ and $\vartheta_{j_i} = \xi_i$.
		\end{itemize}
	\end{enumerate}
	
	The first triple of the sequence is the triple $(v_1,\xi_1,j_1) = (v,\varphi,1)$ and we know that $w_{j_1} = w_1 = w$ and $w_{j_1} \models^\mc{M} v_1$ and $\vartheta_{j_1} = \xi_1 = \varphi$. 
	Next, we consider a subsequence $(v_1,\xi_1,j_1),\dotsc,(v_i,\xi_i,j_i)$, where $1\leq i\leq r-1$, for which we assume that the above properties hold and we expand it with the next triple by maintaining the desirable properties.
	
	
	First, we assume that $v_i$ is a partial node. Moreover, we assume that the value \sat has been assigned to its status. Thus, by Lemma \ref{lem: rat fulfillment},  $\xi_i$ is $\rat$-fulfilled. Additionally, by the properties of the sequence, we know that $(v_1,\xi_1)\rat(v_2,\xi_2)\rat\dotsb\rat(v_i,\xi_i)$. Thus, the formula $\xi_1=\varphi$ is also $\rat$-fulfilled. This case cannot occur, as by assumption we know that $\varphi$ is not $\rat$-fulfilled.
	
	Now, we assume that \tref{asr} has expanded the partial node $v_i$ with all of its satisfiable children. By the definition of $G'$, we know that all of them are in $G'$. The procedure \tref{asr} expands $v$ by applying the static rule. In particular, there is an active $\alpha$/$\beta$ formula in $v_i$ which is regarded as the principal formula and for each one of its reduction sets an appropriate child is created. 
	Recall that we know that $w_{j_i} \models^\mc{M} v_i$ and $\vartheta_{j_i} = \xi_i$ and thus, by Proposition \ref{prop: static sound}, there is at least one child $v'$ of $v_i$ which is also satisfied by $w_{j_i}$.
	First, we assume that $\xi_i$ is not the principal formula and thus, we define $v_{i+1}$  as the node $v'$ and $\xi_{i+1}$ as the formula $\xi_i$. Since $w_{j_i} \models^\mc{M} v_{i+1}$ and $\vartheta_{j_i} = \xi_i = \xi_{i+1}$, we let $j_{i+1}$ be the index $j_i$.
	
	Now, we examine the case that $\xi_i$ is the principal formula of the static rule. Since $\xi_i = \vartheta_{j_i}$ is an $\alpha$/$\beta$ formula, by the properties of $\fpath(\mc{M},w,\varphi,\vtrd)$ (see Definition \ref{def: fpath}), we have that $w_{j_i} = w_{j_i+1}$ and $\vartheta_{j_i} \vtrd \vartheta_{j_i+1}$ and there is a reduction set $\mc{R}$ of $\vartheta_{j_i}$ such that $\vartheta_{j_i+1}\in\mc{R}$ and $w_{j_i} \models^\mc{M} \mc{R}$. Therefore, we define as $v_{i+1}$ the child of $v_i$ which is defined as the result of the addition of the formulas of $\mc{R}$ to the set $\Phi_{v_i}$ (see \tref{asr}). Additionally, we define as $\xi_{i+1}$ the formula $\vartheta_{j_i+1}$ which is a formula of $v_{i+1}$. Finally, since $w_{j_i+1} \models^\mc{M} v_{i+1}$ and $\vartheta_{j_i+1} = \xi_{i+1}$, we let $j_{i+1}$ be the index $j_i+1$.
	By Lemma \ref{lem: properties of vtrd}, since $\xi_i \vtrd \xi_{i+1}$, the formula $\xi_{i+1}$ is either an active eventuality of $v_{i+1}$ of the necessary form or the non-eventuality $\neg\chi$. Since $\neg\chi$ appears in the last pair of $\fpath(\mc{M},w,\varphi,\vtrd)$ and $\xi_{i+1} = \vartheta_{j_{i+1}}$, if $j_{i+1} < k$, then $\xi_{i+1}$ is an active eventuality of $v_{i+1}$. 
	
	By the previous properties and by the way that \tref{csp} works with input the node $v_i$, we have that $(v_i,\xi_i) \rat (v_{i+1},\xi_{i+1})$. By assumption, the status of $v_{i+1}$ cannot be assigned the value \unsat (at least at the moment of the definition of $\sts_v$), as $v$ is the first node of $G'$ whose status takes this value.
	First, we assume that $\sts_{v_{i+1}}$ is undefined. Since $(v_1,\xi_1)\rat\dotsb\rat(v_{i+1},\xi_{i+1})$, we conclude that the active eventuality $\varphi$ of $v$ is $\rat$-dependent on $v_{i+1}$. However, this cannot be true, as $\varphi$ is not $\rat$-dependent on any node. 
	Hence, the only remaining valid conclusion is that $\sts_{v_{i+1}}$ has been assigned the value \tempsat.
	
	
	Now, we assume that $v_i$ is a state and we define the triple $(v_{i+1},\xi_{i+1},j_{i+1})$.
	Since $\xi_i$ is in $\actf_{v_i}$, it is some formula $\neg\forall B.\psi$ with $B\in\atp_\Omega$. Due to $\xi_i = \neg\forall B.\chi$, \tref{ansr} defines a child of $v_i$ which has as active eventuality the formula $\neg\psi$. We define $v_{i+1}$ to be this child of $v_i$ and $\xi_{i+1}$ its active eventuality $\neg\psi$.
	We know that $w_{j_i}$ satisfies $v_i$ and of course $\xi_i$ which is the formula $\vartheta_{j_i}$. Since $\vartheta_{j_i} = \xi_i$ is the formula $\neg\forall B.\psi$ with $B\in\atp_\Omega$, by the properties of $\fpath(\mc{M},w,\varphi,\vtrd)$, we have that $w_{j_i} \onto{B} w_{j_i+1}$ and $\vartheta_{j_i+1} = \neg\psi$. Observe that $\vartheta_{j_i+1} = \xi_{i+1} = \neg\psi$ and by the way that \tref{ansr} defines $v_{i+1}$ and by Proposition \ref{prop: trans sound}, $w_{j_i+1}$ satisfies $v_{i+1}$. Hence, we define $j_{i+1}$ to be the index $j_i+1$.
	By the previous properties and by the way that \tref{css} works, similarly to the case of a partial node, we can show that $(v_i,\xi_i) \rat (v_{i+1},\xi_{i+1})$ and that $\sts_{v_{i+1}}$ has been assigned the value \tempsat.
	
	
	To conclude the construction of the sequence $(v_1,\xi_1,j_1),\dotsc,(v_r,\xi_r,j_r)$ with the necessary properties, we need to show that, eventually, a sequence of arbitrary finite length $r$ is defined such that the index $j_r$ indicates the $k$th pair of $\fpath(\mc{M},w,\varphi,\vtrd)$ (i.e. $j_r=k$).
	The rules that extent the sequence of triples increase the index $j_i$ by 1, except for the case in which $v_i$ is a partial node and at the same time, $\xi_i$ is not the principal formula of the static rule.
	Thus, since the extension of the sequence follows the construction of the tableau, it is enough to show that a partial node reaches a state and thus, the static rule has been applied for all the $\alpha$/$\beta$ formulas.
	Observe that this is a consequence of Lemma \ref{lem: saturation properties}.
	
	
	We have shown that there is a sequence $(v_1,\xi_1,j_1),\dotsc,(v_r,\xi_r,j_r)$ of triples such that $(v_1,\xi_1)\rat\dotsb\rat(v_r,\xi_r)$ and for $i=2,\dotsc,r$, $\sts_{v_i}=\tempsat$. Moreover, since $j_r = k$ and $\vartheta_{j_r} = \xi_r$, by the properties of $\fpath(\mc{M},w,\varphi,\vtrd)$, we can conclude that  $\xi_r$ is the non-eventuality $\neg\chi$. Hence, we can conclude that the active eventuality $\varphi$ of $v$ is $\rat$-fulfilled and thus, we have reached a contradiction. 
	Therefore, the status of no node of $G'$ can be assigned the value \unsat. Finally, by Lemma \ref{lem: sat unsat}, we conclude that the statuses of all the nodes of $G'$ are eventually assigned the value \sat.
\end{proof}

\subsection{Complexity} 		\label{subsec: complexity}

Here, we investigate the complexity of the satisfiability algorithm for \tpdl. We examine various characteristics of a tableau and of its nodes and we close the section with the complexity analysis of the algorithm.

\begin{defi} 		\label{def: closure}
	The closure $cl(\varphi) \subseteq \mc{L}_s$ of a formula $\varphi$ is defined as follows:
	\begin{itemize}
		\item $\varphi\in cl(\varphi)$
		\item if $\psi$ is an $\alpha$/$\beta$ formula of $cl(\varphi)$, then all of its reduction sets are subsets of $cl(\varphi)$.
		
		\item if $\neg\forall A.\chi\in cl(\varphi)$ with $A\in\atp_\Omega$, then $\neg\chi\in cl(\varphi)$.
		\item if $\forall A.\chi\in cl(\varphi)$ with $A\in\atp_\Omega$, then $\chi\in cl(\varphi)$.
		
		\item if $\neg\cpi(\chi_1\Ra\chi_2) \in cl(\varphi)$, then $\{ \neg\chi_1, \chi_2 \} \subseteq cl(\varphi)$.
		\item if $\cpi(\chi_1\Ra\chi_2) \in cl(\varphi)$, then $\{ \chi_1, \neg\chi_2 \} \subseteq cl(\varphi)$.
	\end{itemize}
\end{defi}

\begin{defi} 		\label{def: cpr}
	Let $\varphi$ be some formula,
	$\neg\cpi A$ a capability statement of $cl(\varphi)$ with $A\in\wt{\Sigma}$ and 
	$\Gamma$ a set $\{ \cpi(\varphi_1\Ra\psi_1),\dotsc, \cpi(\varphi_k\Ra\psi_k) \}$ of capabilities statements of $cl(\varphi)$ with preconditions-effects processes, with $k\geq0$.
	We define $\cpr(\varphi,\neg\cpi A,\Gamma)$ as the following set of formulas:
	\begin{equation}	\label{eq: cpr1}
	\big\{ \neg\forall A. \forall(\neg\psi_1).\dotsb\forall(\neg\psi_k).\false,\;\,  
	\neg\forall(\neg\psi_1).\dotsb\forall(\neg\psi_k).\false, \dotsc, \neg\forall(\neg\psi_k).\false,\;\, 
	\neg \false \big\}
	\end{equation}
	If $A$ is some process $\chi_1\Ra\chi_2$, we additionally consider that the four formulas of the following set also belong to the set $\cpr(\varphi,\neg\cpi A,\Gamma)$:
	\begin{equation}	\label{eq: cpr2}
	\begin{split}
	\big\{	\neg\forall(\neg\chi_1). \forall\Omega. \forall(\neg\psi_1).\dotsb\forall(\neg\psi_k).\false,\; 
			\neg \forall\Omega. \forall(\neg\psi_1).\dotsb\forall(\neg\psi_k).\false,&\\
			\neg\forall\Omega. \forall(\chi_2). \forall(\neg\psi_1).\dotsb\forall(\neg\psi_k).\false,\;
			\neg\forall(\chi_2). \forall(\neg\psi_1).\dotsb\forall(\neg\psi_k).\false							&\big\}
	\end{split}
	\end{equation}

	For a specific formula $\varphi$, we define $\cpr_\varphi$ as the set of all the different sets $\cpr(\varphi,\neg\cpi A,\Gamma)$ of formulas.
\end{defi}

\begin{lem} 		\label{lem: cpr properties}
Let $\Delta$ be the set $\cpr(\varphi,\neg\cpi A,\Gamma)$ of formulas. The following properties hold:
\begin{enumerate}[i)]
	\item There are no capabilities statements (even negated ones) in $\Delta$ and there are no formulas of the form $\forall A.\vartheta$ in $\Delta$.
	\item For each $\alpha$/$\beta$ formula $\xi$ in $\Delta$, the formulas of all of its reduction sets are in $cl(\varphi)\cup\Delta$.
	\item For each formula of the form $\neg\forall A.\chi$ in $\Delta$ with $A\in\atp_\Omega$, the formula $\neg\chi$ is in $\Delta$.
\end{enumerate}
\end{lem}
\begin{proof}
	The properties of the lemma are an immediate consequence of Definitions \ref{def: closure} and \ref{def: cpr}.
\end{proof}

\begin{lem} 		\label{lem: closure properties}
	For any formula $\varphi$, the following properties hold:
	\begin{enumerate}[i)]
		\item The cardinality of $cl(\varphi)$ is quadratically bounded by the size of $\varphi$ (recall Definition \ref{def: size}), i.e. $\sharp cl(\varphi)$ is in $\mc{O}(|\varphi|^2)$.
		\item For any set $\cpr(\varphi,\neg\cpi A,\Gamma)$, $\sharp \cpr(\varphi,\neg\cpi A,\Gamma)$ is in  $\mc{O}(\sharp cl(\varphi))$.
		\item The cardinality of $\cpr_\varphi$ is exponential in the size of $\varphi$, i.e. $\sharp \cpr_\varphi  \in  2^{\mc{O}(|\varphi|^2)}$.
	\end{enumerate}
\end{lem}
\begin{proof}
	The first property is by the properties of the Fischer-Ladner closure of a formula \cite{Hartonas1}. Recall that the reduction sets of an $\alpha$/$\beta$ formula are nothing more than the result of a series of applications of the usual $\alpha$/$\beta$ rules (see Definitions \ref{def: reduction sets} and \ref{def: decomposition}).
	The second property is an immediate consequence of the definition of a set $\cpr(\varphi,\neg\cpi A,\Gamma)$.
	The number of sets in $\cpr_\varphi$ depends on the number of different subsets of capabilities statements of $cl(\varphi)$. The number of different subsets of $cl(\varphi)$ is at most $2^{\sharp cl(\varphi)}$.
	Since $\sharp cl(\varphi)$ is in $\mc{O}(|\varphi|^2)$, $\sharp \cpr_\varphi$ is in $2^{\mc{O}(|\varphi|^2)}$.
\end{proof}

\begin{lem}	\label{lem: subset of cl}
	For any node $v$ of a tableau $G$ for a formula $\varphi$, 
	there is a set $\Delta$ of formulas in $\cpr_\varphi \cup \{\emptyset\}$ such that
	\begin{enumerate*}[i)]
		\item if $v$ is a partial node, then $\Phi_v \subseteq cl(\varphi) \cup \Delta$, and
		\item if $v$ is a state node, then $\actf_v \subseteq cl(\varphi) \cup \Delta$ and $\rdf_v \subseteq cl(\varphi) \cup \bigcup\cpr_\varphi$.
	\end{enumerate*}
\end{lem}
\begin{proof}
	We proceed by induction on the inductive definition of $G$.
	Obviously, the only formula of the root (that is the input formula of \tref{isSat}) is in the corresponding closure. 
	In the sequel, we consider a node $v_c$ different from the root and we examine all the cases that lead to its definition. Each node of the tableau, besides the root, has a parent. Thus, we also consider the parent of $v_c$, denoted by $v$, in which one of the procedures \tref{csp} and \tref{css} has been applied and has led to the definition of $v_c$ as a child of $v$. 
	By the induction hypothesis, we know that the lemma holds for $v$ and we show that it also holds for $v_c$.
	
	First, we assume that $v$ is a partial node and that $v_c$ is the result of the application of the static rule on $v$ by the procedure \tref{asr}. By the induction hypothesis, we know that $\Phi_{v}$ is a subset of $cl(\varphi) \cup \Delta$, where $\Delta$ is either the empty set or one of the sets of $\cpr_\varphi$. For the definition of $v_c$, \tref{asr} adds to $\Phi_v$ the formulas of one of the reduction sets of an active $\alpha$/$\beta$ formula $\xi$ of $v$. 
	If $\xi$ is in $cl(\varphi)$, then by Definition \ref{def: closure} we know that all of its reductions sets are subsets of $cl(\varphi)$. Thus, we also have that $\Phi_{v_c}$ is a subset of $cl(\varphi)\cup\Delta$.
	Next, we assume that $\Delta\neq\emptyset$ is one of the sets in $\cpr_\varphi$ and that $\xi$ is in it. By Lemma \ref{lem: cpr properties}, we reach the same conclusion as in the previous case.
	
	Next, we assume that $v_c$ is the result of the application of one of the non static rules on a state $v$ by the procedure \tref{ansr}. By the induction hypothesis, we know that $\actf_v \subseteq cl(\varphi) \cup \Delta$, where $\Delta$ is the empty set or one of the sets of $\cpr_\varphi$.	
	If $v_c$ is the result of the application of the transitional rule (see lines \ref{ansr: trans1B}-\ref{ansr: trans2E} of \tref{ansr}) on $v$,
	then there is a formula $\neg\forall A.\chi$ in $\actf_v$ with $A\in\atp_\Omega$ which is considered as responsible for the definition of $v_c$. Since $\Phi_v \subseteq cl(\varphi) \cup \Delta$, we distinguish two cases for $\neg\forall A.\chi$. 
	If $\neg\forall A.\chi$ is in $cl(\varphi)$, then, by Definition \ref{def: closure}, $\neg\chi$ is also in $cl(\varphi)$. Additionally, by Lemma \ref{lem: cpr properties}, we know that all the active formulas of $v$ of the form $\forall A.\vartheta$ with $A\in\atp_\Omega$ belong to $cl(\varphi)$. 
	Thus, $\Phi_{v_c}$ is a subset of $cl(\varphi)$.
	If $\neg\forall A.\chi$ is in $\Delta$, then, by Lemma \ref{lem: cpr properties}, we have that $\Phi_{v_c} \subseteq cl(\varphi) \cup \Delta$.
	
	Finally, we consider that $v_c$ is the result of the application of the capability rule on $v$. So, there is a capability statement $\neg\cpi A$ in $\actf_v$ with $A\in \wt{\Sigma}$. Moreover, let $\Gamma$ be the largest subset of $\actf_v$ of non-negated capabilities statements of the form $\cpi(\chi_1\Ra\chi_2)$. By Lemma \ref{lem: cpr properties}, we know that $\neg\cpi A$ and that all the formulas of $\Gamma$ belong to $cl(\varphi)$. 
	In the following, we show that $\Phi_{v_c}$ is a subset of $cl(\varphi) \cup \cpr(\varphi,\neg\cpi A,\Gamma)$.
	By Definition \ref{def: closure}, we know that for each capability statement $\cpi(\chi_1\Ra\chi_2)$ in $cl(\varphi)$, $\chi_1$ is in $cl(\varphi)$. 
	Furthermore, by Definition \ref{def: cpr}, we know that if $\Gamma$ is the set $\{ \cpi(\chi_1\Ra\psi_1), \dotsc, \cpi(\chi_k\Ra\psi_k) \}$, then the formula $\neg\forall A.\forall(\neg\psi_1).\dotsb.\forall(\neg\psi_k).\false$ is in $\cpr(\varphi,\neg\cpi A,\Gamma)$.
	Therefore, by the way that \tref{ansr} defines $v_c$ (see lines \ref{ansr: capB}-\ref{ansr: capE}), we can conclude that $\Phi_{v_c}\subseteq cl(\varphi) \cup \cpr(\varphi,\neg\cpi A,\Gamma)$.
	
	Besides the previous cases, the algorithm might alter the set of formulas of a pre-existing state $v$. We assume that the algorithm has defined a state $v'$ which is never used, as there is a pre-existing state $v$ with the same active set. According to line \ref{ct: extrdf} of \tref{ct}, the formulas of $\rdf_{v'}$ are added to $\rdf_v$. Since $\actf_v$ does not change, we deal only with its reduced set.
	By the induction hypothesis, we know that the lemma holds for $v$ and $v'$. Since $\rdf_v$ and $\rdf_{v'}$ are subsets of $cl(\varphi) \cup \bigcup\cpr_\varphi$, it is immediate that $\rdf_v \cup \rdf_{v'}$ is also a subset of $cl(\varphi) \cup \bigcup\cpr_\varphi$.
\end{proof}

\begin{lem}	\label{lem: card Phi of a node}
	For any node $v$ of a tableau $G$ for a formula $\varphi$, 
	if $v$ is a partial node, then the cardinality of $\Phi_v$ is in $\mc{O}(|\varphi|^2)$, and 
	if $v$ is a state node, then the cardinality of $\actf_v$ is in $\mc{O}(|\varphi|^2)$ and the cardinality of $\rdf_v$ in $2^{\mc{O}(|\varphi|^2)}$.
\end{lem}
\begin{proof}
	The proof is an immediate consequence of Lemmas \ref{lem: closure properties} and \ref{lem: subset of cl}.
\end{proof}

\begin{lem}	 		\label{lem: reduction sets exp}
	If $\xi$ is an $\alpha$/$\beta$ eventuality of the closure $cl(\varphi)$ of a formula $\varphi$, then its reduction degree is in $2^{\mc{O}(|\varphi|^2)}$.
\end{lem}
\begin{proof}
	According to Definition \ref{def: reduction sets}, the number of the reduction sets of an $\alpha$/$\beta$ eventuality $\xi$ is equal to $\sharp\fd(\xi)$. Now, according to Definition \ref{def: decomposition}, the cardinality of $\fd(\xi)$ depends on the set $\mc{D}(\xi)$. In the following, we examine the cardinalities of these sets.
	The decomposition set $\mc{D}(\xi)$ can be represented as a tree of triples that satisfies the following conditions:
	\begin{enumerate*}[i)]
		\item its root is the triple $(\emptyset,\emptyset, \xi)$, 
		\item each pair of parent-child nodes forms an instance of one of the decomposition rules of Definition \ref{def: decomposition} and
		\item the leaves of the tree are triples $(\mc{P},\mc{T},\vartheta)$ such that $\vartheta$ is in $\mc{P}$ or $\vartheta$ is a non-eventuality or $\vartheta$ is some formula $\neg\forall A.\psi$ with $A\in\atp_\Omega$.
	\end{enumerate*}
	
	According to Definition \ref{def: decomposition}, there are at most two decomposition rules for the same triple. Thus, we can conclude that each node of the tree of $\mc{D}(\xi)$ has at most two children. For any triple of the tree, there is a sequence of triples which forms the path from the root towards it. In particular, for some triple $tr$ of the tree, there is a sequence $(\mc{P}_1,\mc{T}_1,\vartheta_1), \dotsc, (\mc{P}_m,\mc{T}_m,\vartheta_m)$ of triples with $m\geq1$ such that
	$(\mc{P}_1,\mc{T}_1,\vartheta_1) = (\emptyset,\emptyset,\xi)$ and
	$(\mc{P}_m,\mc{T}_m,\vartheta_m)=tr$ and
	for $i=1,\dotsc,m-1$, $\dfrac{(\mc{P}_i,\mc{T}_i,\vartheta_i)}  {(\mc{P}_{i+1},\mc{T}_{i+1},\vartheta_{i+1})}$ is an instance of one of the decomposition rules of Definition \ref{def: decomposition} such that $\vartheta_i\notin\mc{P}_i$ is an $\alpha$/$\beta$ formula and $\mc{P}_{i+1}=\mc{P}_i\cup\{\vartheta_i\}$.
	It is clear that for $i=1,\dotsc,m-1$, $\mc{P}_i$ is the set $\{ \vartheta_1, \dotsc, \vartheta_{i-1} \}$ of $\alpha$/$\beta$ formulas. Thus, $\mc{P}_m = \{\vartheta_1,\dotsc,\vartheta_{m-1}\}$ and the cardinality $\sharp\mc{P}_m$ is $m-1$.
	So, by the properties of the Fischer-Ladner closure of a formula, see \cite{Hartonas1}, and by the way that the decomposition rules have been defined, we can conclude that for any triple $(\mc{P},\mc{T},\vartheta)$, $\sharp\mc{P}$ is in $\mc{O}(|\varphi|^2)$. 	
	Since the number of the principal formulas that can appear is bounded, eventually we reach a leaf triple $(\mc{P},\mc{T},\vartheta)$ for which one of the following properties holds: 
	\begin{enumerate*}[i)]
		\item $\vartheta$ is a non-eventuality,
		\item $\vartheta$ is a formula of the form $\neg\forall A.\chi$ with $A\in\atp_\Omega$,
		\item $\vartheta$ is in $\mc{P}$.
	\end{enumerate*}
	Therefore, the height of the tree is also in $\mc{O}(|\varphi|^2)$. Thus, the number of the triples of the tree of $\mc{D}(\xi)$ is in $2^{\mc{O}(|\varphi|^2)}$ and the same also holds for the number of the leaf triples. Notice that the finalized decomposition set $\fd(\xi)$ is a subset of the set of the pairs $(\mc{T},\vartheta)$ for which there is a leaf triple $(\mc{P},\mc{T},\vartheta)$ in the tree of $\mc{D}(\xi)$ (see Definition \ref{def: decomposition}). Therefore, we have to admit that $\sharp\fd(\xi)$ is also in $2^{\mc{O}(|\varphi|^2)}$.
\end{proof}

\begin{lem} 		\label{lem: number of nodes}
	The number of the nodes of a tableau $G$ for a formula $\varphi_0$ is in $2^{\mc{O}(n^2)}$, where $n$ is the size of $\varphi_0$.
\end{lem}
\begin{proof}
	By the definition of the algorithm (see lines \ref{ct: caching}-\ref{ct: extrdf} of \tref{ct}), we know that there are no different nodes with similar labels in $G$. More specifically, according to Definition \ref{def: similar labels}, there are no partial nodes with the same active set and with the same reduced set and there are no states with the same active set.
	By Definition \ref{def: label}, the active set of a state cannot have $\alpha$/$\beta$ formulas and the active set of a partial node has at least one $\alpha$/$\beta$ formula. Thus, an active set of a state cannot be the active set of a partial node and vice versa.
	
	By Lemma \ref{lem: subset of cl}, we know that the active set of a node is the union of a subset of $cl(\varphi_0)$ and of a set in $\cpr_{\varphi_0}$. 
	By Lemma \ref{lem: closure properties}, we know that the number of the different subsets of $cl(\varphi_0)$ is in $2^{\mc{O}(n^2)}$. By the same lemma, we also know that the cardinality of $\cpr_{\varphi_0}$ is also in $2^{\mc{O}(n^2)}$. Therefore, the number of the different sets of formulas which are defined as the union of a subset of $cl(\varphi_0)$ and of a set in $\cpr_{\varphi_0}$ is in $2^{\mc{O}(n^2)} \cdot 2^{\mc{O}(n^2)} \in 2^{\mc{O}(n^2)}$.
	Thus, the number of the different sets of formulas which can be active sets of nodes is in $2^{\mc{O}(n^2)}$. The same also holds for the number of the different reduced sets for the partial nodes.
	Therefore, since by Definition \ref{def: label}, the label of a node essentially consists of two disjoint sets, namely an active set and a reduced set, the number of the different nodes in $G$ is in $2^{\mc{O}(n^2)} \cdot 2^{\mc{O}(n^2)} \in 2^{\mc{O}(n^2)}$.
\end{proof}

\begin{thm}	[Complexity] 	\label{thm: complexity} 
	The procedure \tref{isSat} with input a formula $\varphi_0$ runs in exponential time in the size $n$ of $\varphi_0$.
\end{thm}
\begin{proof}
	The procedure \tref{isSat} constructs a tableau $G=(V,E)$ for the formula $\varphi_0$. By Lemma \ref{lem: number of nodes}, we know that the number of the nodes of $G$ is in $2^{\mc{O}(n^2)}$. In the following, we examine the computational cost of the expansion of a node and of the calculation of its status and of its dependency set, including also the definition of the fulfillment relation for its active eventualities. We additionally consider the cost of the update process which is triggered after the definition of the status of a node. It is enough to show that each one of the previous tasks requires at most exponential time.
	
	The expansion of a node is a responsibility of the procedures \tref{asr} and \tref{ansr}.
	By Lemma \ref{lem: card Phi of a node}, we know that the cardinality of the active set of formulas of a node is quadratically bounded. According to Definition \ref{def: reduction sets}, the reduction degree of an $\alpha$/$\beta$ non-eventuality is at most two, whereas, in the case of an $\alpha$/$\beta$ eventuality, by Lemma \ref{lem: reduction sets exp}, it is exponential in $n$. Thus, the procedure \tref{asr} works in exponential time. 
	Next, we argue that \tref{ansr} expands a state node in $\mc{O}(n^2)$. For each formula $\neg\forall A.\chi$ with $A\in\atp_\Omega$ and for each formula $\neg\cpi A$ with $A\in\wt{\Sigma}$, \tref{ansr} defines a child. Since by Lemma \ref{lem: card Phi of a node} the cardinality of the active set of formulas of a node is in $\mc{O}(n^2)$, the number of children of a state is also in $\mc{O}(n^2)$.
	
	Now, we examine the procedures \tref{csp} and \tref{css}. 
	Recall that we have already argued that the number of children of a partial node is exponential and that of a state is quadratically bounded. Moreover, the cardinality of the set of formulas of a node is quadratically bounded.
	The main computational cost of both procedures lies in the fulfillment relation $\rat$, as they examine whether an active eventuality of a node is $\rat$-fulfilled. 
	By Lemma \ref{lem: number of nodes}, we know that the number of the nodes of $G$ is in $2^{\mc{O}(n^2)}$. Additionally, by Lemma \ref{lem: card Phi of a node}, we know that the number of the active eventualities of a node is in $\mc{O}(n^2)$. Since the $\rat$-related pairs are pairs of nodes with their own active eventualities (see also Lemma \ref{lem: rat}), we have that the number of such pairs is in $\mc{O}(n^2) \cdot 2^{\mc{O}(n^2)} \in 2^{\mc{O}(n^2)}$. 
	Moreover, by Lemma \ref{lem: rat}, we know that if $(v,\varphi)\rat(v',\varphi')$, then $v'$ is a child of $v$. As argued in this proof, the number of children of a node is exponential in $n$ and as a result, the number of related pairs  in $\rat$ is also exponential in $n$.
	To examine if an active eventuality is $\rat$-fulfilled or is $\rat$-dependent, by Definition \ref{def: fulfilled dependent}, we need to examine all the  reachable pairs by visiting each pair only once. In the worst case, the procedure might traverse all the involved pairs.
	Therefore, the determination of whether an active eventuality is $\rat$-fulfilled requires exponential time.
	
	After the expansion of a node $v$ with its children and after the definition of its status and of its dependency set, the procedure \tref{ct} updates the nodes which have $v$ in their dependency sets by using the procedure \tref{upd} (see line \ref{ct: update}). The evaluation of the boolean conditions in the two while-loops of \tref{upd} requires exponential time, as by Lemma \ref{lem: number of nodes}, we know that the number of the nodes of $G$ is in $2^{\mc{O}(n^2)}$. If the evaluation is true, then the necessary nodes are updated. 
	Let us assume that $v'$ is a node such that $v$ is in $D_{v'}$. First, we assume that it is updated according to lines \ref{upd: while_Unsat}-\ref{upd: propag_Unsat} and as a result, $D_{v'}$ is defined as the empty set. In the second case, we assume that it is updated to the set $(D_{v'}\setminus \{v\})  \cup  D_v$ according to lines \ref{upd: while_Sat_TempSat}-\ref{upd: TempSat}. By Lemma \ref{lem: dependency set}, we know that all the nodes of $D_v$ are forward ancestors of $v$ and thus, $v$ cannot be in $D_v$. 
	In both cases, since $v$ cannot anymore be in $D_{v'}$, we conclude that $v'$ cannot be updated more than once due to $v$.
	Since the number of the nodes of $G$ is in $2^{\mc{O}(n^2)}$, the number of the executions of the while-loops is also in $2^{\mc{O}(n^2)}$.
	The same also holds for the procedure \tref{propag} (which is used within the while-loops), as it updates only nodes whose dependency sets contain the necessary node.
\end{proof}

\section{Conclusions and Further Work} 	\label{sec: conclusions}

In this paper, we addressed the problem of the complexity of the satisfiability of \tpdl by presenting a deterministic tableau-based algorithm which handles appropriately the new operators of \tpdl and which was shown to be sound, complete and exponential. 

The precondition-effect construct is handled through the appropriate $\beta$ rules of Table \ref{tab: alpha/beta} and of course, through the static rule which is applied by the procedure \tref{asr}. The specific operator gives us the ability to express the universal definable relation through the term $\Omega=\true\!\Ra\!\true$ (see Lemma \ref{lem: Omega}). From the moment that for any process $\varphi\Ra\psi$, we have that its interpretation is equal to $\lonto{\neg\varphi\Omega} \cup \lonto{\Omega\psi}$ (recall Corollary \ref{cor: =>2}) things are considerably simplified. The algorithm uses the corresponding tableau rules and the problem of handling the precondition-effect contruct is reduced to the problem of handling the simpler term $\Omega$.

The capabilities statements required special attention due to their interpretation. The expected $\alpha$/$\beta$ rules (recall Table \ref{tab: alpha/beta}) whose application leads to capabilities statements such that their processes cannot be decomposed further were not enough. A negated capability statement $\neg\cpi A$ such that $A\in\wt{\Sigma}$ along with a non-negated statement $\cpi B$ such that $B$ is also in $\wt{\Sigma}$ may cause problems. The algorithm should somehow make sure that the interpretation of $A$ is not a subrelation of that of $B$, as otherwise the agent $\imath$ should also have the capability to execute $A$. In other words, in terms of type processes, it should define a transition of type $A$, but not of type $B$. As we have seen, the capability rule (see Subsection \ref{subsec: tableau calculus}) takes care of such cases. 

The complexity analysis of the algorithm in Subsection \ref{subsec: complexity} reveals that the capability rule has a worth-mentioning difference from the other rules. As mentioned earlier, the proposed capability rule is required to express the existence of a transition with the necessary properties to hold at the two involved states. To achieve this, the capability rule introduces new formulas which do not belong to the closure of the input formula of the algorithm (a variation of the Fischer-Ladner closure \cite{Hartonas1}), as it becomes apparent from Definitions \ref{def: closure} and \ref{def: cpr} and from Lemma \ref{lem: subset of cl}. According to its definition (recall equations \eqref{eq: cap Delta} and \eqref{eq: capability rule}), for each negated capability statement and for each subset of non-negated capabilities statements with precondition-effect processes, an appropriate formula is defined which indicates the necessary transition. As the number of the subsets of the non-negated capabilities statements is at most exponential, the number of these new formulas is also at most exponential. Fortunately, this does not have a significant impact on the complexity of the algorithm and it still runs in exponential time. 

The algorithm has available for possible reuse not only the states, but also the partial nodes of a tableau. Nevertheless, there is a difference between the states and the partial nodes. The reuse of a state depends exclusively on its active set, whereas in the case of a partial node, its reduced set is also considered. Observe that states with the same reduced sets have also the same active sets, but states with the same active sets may have different reduced sets. On the other hand, as far as a partial node is concerned, if the algorithm does not consider its reduced set, a loop may be formed in which an active set never becomes saturated. 

The algorithm restricts or-branching whenever possible. The procedure \tref{asr} does not expand a partial node with its remaining children if the status of one of its previously defined children has been assigned the value \sat. 
Thanks to the dependency set of a node, the algorithm has the ability to define the status of a node accordingly.  
While the status value \sat indicates satisfiability, the value \tempsat indicates possible satisfiability of the node, as it depends on ancestor nodes.
Thus, the algorithm distinguishes between the actual and the possible satisfiability and in the former case, it has the opportunity to discard or-branches.

An issue that we need to consider is the extension of our algorithm for the backwards possibility operator or even for the converse by applying the approach devised in \cite{Gore-CPDL}. 
The main concerns are two. 
First, we need to examine whether our rules for the new operators of \tpdl can be adapted to the setting of the algorithm of \cite{Gore-CPDL}. Since we adopt similar $\alpha$/$\beta$ rules and similar non-static rules (i.e. the transitional and the capability rule) to those of \cpdl, it is expected that problematic situations should be avoided. As we consider operators which ``look backwards'', keep in mind that the capability rule does not require the existence of a transition between the premise and the conclusions of the rule (see Subsection \ref{subsec: tableau calculus}).
Second, we should examine whether our approach for the restriction of or-branching causes any problems. 
We should mention here that similar optimizations have been proposed for \pdl and \cpdl in \cite{Widmann_Thesis} where the lack of diamond formulas with atomic processes in a node indicates satisfiability.

\phantomsection
\addcontentsline{toc}{section}{References}

{\small 
	\newcommand{\etalchar}[1]{$^{#1}$}
	
}

\end{document}